\documentclass[11pt,twoside]{article}
\usepackage{fancyhdr}
\pdfoutput=1
\usepackage[colorlinks,citecolor=blue,urlcolor=blue,linkcolor=blue,bookmarks=false]{hyperref}
\usepackage{amsfonts,epsfig,graphicx}
\usepackage{afterpage}
\usepackage{pgfplots}

\usepackage{tikz}
\usetikzlibrary{arrows.meta,calc,positioning}
\usepackage{nicefrac}
\usepackage{comment}
\usepackage{bbm}
\usepackage{amsmath,amssymb,amsthm} 
\usepackage{fullpage}
\usepackage[T1]{fontenc} 
\usepackage{epsf} 
\usepackage{cancel}

\usepackage{graphics} 
\usepackage{amsfonts,amsmath}
\usepackage[round]{natbib} 
\usepackage{psfrag,xspace}
\usepackage{color,etoolbox}
\usepackage{subcaption} 
\usepackage{listings}
\usepgfplotslibrary{groupplots}
\usepackage[noabbrev,capitalize, nameinlink]{cleveref}

\usepgfplotslibrary{fillbetween}

\newcommand{\E}{\mathbb{E}}
\newcommand{\R}{\mathbb{R}}

\usepackage{pgfplots}
\pgfmathdeclarefunction{poiss}{1}{%
  \pgfmathparse{(#1^x)*exp(-#1)/(x!)}%
}

\DeclareSymbolFont{bbold}{U}{bbold}{m}{n}
\DeclareSymbolFontAlphabet{\mathbbold}{bbold}

\newtheorem{theorem}{Theorem}
\newtheorem{lemma}{Lemma}
\newtheorem{corollary}{Corollary}
\newtheorem{proposition}{Proposition}

\theoremstyle{definition}

\theoremstyle{remark}
\newtheorem{assumption}{Assumption}

\newtheorem{remark}{Remark}

\definecolor{dkgreen}{rgb}{0,0.6,0}
\definecolor{gray}{rgb}{0.5,0.5,0.5}
\definecolor{mauve}{rgb}{0.58,0,0.82}
\definecolor{brightblue}{HTML}{00BFC4}

\newcommand\blfootnote[1]{%
  \begingroup
  \renewcommand\thefootnote{}%
  \footnote{#1}%
  \addtocounter{footnote}{-1}%
  \endgroup
}

\pgfplotsset{compat=1.18}
\begin{document}

\def\spacingset#1{\renewcommand{\baselinestretch}%
{#1}\small\normalsize} \spacingset{1}

\raggedbottom
\allowdisplaybreaks[1]


  \title{\vspace*{-.4in} {Distributional Discontinuity Design}}
   \author{\\ $\text{Kyle Schindl}^{\dagger}$, $\text{Larry Wasserman}^{\ddag}$ \\ \\
    $^{\dag}$Department of Statistics \\
    Iowa State University \\
    \texttt{kschindl@iastate.edu} \\ \\ 
    $^\ddag$Department of Statistics \& Data Science \\
    Machine Learning Department \\
    Carnegie Mellon University \\
    \texttt{larry@stat.cmu.edu}
\date{}
    }

  \maketitle
  \blfootnote{Accompanying \texttt{R} code is available via \href{https://github.com/kyleschindl/discontinuity-designs}{github.com/kyleschindl/discontinuity-designs}}
  \thispagestyle{empty}

\begin{abstract}
{\em Regression discontinuity and kink designs are typically analyzed through mean effects, even when treatment changes the shape of the entire outcome distribution. To address this, we introduce distributional discontinuity designs, a framework for estimating causal effects for a scalar outcome at the boundary of a discontinuity in treatment assignment. Our estimand is the Wasserstein distance between limiting conditional outcome distributions; a single scale-interpretable measure of distribution shift. We show that this weakly bounds the average treatment effect, where equality holds if and only if the treatment effect is purely additive; thus, departure from equality measures effect heterogeneity. To further encode effect heterogeneity we show that the Wasserstein distance admits an orthogonal decomposition into squared differences in $L$-moments, thereby quantifying the contribution from location, scale, skewness, and higher-order shape components to the overall distributional distance. Next, we extend this framework to distributional kink designs by evaluating the Wasserstein derivative at a policy kink; this describes the flow of probability mass through the kink. In the case of fuzzy kink designs, we derive new identification results. Finally, we apply our methods on real data by re-analyzing two natural experiments to compare our distributional effects to traditional causal estimands.}
\end{abstract}


\bigskip 

\section{Introduction}

First introduced by \cite{thistlethwaite1960regression} and formalized by \cite{hahn2001}, regression discontinuity design is a quasi-experimental design method that exploits discontinuities in treatment assignment to identify causal effects. The key idea is that observational units arbitrarily close to either side of the treatment discontinuity can be thought of as similar in all respects except for treatment status. Thus, in this neighborhood of the discontinuity, treatment assignment can be considered ``as good as random,'' and it is therefore reasonable to assume that exchangeability holds. Over the years, a very rich and deep literature for regression discontinuity design methods has been developed, with contributions too broad to enumerate. Modern regression discontinuity design tends to 
focus on local polynomial estimation and bandwidth selection \citep{imbens_kalyanaraman_2012}, robust bias-corrected inference \citep{calonico_2014, calonico_2019}, and a suite of diagnostic tools such as density-manipulation tests \citep{MCCRARY2008698, Cattaneo02072020}. Readers interested in the history of regression discontinuity design should refer to \cite{COOK2008636}, and to \cite{leelemieux2010}, \cite{IMBENS2008615}, and \cite{Cattaneo2022} for in-depth reviews. Notably, most traditional methods primarily focus on differences in mean effects above and below the cutoff; there is a much smaller literature considering \textit{distributional} causal effects.

Focusing exclusively on mean effects often limits both the usefulness and generalizability of an analysis, because averages can mask treatment heterogeneity across the outcome distribution. It is easy to imagine a treatment that leaves the average unchanged, but has asymmetric effects on the lower and upper tails of the outcome distribution. Consequently, over the past several years there has been an increased focus on developing causal effects that consider how the entire outcome distribution changes with respect to a treatment. In the regression discontinuity design setting \cite{FRANDSEN2012382} introduced pointwise quantile treatment effects for both sharp and fuzzy designs. Since then, there have been several extensions and generalizations of their work, to include uniform confidence bands \citep{QU20151}, bias-corrected estimators \citep{Qu02102019, CHIANG2019589}, and quantile effects in regression kink designs \citep{CHIANG2019405, Chen_Chiang_Sasaki_2020}. More recently, \cite{Jin_Zhang_Zhang_Zhou_2025} considered quantile effects under a local rank similarity condition and \cite{vandijcke2025regressiondiscontinuitydesigndistributionvalued} considered local average quantile treatment effects under distribution-valued outcomes; their framework is conceptually similar, but non-overlapping with our own as we consider scalar valued outcomes. While each of these methods is interesting and useful, in practice they can be difficult to implement and interpret. Practitioners are often not interested in specific quantile effects and considering many quantiles can be hard to summarize and communicate. Furthermore, because quantile treatment effects describe the marginal outcome distributions, in order to interpret effects at the individual level a strong rank-invariance assumption must be made. We address these problems by defining causal effects in terms of the distributional distance between conditional counterfactual distributions, which yields a single transparent measure of the overall distribution shift.

In this paper we introduce \textit{distributional discontinuity designs}, a framework for studying distributional causal effects above and below some treatment discontinuity. Specifically, we define our causal effect to be the Wasserstein distance between the limiting 
conditional distribution
of the counterfactual $Y(a) \mid  X = x$ above and below the treatment discontinuity. This provides a clean, one number summary of the entire distance between treatment groups, thereby encoding the total magnitude of the treatment effect and establishing a relative scale for all treatment effects. Using the Wasserstein distance as our causal effect yields a number of nice properties. First, we show that it is weakly greater than the average treatment effect at the cutoff, where equality holds if and only if the treatment effect is purely additive; this immediately provides a useful reference point to establish the amount of treatment heterogeneity. Second, we show that the Wasserstein distance can be decomposed into the effect on individual $L$-moments \citep{hosking1990}. This allows us to define a ``distributional $R^2$,'' i.e. the amount of the distributional distance explained by each $L$-moment, thereby providing a novel way of summarizing treatment heterogeneity by its effect on location, scale, skewness, etc. Third, since the Wasserstein distance describes the total magnitude of the treatment effect, we can use it to define the degree to which one quantile function stochastically dominates the other. In our analysis, we consider both sharp and fuzzy treatment assignments. Additionally, we extend the distributional discontinuity design framework to regression kink designs by defining our causal effect to be the Wasserstein derivative at the policy kink; this describes the \textit{flow} of probability mass at the cutoff, and neatly generalizes traditional kink designs. Notably, we also extend the work of \cite{wang2025unifiedframeworkidentificationinference} to establish identification of fuzzy local treatment effects at a policy kink.

Broadly speaking, our analysis fits into a growing literature of papers that apply optimal transport methods to causal inference problems in order to compare entire outcome distributions, rather than just averages. For example, \cite{gunsilius_synthetic_controls2023} develops distributional synthetic controls that reconstruct a treated unit’s distribution from controls. In difference-in-differences, optimal transport methods align pre and post treatment outcome distributions across groups instead of relying on mean-level parallel trends; see \cite{TorousGunsiliusRigollet+2024} for a nonlinear difference-in-differences and \cite{zhou2025geodesicdifferenceindifferences} for a geodesic variant. More generally, \cite{kurisu2025geodesiccausalinference} and \cite{schindl2025causalgeodesycounterfactualestimation} consider causal change as a movement along paths in the space of probability distributions. Interested readers can see \cite{gunsilius2025primeroptimaltransportcausal} for an extended discussion and review of the literature.

The remainder of the paper is organized as follows: In \cref{setup_section} we define all relevant notation and definitions. In \cref{ddd_section} we formally define the distributional discontinuity framework, first in the sharp treatment assignment setting. Within this section, we consider identification of the Wasserstein effect (\cref{identification_section}), interpretation of these effects and their relationship to traditional mean and quantile based effects (\cref{interpretation_section}), estimation and the limiting distribution of the Wasserstein effect (\cref{estimation_inference_section}), inference for the Wasserstein effect (\cref{inference_section}), and finally we extend our results to the fuzzy treatment discontinuity setting (\cref{fuzzy_rdd_section}). In \cref{kink_designs_section} we extend our framework to kinked distributional designs, by defining a novel causal effect in terms of the Wasserstein derivative at the policy kink and in \cref{fuzzy_kink_design_section} we extend the work of \cite{wang2025unifiedframeworkidentificationinference} to establish causal identification in fuzzy kink designs. In \cref{application_section} we apply our method to real data sets by re-analyzing several natural experiments and directly comparing the Wasserstein effect to the average treatment effect at the cutoff. Finally, in \cref{discussion_section} we provide a discussion and conclusion of distributional discontinuity designs, including limitations and directions for future work.

\section{Setup \& Notation} \label{setup_section}

Suppose we observe $Z_1, \ldots, Z_n \overset{iid}{\sim} P$ where $Z_i = (X_i, A_i, Y_i)$ where $X_i \in \mathbb{R}$ is the ``running variable,'' $A_i \in \{0, 1\}$ is the treatment assignment, and $Y_i \in \mathbb{R}$ is the observed outcome. Note that $Y$ is a scalar, unlike in \cite{vandijcke2025regressiondiscontinuitydesigndistributionvalued}, which considers a distribution-valued outcome. To begin, we assume that treatment is assigned such that
\begin{align*}
    A_i = \begin{cases}
        1 & \text{if $X_i \geq x_0$} \\
        0 & \text{if $X_i < x_0$}
    \end{cases}
\end{align*}
at cutoff $X = x_0$. In \cref{fuzzy_rdd_section} and beyond we relax this assignment rule to the fuzzy setting. We assume that
$Y$ 
has a continuous distribution.
We are interested in the conditional distribution of $Y(a)  \mid  X = x$ where $Y(a)$ are the potential outcomes under treatment assignment $A = a$. Furthermore, note that our framework does allow for the inclusion of some vector of covariates; however, since this is not required for identification and is notationally cumbersome to include, we omit such terms from our analysis. A brief discussion of conditioning on additional covariates can be found in \cref{estimation_inference_section}.

Throughout the paper, for some function $f(\cdot)$ we use the notation $\text{lim}_{x \uparrow x_0} f(x) = f(x^-_0)$ to denote the left-hand limit (i.e. $x < x_0$ and $x \to x_0$) and similarly $\text{lim}_{x \downarrow x_0} f(x) = f(x^+_0)$ to denote the right-hand limit. We say $\mathcal{P}_2(\cdot)$ is the set of all probability measures with finite second moments. We use $F_{Y \mid X}(y \! \mid \! x)$ to denote the cumulative distribution function of $Y \mid  X = x$ with the associated quantile function $Q_x(u) = F^{-1}_x(u) = \text{inf} \, \{ y :  F_{Y \mid X} (y \! \mid \! x) \geq u\}$. When we consider the limiting quantiles, we drop the notation on $x$ and simply say that $Q_{1}(u) = \text{inf} \, \{ y : \text{lim}_{x \downarrow x_0} F_{Y \mid X} (y \! \mid \! x) \geq u\}$ and $Q_{0}(u) = \text{inf} \, \{ y : \text{lim}_{x \uparrow x_0} F_{Y \mid X} (y  \!\mid \! x) \geq u\}$, where the zero and one notation is used to denote taking the limit from above or below the cutoff.

\section{Distributional Discontinuity Design} \label{ddd_section}

In this section, we introduce distributional causal effects that compare the entire outcome distribution below and above some treatment discontinuity. 
Let $P_{a \mid x}$ denote the conditional counterfactual distribution of $Y(a) \mid X = x$ and suppose that treatment is assigned in a discontinuous way, where $A = \mathbb{I}(X \geq x_0)$ for some running variable $X$ and cutoff $x_0$. Then, we define our causal estimand to be the 2-Wasserstein distance between the counterfactual treatment distributions at the treatment discontinuity $X = x_0$, i.e. 
\begin{align*}
    \Psi = W_2(P_{1 \mid x_0} , P_{0 \mid x_0}). 
\end{align*}
The 2-Wasserstein distance between any two probability distributions $P$ and $Q$ is defined as
\begin{align*}
    W^2_2(P, Q) = \underset{\gamma \in \Gamma(P, Q)}{\text{inf}} \int ||x - y||^2_2 d\gamma(x, y)
\end{align*}
where $\Gamma(P, Q)$ is the set of all couplings of $P$ and $Q$, i.e. the set of all joint distributions $\gamma$ that preserve the marginals of $P$ and $Q$ \citep{villani2009optimal}. Roughly speaking, $\gamma$ describes a way of pairing points from $P$ and $Q$ such that the total quadratic transport cost between distributions is minimized under the best possible pairing, as visualized in \cref{fig:ot-map-flow}. More intuitively, this describes the minimal transportation cost of transforming or ``morphing'' $P$ into $Q$. 
If $P$ has a density then
$W^2_2(P,Q) = \inf \E[||T(X)-X||^2_2]$
where the infimum is over all maps $T$ such that $T(X)\sim Q$. The map $T$ is called the optimal transport map.

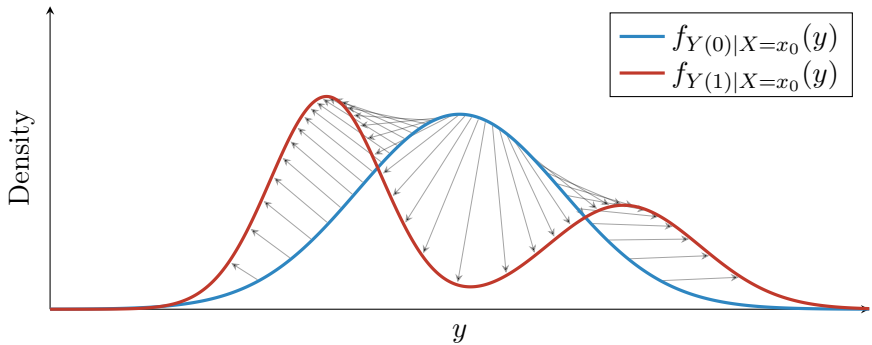
\begin{figure}[h]
\centering
\begin{tikzpicture}

\definecolor{plotlyBlue}{HTML}{2E86C1}
\definecolor{plotlyRed}{HTML}{C0392B}
\tikzset{
  dens0/.style={very thick, draw=plotlyBlue},
  dens1/.style={very thick, draw=plotlyRed},
  Ray/.style={
    line width=0.35pt,
    draw=black!90,
    draw opacity=0.40, 
    -{Stealth[length=3.0pt,width=2.8pt]},
    line cap=round,
    line join=round,
    shorten <=0.8pt,
    shorten >=1.0pt
  }
}

\begin{axis}[
  width=0.8\linewidth,
  height=0.36\linewidth,
  xmin=-4, xmax=4,
  ymin=0, ymax=0.62,
  axis lines=left,
  xlabel={$y$},
  ylabel={Density},
  ytick=\empty,
  xtick=\empty,
  legend style={draw=black, fill=none, anchor=north east, legend columns=1},
  clip=false
]

\draw[Ray] (axis cs:-1.9600,0.0584) -- (axis cs:-2.2524,0.0972);
\draw[Ray] (axis cs:-1.6449,0.1031) -- (axis cs:-2.0606,0.1673);
\draw[Ray] (axis cs:-1.4395,0.1416) -- (axis cs:-1.9327,0.2246);
\draw[Ray] (axis cs:-1.2816,0.1755) -- (axis cs:-1.8321,0.2726);
\draw[Ray] (axis cs:-1.1503,0.2059) -- (axis cs:-1.7467,0.3129);
\draw[Ray] (axis cs:-1.0364,0.2332) -- (axis cs:-1.6710,0.3467);
\draw[Ray] (axis cs:-0.9346,0.2578) -- (axis cs:-1.6017,0.3744);
\draw[Ray] (axis cs:-0.8416,0.2800) -- (axis cs:-1.5369,0.3967);
\draw[Ray] (axis cs:-0.7554,0.2999) -- (axis cs:-1.4753,0.4137);
\draw[Ray] (axis cs:-0.6745,0.3178) -- (axis cs:-1.4158,0.4257);
\draw[Ray] (axis cs:-0.5978,0.3337) -- (axis cs:-1.3576,0.4329);
\draw[Ray] (axis cs:-0.5244,0.3477) -- (axis cs:-1.3001,0.4353);
\draw[Ray] (axis cs:-0.4538,0.3599) -- (axis cs:-1.2425,0.4330);
\draw[Ray] (axis cs:-0.3853,0.3704) -- (axis cs:-1.1844,0.4259);
\draw[Ray] (axis cs:-0.3186,0.3792) -- (axis cs:-1.1249,0.4140);
\draw[Ray] (axis cs:-0.2533,0.3863) -- (axis cs:-1.0633,0.3971);
\draw[Ray] (axis cs:-0.1891,0.3919) -- (axis cs:-0.9986,0.3751);
\draw[Ray] (axis cs:-0.1257,0.3958) -- (axis cs:-0.9295,0.3476);
\draw[Ray] (axis cs:-0.0627,0.3982) -- (axis cs:-0.8540,0.3143);
\draw[Ray] (axis cs:0.0000,0.3989) -- (axis cs:-0.7691,0.2746);
\draw[Ray] (axis cs:0.0627,0.3982) -- (axis cs:-0.6695,0.2278);
\draw[Ray] (axis cs:0.1257,0.3958) -- (axis cs:-0.5444,0.1729);
\draw[Ray] (axis cs:0.1891,0.3919) -- (axis cs:-0.3652,0.1095);
\draw[Ray] (axis cs:0.2533,0.3863) -- (axis cs:-0.0234,0.0499);
\draw[Ray] (axis cs:0.3186,0.3792) -- (axis cs:0.4558,0.0691);
\draw[Ray] (axis cs:0.3853,0.3704) -- (axis cs:0.7378,0.1103);
\draw[Ray] (axis cs:0.4538,0.3599) -- (axis cs:0.9347,0.1437);
\draw[Ray] (axis cs:0.5244,0.3477) -- (axis cs:1.0942,0.1695);
\draw[Ray] (axis cs:0.5978,0.3337) -- (axis cs:1.2334,0.1888);
\draw[Ray] (axis cs:0.6745,0.3178) -- (axis cs:1.3610,0.2022);
\draw[Ray] (axis cs:0.7554,0.2999) -- (axis cs:1.4820,0.2102);
\draw[Ray] (axis cs:0.8416,0.2800) -- (axis cs:1.6000,0.2128);
\draw[Ray] (axis cs:0.9346,0.2578) -- (axis cs:1.7180,0.2102);
\draw[Ray] (axis cs:1.0364,0.2332) -- (axis cs:1.8390,0.2022);
\draw[Ray] (axis cs:1.1503,0.2059) -- (axis cs:1.9666,0.1888);
\draw[Ray] (axis cs:1.2816,0.1755) -- (axis cs:2.1059,0.1695);
\draw[Ray] (axis cs:1.4395,0.1416) -- (axis cs:2.2654,0.1436);
\draw[Ray] (axis cs:1.6449,0.1031) -- (axis cs:2.4628,0.1098);
\draw[Ray] (axis cs:1.9600,0.0584) -- (axis cs:2.7506,0.0656);


\addplot[dens0, domain=-4:4, samples=450]
  {1/sqrt(2*pi)*exp(-(x^2)/2)};
\addlegendentry{$f_{Y(0)\mid X=x_0}(y)$}

\addplot[dens1, domain=-4:4, samples=900]
  {0.6*(1/(0.55*sqrt(2*pi))*exp(-((x+1.3)^2)/(2*(0.55^2))))
  +0.4*(1/(0.75*sqrt(2*pi))*exp(-((x-1.6)^2)/(2*(0.75^2))))};
\addlegendentry{$f_{Y(1)\mid X=x_0}(y)$}

\end{axis}
\end{tikzpicture}
\caption{Optimal transport maps between counterfactual distributions.}
\label{fig:ot-map-flow}
\end{figure}

In the case of a treatment discontinuity, $\Psi$ measures how far probability mass must be moved in order to transform the untreated distribution at the cutoff into the treated distribution. Thus, it measures differences not only in the means, but also higher moment effects such as the variance or skewness. For this reason, the Wasserstein effect can detect and quantify complex, higher order treatment effects that the traditional regression discontinuity design estimand would miss. For example, in \cref{distributional_diff}, we can see that above the treatment discontinuity at $x = 0$ not only does the mean change, but the variance and overall distributional shape does as well. Focusing solely on the difference in means would not adequately describe the full effect of treatment here.

As a simple motivating example, suppose that there is a treatment discontinuity at $x_0=0$
and that
\begin{align*}
    Y(0) \mid X = x &\sim N(0, 1) \hspace{-1.15in} && \text{for all $x$, and}\\ 
    Y(1) \mid X = x &\sim N(0, 2^2 )  \hspace{-1.15in} &&\text{for all $x$}. 
\end{align*}
Then, it is clear that the average treatment effect at the cutoff 
$\E[Y(1) \mid X=x_0] - \E[Y(0) \mid X=x_0]$
is zero, as there is no change in location above and below the cutoff. However, the standard deviation doubles. Researchers who only consider the average treatment effect at the cutoff would conclude there was no treatment effect, but in reality, a doubling of the standard deviation could have large practical implications. Similarly, as discussed in \cite{kim2024causaleffectsbaseddistributional}, treatment effects could easily take a multimodal structure where $Y(0) = 0$ almost surely, but $Y(1) = 1$ or $Y(1) = - 1$ with equal probability. In this setting, the average treatment effect is again zero, but treatment harms half the population and benefits the other half. Fortunately, both of these causal effects can be detected by $\Psi$. For example, in the first setting with two Gaussians, it can be shown that $\Psi = |\sigma_1 - \sigma_2| = 1$, indicating a sharp difference in the outcome distributions. In \cref{interpretation_section}, we provide more guidance on the interpretation of the Wasserstein effect, and its comparison to the average treatment effect at the cutoff. Now that we have defined our effect of interest, we establish the conditions under which it is causally identified.

\begin{figure}[h]
    \centering
\definecolor{plotlyBlue}{HTML}{2E86C1}
\definecolor{plotlyRed}{HTML}{C0392B}

\begin{tikzpicture}
\begin{axis}[
  view={-45}{25},
  width=0.8\linewidth, height=0.68\linewidth,
  y dir=reverse,
  x dir=reverse,
    axis line style={draw=none},
  xtick pos=bottom,               
  ytick pos=left,                 
  ztick pos=left,   
  xtick={-10,-5,0,5},
xticklabels={5,0,-5,-10},
  tick align=outside,    
  tick style={draw=none},
  xlabel={$y$}, ylabel={$x$}, zlabel={Conditional Density},
  xmajorgrids, ymajorgrids, zmajorgrids,
  xminorgrids, yminorgrids, zminorgrids,
  grid style={line width=0.1pt, draw=gray!20, opacity=0.5},
  axis on top
]
 \addplot3[very thin, solid, draw=plotlyRed, mark=none, forget plot]
  table[col sep=comma, header=true, trim cells=true, x index=0, y index=1, z index=2]
  {files/curves/above/s001.csv};
\addplot3[very thin, solid, draw=plotlyRed, mark=none, forget plot]
  table[col sep=comma, header=true, trim cells=true, x index=0, y index=1, z index=2]
  {files/curves/above/s002.csv};
\addplot3[very thin, solid, draw=plotlyRed, mark=none, forget plot]
  table[col sep=comma, header=true, trim cells=true, x index=0, y index=1, z index=2]
  {files/curves/above/s003.csv};
\addplot3[very thin, solid, draw=plotlyRed, mark=none, forget plot]
  table[col sep=comma, header=true, trim cells=true, x index=0, y index=1, z index=2]
  {files/curves/above/s004.csv};
\addplot3[very thin, solid, draw=plotlyRed, mark=none, forget plot]
  table[col sep=comma, header=true, trim cells=true, x index=0, y index=1, z index=2]
  {files/curves/above/s005.csv};
\addplot3[very thin, solid, draw=plotlyRed, mark=none, forget plot]
  table[col sep=comma, header=true, trim cells=true, x index=0, y index=1, z index=2]
  {files/curves/above/s006.csv};
\addplot3[very thin, solid, draw=plotlyRed, mark=none, forget plot]
  table[col sep=comma, header=true, trim cells=true, x index=0, y index=1, z index=2]
  {files/curves/above/s007.csv};
\addplot3[very thin, solid, draw=plotlyRed, mark=none, forget plot]
  table[col sep=comma, header=true, trim cells=true, x index=0, y index=1, z index=2]
  {files/curves/above/s008.csv};
\addplot3[very thin, solid, draw=plotlyRed, mark=none, forget plot]
  table[col sep=comma, header=true, trim cells=true, x index=0, y index=1, z index=2]
  {files/curves/above/s009.csv};
\addplot3[very thin, solid, draw=plotlyRed, mark=none, forget plot]
  table[col sep=comma, header=true, trim cells=true, x index=0, y index=1, z index=2]
  {files/curves/above/s010.csv};
\addplot3[very thin, solid, draw=plotlyRed, mark=none, forget plot]
  table[col sep=comma, header=true, trim cells=true, x index=0, y index=1, z index=2]
  {files/curves/above/s011.csv};
\addplot3[very thin, solid, draw=plotlyRed, mark=none, forget plot]
  table[col sep=comma, header=true, trim cells=true, x index=0, y index=1, z index=2]
  {files/curves/above/s012.csv};
\addplot3[very thin, solid, draw=plotlyRed, mark=none, forget plot]
  table[col sep=comma, header=true, trim cells=true, x index=0, y index=1, z index=2]
  {files/curves/above/s013.csv};
\addplot3[very thin, solid, draw=plotlyRed, mark=none, forget plot]
  table[col sep=comma, header=true, trim cells=true, x index=0, y index=1, z index=2]
  {files/curves/above/s014.csv};
\addplot3[very thin, solid, draw=plotlyRed, mark=none, forget plot]
  table[col sep=comma, header=true, trim cells=true, x index=0, y index=1, z index=2]
  {files/curves/above/s015.csv};
\addplot3[very thin, solid, draw=plotlyRed, mark=none, forget plot]
  table[col sep=comma, header=true, trim cells=true, x index=0, y index=1, z index=2]
  {files/curves/above/s016.csv};
\addplot3[very thin, solid, draw=plotlyRed, mark=none, forget plot]
  table[col sep=comma, header=true, trim cells=true, x index=0, y index=1, z index=2]
  {files/curves/above/s017.csv};
\addplot3[very thin, solid, draw=plotlyRed, mark=none, forget plot]
  table[col sep=comma, header=true, trim cells=true, x index=0, y index=1, z index=2]
  {files/curves/above/s018.csv};
\addplot3[very thin, solid, draw=plotlyRed, mark=none, forget plot]
  table[col sep=comma, header=true, trim cells=true, x index=0, y index=1, z index=2]
  {files/curves/above/s019.csv};
\addplot3[very thin, solid, draw=plotlyRed, mark=none, forget plot]
  table[col sep=comma, header=true, trim cells=true, x index=0, y index=1, z index=2]
  {files/curves/above/s020.csv};

\addplot3[very thin, solid, draw=plotlyBlue, mark=none, forget plot]
  table[col sep=comma, header=true, trim cells=true, x index=0, y index=1, z index=2]
  {files/curves/below/s001.csv};
\addplot3[very thin, solid, draw=plotlyBlue, mark=none, forget plot]
  table[col sep=comma, header=true, trim cells=true, x index=0, y index=1, z index=2]
  {files/curves/below/s002.csv};
\addplot3[very thin, solid, draw=plotlyBlue, mark=none, forget plot]
  table[col sep=comma, header=true, trim cells=true, x index=0, y index=1, z index=2]
  {files/curves/below/s003.csv};
\addplot3[very thin, solid, draw=plotlyBlue, mark=none, forget plot]
  table[col sep=comma, header=true, trim cells=true, x index=0, y index=1, z index=2]
  {files/curves/below/s004.csv};
\addplot3[very thin, solid, draw=plotlyBlue, mark=none, forget plot]
  table[col sep=comma, header=true, trim cells=true, x index=0, y index=1, z index=2]
  {files/curves/below/s005.csv};
\addplot3[very thin, solid, draw=plotlyBlue, mark=none, forget plot]
  table[col sep=comma, header=true, trim cells=true, x index=0, y index=1, z index=2]
  {files/curves/below/s006.csv};
\addplot3[very thin, solid, draw=plotlyBlue, mark=none, forget plot]
  table[col sep=comma, header=true, trim cells=true, x index=0, y index=1, z index=2]
  {files/curves/below/s007.csv};
\addplot3[very thin, solid, draw=plotlyBlue, mark=none, forget plot]
  table[col sep=comma, header=true, trim cells=true, x index=0, y index=1, z index=2]
  {files/curves/below/s008.csv};
\addplot3[very thin, solid, draw=plotlyBlue, mark=none, forget plot]
  table[col sep=comma, header=true, trim cells=true, x index=0, y index=1, z index=2]
  {files/curves/below/s009.csv};
\addplot3[very thin, solid, draw=plotlyBlue, mark=none, forget plot]
  table[col sep=comma, header=true, trim cells=true, x index=0, y index=1, z index=2]
  {files/curves/below/s010.csv};
\addplot3[very thin, solid, draw=plotlyBlue, mark=none, forget plot]
  table[col sep=comma, header=true, trim cells=true, x index=0, y index=1, z index=2]
  {files/curves/below/s011.csv};
\addplot3[very thin, solid, draw=plotlyBlue, mark=none, forget plot]
  table[col sep=comma, header=true, trim cells=true, x index=0, y index=1, z index=2]
  {files/curves/below/s012.csv};
\addplot3[very thin, solid, draw=plotlyBlue, mark=none, forget plot]
  table[col sep=comma, header=true, trim cells=true, x index=0, y index=1, z index=2]
  {files/curves/below/s013.csv};
\addplot3[very thin, solid, draw=plotlyBlue, mark=none, forget plot]
  table[col sep=comma, header=true, trim cells=true, x index=0, y index=1, z index=2]
  {files/curves/below/s014.csv};
\addplot3[very thin, solid, draw=plotlyBlue, mark=none, forget plot]
  table[col sep=comma, header=true, trim cells=true, x index=0, y index=1, z index=2]
  {files/curves/below/s015.csv};
\addplot3[very thin, solid, draw=plotlyBlue, mark=none, forget plot]
  table[col sep=comma, header=true, trim cells=true, x index=0, y index=1, z index=2]
  {files/curves/below/s016.csv};
\addplot3[very thin, solid, draw=plotlyBlue, mark=none, forget plot]
  table[col sep=comma, header=true, trim cells=true, x index=0, y index=1, z index=2]
  {files/curves/below/s017.csv};
\addplot3[very thin, solid, draw=plotlyBlue, mark=none, forget plot]
  table[col sep=comma, header=true, trim cells=true, x index=0, y index=1, z index=2]
  {files/curves/below/s018.csv};
\addplot3[very thin, solid, draw=plotlyBlue, mark=none, forget plot]
  table[col sep=comma, header=true, trim cells=true, x index=0, y index=1, z index=2]
  {files/curves/below/s019.csv};
\addplot3[very thin, solid, draw=plotlyBlue, mark=none, forget plot]
  table[col sep=comma, header=true, trim cells=true, x index=0, y index=1, z index=2]
  {files/curves/below/s020.csv};

\pgfplotsextra{%
  \pgfkeysgetvalue{/pgfplots/xmin}{\xmin}%
  \pgfkeysgetvalue{/pgfplots/xmax}{\xmax}%
  \pgfkeysgetvalue{/pgfplots/ymin}{\ymin}%
  \pgfkeysgetvalue{/pgfplots/ymax}{\ymax}%
  \pgfkeysgetvalue{/pgfplots/zmin}{\zmin}%
  \pgfkeysgetvalue{/pgfplots/zmax}{\zmax}%

  \draw[black, dashed, line width=0.5pt, on layer=axis foreground]
    (axis cs:\xmin,-0.25,\zmin) -- (axis cs:\xmax,-0.25,\zmin);

  \def\xRight{\xmin} 
  \draw[black, dashed, line width=0.5pt, on layer=axis foreground]
    (axis cs:\xRight,-0.25,\zmin) -- (axis cs:\xRight,-0.25,\zmax);
}%
  \end{axis}
\end{tikzpicture}

    \caption{Counterfactual distributions above and below a treatment discontinuity}
    \label{distributional_diff}
\end{figure}
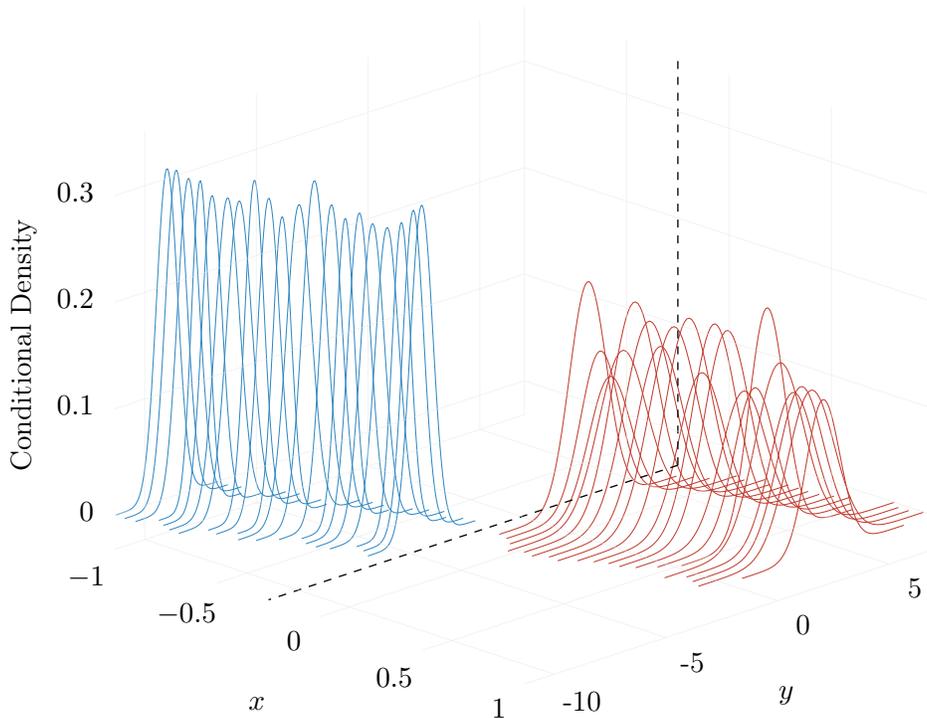

\subsection{Identification} \label{identification_section}

In this section, we discuss the assumptions required for causal identification of the distributional effect $\Psi$. These conditions are nearly identical to the identification requirements established in \cite{FRANDSEN2012382} for quantile treatment effects in discontinuity designs, since in one dimension the Wasserstein distance can be expressed as the $L^2$ distance between quantile functions \citep{vallender1974} --- the only additional assumption required is finite second moments of $P_{a \mid x}$ in order for the Wasserstein distance to be well-defined. For completeness, we still outline each assumption required. Let $F_{Y(a) \mid X}(y \! \mid \! x)$ be the cumulative distribution function of $P_{a \mid x}$. Then, for a sharp treatment assignment $A = \mathbb{I}(X \geq x_0)$, we require:
\begin{enumerate}
    \item[$(i)$] \textit{Consistency:} $Y = Y(a)$ if $A = a$ for $a \in \{0, 1\}$.
    
    \item[$(ii)$] \textit{Continuity:} For $a \in \{0, 1\}$ and all $y \in \mathbb{R}$, $\text{lim}_{x \to x_0} F_{Y(a) \mid X}(y \! \mid \! x) = F_{Y(a) \mid X}(y \! \mid \! x_0)$.
    
    \item[$(iii)$] \textit{Density at threshold:} $f_X(x)$ is differentiable at $x = x_0$ and $\text{lim}_{x \to x_0} f_X(x) > 0$.
    \item[$(iv)$] \textit{Regularity:} $P_{a \mid x} \in \mathcal{P}_2(\mathbb{R})$ for $a \in \{0, 1\}$.
\end{enumerate}

Assumption $(i)$ rules out any interference or spillover effects, where the treatment of one observation affects the outcomes of another. Assumption $(ii)$ ensures that as we approach the cutoff the cumulative distribution functions of the counterfactuals have well-defined limits. This rules out sudden jumps or discontinuities in the outcome distribution that could be unrelated to the treatment assignment. Assumption $(iii)$ guarantees that there are observations arbitrarily close to the cutoff on both sides, which is necessary for well-defined limiting distributions. Finally, assumption $(iv)$ ensures that the Wasserstein distance is well defined by requiring the counterfactual distributions to have a finite second moment. With these assumptions defined, we now establish causal identification in the following lemma.
\begin{lemma}[Identification] \label{identification_theorem}
    Under assumptions $(i)$-$(iv)$ by \cite{FRANDSEN2012382} it follows that
    \begin{align*}
        \Psi = \left\{\int^1_0 \left( Q_1(u) - Q_0(u) \right)^2 du \right\}^{1/2}
\end{align*}
where $Q_{1}(u) = \text{inf} \, \{ y : \text{lim}_{x \downarrow x_0} F_{Y \mid X} (y \! \mid \! x) \geq u\}$ and $Q_{0}(u) = \text{inf} \, \{ y : \text{lim}_{x \uparrow x_0} F_{Y \mid X} (y \! \mid \! x) \geq u\}$ are the limiting conditional quantiles of $Y \mid X = x$ above and below the cutoff.
\end{lemma}

By \cref{identification_theorem}, we can see that the Wasserstein effect $\Psi$ may be expressed as the squared difference in the $u$-th quantile below and above the cutoff, integrated across the entire distribution. This highlights the fact that $\Psi$ measures distributional changes of any form, whether it be changes in location, scale, shape, etc. In the next section, we explore how $\Psi$ can be interpreted, and compare it to the traditional regression discontinuity design estimand $\tau$.
We note that the reduction to quantiles only holds because $Y$ is scalar;
when $Y$ is multivariate the estimation of the Wasserstein effect
is more complicated and will be dealt with in future work.

\subsection{Interpretation} \label{interpretation_section}

In this section, we build intuition for how to interpret the Wasserstein effect, $\Psi$. In particular, we establish an inequality that directly compares $\Psi$ to the average treatment effect at the cutoff and the conditions under which the effects are equal, we demonstrate how the direction of the effect at each quantile can be neatly visualized, we decompose the distributional effect into individual moment effects, and we define a novel measure of effect magnitude by considering the degree to which $Q_1$ stochastically dominates $Q_0$.

\subsubsection{Relation to the Average Treatment Effect} \label{ate_comparison_sec}
In traditional regression discontinuity designs, practitioners are typically interested in estimating the difference in means above and below the treatment cutoff, defined by
\begin{align*}
    \tau = \mathbb{E}[Y(1) - Y(0) \mid X = x_0].
\end{align*}
Notably, this can be interpreted through a distributional lens; $\tau$ is simply measuring the distance between the means of the counterfactual distributions at the cutoff. In fact, if the treatment effect is purely additive (such that it only impacts the distribution means) then it can be shown that these two causal effects are equal. In the following theorem, we establish an inequality between the Wasserstein and mean effects at the cutoff that shows $\Psi$ must be weakly greater than $|\tau|$. Furthermore, we establish the condition under which these effects are identical.

\begin{theorem}[Effect Inequality] \label{effect_inequality}
    The Wasserstein effect upper bounds the average treatment effect at the cutoff, i.e.
    \begin{align*}
        |\tau| \leq \Psi.
    \end{align*}
    Furthermore, equality holds if and only if the treatment effect is purely additive; that is, if for some $\delta \in \mathbb{R}$ and for all $u \in (0, 1)$ that $Q_1(u) = Q_0(u) + \delta$.
\end{theorem}

\cref{effect_inequality} shows that the jump, or discontinuity, in the outcome distributions at the cutoff is always at least as large as the jump in the means. Intuitively, we can think of the relationship between these effects by framing both in terms of the quantile effect function
\begin{align*}
    \Delta Q(u) = Q_1(u) - Q_0(u).
\end{align*}
Suppose that $U \sim \text{Uniform}(0, 1)$. Then, it is clear that $\tau = \int^1_0 \Delta Q(u) du = \mathbb{E}[\Delta Q(U)]$ is simply the average (or signed area) of the quantile effect curve. Meanwhile, we can see that the Wasserstein effect can equivalently be written as $\Psi^2 = \int^1_0 \Delta Q(u)^2 du = \mathbb{E}[\Delta Q(U)^2]$, i.e. the area under the squared quantile effect curve. Immediately, this yields the variance decomposition
\begin{align*}
    \Psi^2 = \tau^2 + \mathbb{V}\! \left(\Delta Q(U) \right).
\end{align*}
Consequently, we can see that $\Psi$ captures the shift in location (as measured by $\tau$) and the heterogeneity around that shift (as measured by the variance of $\Delta Q(U)$). In fact, we can use this decomposition to define a heterogeneity index; let
\begin{align*}
    \gamma := \frac{\mathbb{V}\! \left(\Delta Q(U) \right)}{\Psi^2} = 1 - \left( \frac{|\tau|}{\Psi} \right)^2.
\end{align*}
Then, it is clear that $\gamma \in [0, 1]$. When $\gamma = 0$, the treatment effect is purely additive. Meanwhile, when $\gamma = 1$ the difference in means explains none of the distributional distance.

\subsubsection{Visualizing Quantile Effect Curves} \label{quantile_effect_curve_section}

Considering $\tau$ by itself can conceal important differences: positive and negative quantile effects may cancel out in the average, thereby leaving a small average treatment effect. This problem is readily addressed by the Wasserstein effect. Here, no treatment effect is lost or canceled out since $\Psi$ aggregates these effect differences across all quantiles. However, considering $\Psi$ in isolation can be restrictive since it doesn't describe the direction of the effect at each quantile (e.g. is treatment harmful or helpful). This concern is easily addressed by plotting the quantile effect curve $\Delta Q(u)$ across $u \in (0, 1)$ which lets us directly visualize quantile-by-quantile contributions to the Wasserstein effect. In this sense, our analysis neatly complements existing methods for studying quantile treatment effects, such as in \cite{FRANDSEN2012382}, \cite{QU20151}, \cite{Qu02102019}, and \cite{CHIANG2019589}. In the left panel of \cref{quantile_effect_curves}, we can see two curves, both of which have an average treatment effect of zero. The height at each quantile shows the individual contribution to the overall effect; notably, one effect curve is nearly constant, suggesting a null treatment effect. However, the other curve has a significant negative treatment effect in the left tail of the distribution that is masked by a positive effect near the median. This juxtaposition between effect curves highlights the importance of considering distributional effects over traditional difference-in-means analyses.   Furthermore, we can also neatly visualize the contribution of each quantile to the Wasserstein effect via the contribution function $u \mapsto \frac{1}{\Psi^2}\Delta Q(u)^2$, as shown in the right panel of \cref{quantile_effect_curves}. Here, we can see that most of the Wasserstein effect in the skewed distribution is driven by the left tail. Meanwhile, the null effect curve has nearly a uniformly distributed contribution plot across $u \in (0, 1)$.
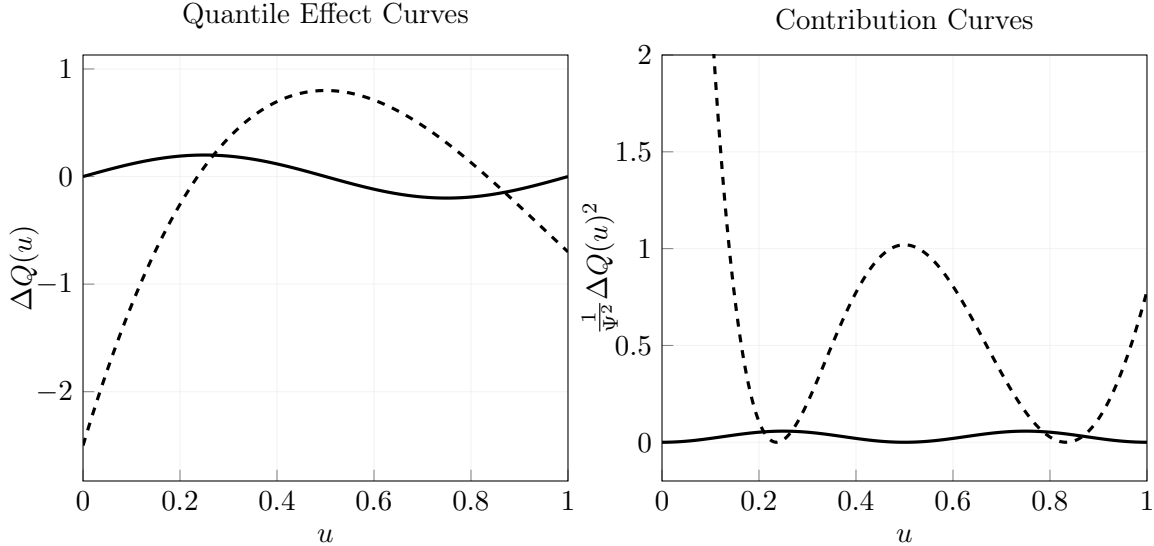
\begin{figure}
    \centering
\begin{tikzpicture}
  \pgfmathsetmacro{\delta}{0}  
  \pgfmathsetmacro{\epsA}{0.20}   
  \pgfmathsetmacro{\psiASq}{\delta*\delta + 0.5*\epsA*\epsA} 

  \pgfmathsetmacro{\alphaB}{0.90} 
  \pgfmathsetmacro{\betaB}{0.30}  
  \pgfmathsetmacro{\psiBSq}{
    \delta*\delta + (4.0/45.0)*\alphaB*\alphaB
    + 0.5*\betaB*\betaB + 2.0*\alphaB*\betaB*(2.0/(pi*pi))
  }

    \pgfmathsetmacro{\alphaA}{2.8} 
  \pgfmathsetmacro{\psiASq}{ (4.0/45.0)*\alphaA*\alphaA }

  \pgfmathsetmacro{\alphaB}{2.4}
  \pgfmathsetmacro{\betaB}{0.9}  
  \pgfmathsetmacro{\psiBSq}{ (4.0/45.0)*\alphaB*\alphaB + (1.0/7.0)*\betaB*\betaB }

  \begin{groupplot}[
      group style={group size=2 by 1, horizontal sep=1.25cm},
      width=0.515\textwidth,
      height=0.465\textwidth,
      grid=both,
      grid style={line width=0.1pt, draw=gray!20, opacity=0.5},
      legend cell align=left,
      legend image post style={sharp plot,-, very thick},
      xtick pos=bottom,
      ytick pos=left,
      domain=0:1, samples=300,
      xmin=0, xmax=1,
      xlabel={$u$},
      ylabel style={yshift=-0.3cm}
    ]

    \nextgroupplot[
      title={Quantile Effect Curves},
      ylabel={$\Delta Q(u)$},
      legend pos=south east
    ]
      \addplot[very thick] ({x}, {\delta + \epsA * sin(deg(2*pi*x))});

      \addplot[very thick, dashed]
        ({x}, { \alphaB*( 1.0/3.0 - 4*(x-0.5)^2 ) + \betaB*(2*x-1)^3 });

    \nextgroupplot[
      title={Contribution Curves},
      ymax=2,
      ylabel={$\frac{1}{\Psi^2}\Delta Q(u)^2$},
      legend pos=north east
    ]
      \addplot[very thick]
        ({x}, { ( (\delta + \epsA * sin(deg(2*pi*x)))^2 ) / (\psiASq) });

      \addplot[very thick, dashed]
        ({x}, { ( \alphaB*( 1.0/3.0 - 4*(x-0.5)^2 ) + \betaB*(2*x-1)^3 )^2 / (\psiBSq) });

  \end{groupplot}
\end{tikzpicture}
    \caption{Quantile effect curves (left panel) and contribution curves (right panel) for a hypothetical null effect curve (solid) and a skewed effect curve (dashed). Both effect curves are defined such that the average treatment effect $\tau = 0$. }
    \label{quantile_effect_curves}
\end{figure}

\subsubsection{Direction of the Treatment Effect}

Visualizing the quantile effect curve is a useful exercise and can help practitioners better interpret the Wasserstein effect, however, it can leave some ambiguity in terms of the overall direction of the treatment effect. In this section, we define a novel one-number summary of the degree to which the treated quantiles dominate the untreated ones. Recall that for any two quantile functions $Q_a(u)$ and $Q_b(u)$, $Q_a$ stochastically dominates $Q_b$ if and only if
\begin{align*}
    Q_a(u) \geq Q_b(u) \quad \text{for all $u \in (0, 1)$},
\end{align*}
as discussed in \cite{QU20151}. Importantly, we can decompose the Wasserstein effect into directional-dominance effects by defining the positive and negative splits,
\begin{align*}
    \Delta Q_+(u) = \text{max}\{ \Delta Q(u), 0 \} \quad \text{and} \quad \Delta Q_-(u) =  \text{max}\{ -\Delta Q(u), 0 \}.
\end{align*}
Intuitively, $\Delta Q_+(u)$ captures all of the positive treatment effects across quantiles (where the difference between $Q_1(u)$ and $Q_0(u)$ is greater than zero), and $\Delta Q_{-}(u)$ captures all of the negative treatment effects. Then, it follows that we may write
\begin{align*}
    \Psi^2 = \int^1_0 \{ \Delta Q(u) \}^2 du = \int^1_0 \{ \Delta Q_+(u) \}^2 du + \int^1_0 \{ \Delta Q_-(u) \}^2 du
\end{align*}
which for notational simplicity we write as $\Psi^2_+ + \Psi^2_-$. Now that we have split the Wasserstein effect into positive and negative quantile effects, we may define the Wasserstein Dominance,
\begin{align*}
    \rho = \frac{\Psi^2_+ - \Psi^2_-}{\Psi^2_+ + \Psi^2_-} \in [-1, 1].
\end{align*}
If $Q_1(u)$ stochastically dominates $Q_0(u)$ then $\Psi^2_- = 0$ and then $\rho = 1$. Similarly, if $Q_0(u)$ stochastically dominates $Q_1(u)$ then $\Psi^2_+ = 0$ and $\rho = -1$. Thus, $\rho$ neatly describes the degree to which one treatment effect dominates the other. When $\rho$ is close to zero, it follows that the quantile effects cross each other, leading to cancellations.

\subsubsection{Decomposition into $L$-Moments}

Although decomposing the Wasserstein effect into $\Delta Q(u)$ is useful and lets us neatly visualize the signed contributions of each quantile, it doesn't say anything about the moments of the counterfactual distributions at the cutoff. Practitioners may be interested in understanding the effect contribution from the differences in means, standard deviations, skewnesses, etc. Fortunately, following a similar approach to \cite{sillitto69}, the Wasserstein effect can be written as a generalized Fourier series using the shifted Legendre polynomials as an orthogonal basis. The shifted Legendre polynomials are defined by $P^*_k(x) = P_k(2x-1)$ where $P_k(x)$ are the usual Legendre polynomials and form an orthogonal basis on $L^2([0, 1])$. A closed form expression for the $k$th shifted Legendre polynomial is given by
\begin{align*}
    P^*_k(x) = (-1)^k \sum^k_{j=0} \binom{k}{j} \binom{k + j}{j} (-x)^j.
\end{align*}
Importantly, under this orthogonal basis it can be shown that $\Psi$ may be decomposed into the summation of squared differences in $L$-moments. First introduced by \cite{hosking1990}, for any random variable $X$ with a finite first moment, the $k$th $L$-moment is defined as
\begin{align*}
    \lambda_k = \int^1_0 Q_x(u) P^*_{k-1}(u) du
\end{align*}
where $Q_x(u)$ is the quantile function for $X$. Note that $P^*_0 = 1$. As shown in \cite{hosking1990}, $L$-moments are defined by taking expectations of linear combinations of order statistics, and represent a ``robust'' analogue of conventional moments of a probability distribution that are typically less sensitive to heavy tailed distributions and are better behaved in small samples. Notably, they are always well-defined (as long as the first moment exists) even when not all conventional moments exist. To build intuition, let $X_{1:n} \leq X_{2:n} \leq \cdots \leq X_{n:n}$ be the order statistics of a random sample of size $n$ from the distribution of $X$. Then, the first three $L$-moments are given by:
\begin{align*}
    \lambda_1 &= \mathbb{E}[X], \\
    \lambda_2 &= \frac{1}{2}\mathbb{E}[X_{2:2} - X_{1:2}], \quad \text{and}\\
    \lambda_3 &= \frac{1}{3} \mathbb{E}[(X_{3:3} - X_{2:3}) - (X_{2:3} - X_{1:3})].
\end{align*}
Focusing on the second $L$-moment, we can see that it is proportional to the expected difference between two independent draws from a distribution. Thus, it provides an alternate measure of dispersion to the traditional standard deviation. Similarly, $\lambda_3$ provides an alternate measure of asymmetry to the traditional skewness by taking the expected difference between upper and lower order statistics. In the following theorem, we establish how $\Psi$ can be decomposed into a summation of squared differences in $L$-moments.

\begin{theorem}[$L$-Moment Decomposition] \label{series_representation}
     Suppose that $P_{a \mid x} \in \mathcal{P}_2(\mathbb{R})$ for $a \in \{0, 1\}$. Then, 
     \begin{align*}
         \Psi^2 = \sum^\infty_{k=1} (2k - 1) \big(\lambda^{(1)}_k - \lambda^{(0)}_k \big)^2
     \end{align*}
     where $\lambda^{(a)}_k = \int^1_0 Q_a(u) P^*_{k-1}(u) du$ are the $k$th $L$-moments above and below the cutoff.
\end{theorem}

By \cref{series_representation} we obtain an important decomposition of $\Psi$: we may now define what can be thought of as a ``distributional $R^2$,'' that is, the amount of the Wasserstein effect that can be explained by a given $L$-moment. For example, for each $k \geq 1$ the share of the total distributional distance explained by the $k$th $L$-moment is given by
\begin{align} \label{rk_equation}
    R^2_k = \frac{(2k - 1)\big(\lambda^{(1)}_k - \lambda^{(0)}_k \big)^2}{\Psi^2}
\end{align}
such that $\sum^\infty_{k=1} R^2_k = 1$. This decomposition is purely distributional: it decomposes the Wasserstein distance between the marginal counterfactual outcome distributions at the cutoff and does not require a rank invariance assumption. As an illustrative example, in \cref{comparison_table} we can see the explanatory power of the first three moments for the effect curves shown in \cref{quantile_effect_curves}. Notably, the null effect curve is primarily explained by variation in its $L$-scale and higher-order moments, as its quantile effect curve is symmetric. Meanwhile, the skewed effect curve is (unsurprisingly) primarily driven by the differences in its $L$-skewness. Note that both have an $L$-location value of zero, since they are both defined to have an average treatment effect of zero. The moment decomposition outlined in \cref{comparison_table} provides a new and powerful tool for decoding treatment effect heterogeneity.

\begin{table}[h]
\centering
\begin{tabular}{r|c|c}
\textit{Moment} & \textit{Null Effect Curve} & \textit{Skewed Effect Curve}\\ \hline
$k = 1$ & 0.0000 & 0.0000 \\
$k = 2$ & 0.6079 & 0.1548 \\
$k = 3$ & 0.0000 & 0.8157 \\
$k \geq 4$ & 0.3921 & 0.0295 \\
\end{tabular}
\caption{Comparison of Explained Distributional Distance}
\label{comparison_table}
\end{table}

Now that we have established several methods of interpreting the Wasserstein effect and how it compares to traditional causal effects, we turn to estimation and inference. In the next section, we formalize an estimator for the Wasserstein effect and derive its asymptotic properties around some chosen bandwidth of the treatment threshold $X = x_0$. We show that standard bias correction techniques can be applied to estimation of the Wasserstein effect such that empirical bandwidth selection methods can be implemented.

\subsection{Estimation and Asymptotics} \label{estimation_inference_section}

In this section, we establish formal properties for estimation of the Wasserstein effect. Note that $\Psi$ depends on one-sided limiting conditional distributions evaluated at a single point; such functionals are not pathwise differentiable, so there is no $\sqrt{n}$-regular estimator and no efficient influence function. We therefore employ a simple plug-in estimator, defined by
\begin{align*}
    \widehat{\Psi}_n = \left\{\int_0^1 (\widehat{Q}_{1}(u) - \widehat{Q}_{0}(u))^2 du\right\}^{1/2}
\end{align*}
where $Q_{1}(u) = \text{inf} \, \{ y : \text{lim}_{x \downarrow x_0} F_{Y \mid X} (y \! \mid \! x) \geq u\}$ and $Q_{0}(u) = \text{inf} \, \{ y : \text{lim}_{x \uparrow x_0} F_{Y \mid X} (y \! \mid \! x) \geq u\}$ are the limiting conditional quantiles of $Y \mid X = x$. Thus, estimation of the Wasserstein effect reduces to estimation of conditional quantile processes (which is a well studied problem), followed by numerical integration.

There are many ways that $Q_{a}(u)$ can be estimated. For example, one natural route is local linear quantile regression, as proposed by \cite{Yu01031998}, which minimizes the check loss of a kernel-weighted polynomial estimator in order to produce boundary-adaptive estimates of the conditional quantile curves. This approach was adapted by \cite{FRANDSEN2012382} when first defining quantile treatment effects in a discontinuity design framework. However, the methods established in \cite{FRANDSEN2012382} only yield pointwise confidence intervals for conditional quantiles. Furthermore, their bandwidth condition requires $\sqrt{nh}h^2 \to \gamma < \infty$. When $\gamma > 0$, the squared bias and variance are of the same order; consequently, undersmoothing must be employed so the bias is negligible relative to the variance and $\gamma \to 0$. In practice, this means that the standard mean-squared-error optimal bandwidth selection of $h \propto n^{1/5}$ can lead to improper coverage. More recently, \cite{QU20151} showed that local quantile regression admits a uniform Bahadur representation which they then leverage to obtain uniform confidence intervals for quantile treatment effects. Building on this framework \cite{Qu02102019} show that by estimating the leading bias term it is possible to obtain bias-adjusted uniform inference in the spirit of \cite{calonico_2014}. Ultimately, the methods established by \cite{QU20151} and \cite{Qu02102019} rely on the fact that the asymptotic distribution is conditionally pivotal, so they are not suitable for the local Wald ratios required by fuzzy designs (which we consider in \cref{fuzzy_rdd_section}), thus, we turn to the framework established in \cite{CHIANG2019589}. Their approach develops a general theory for local Wald estimands that allows for uniform inference across quantiles and can accommodate empirical bandwidth selection. Moreover, it encompasses both sharp and fuzzy discontinuity designs, as well as kinked designs (which we also consider in \cref{kink_designs_section}). We formalize these technical details in what follows. 

In order to estimate $Q_a(u)$, \cite{CHIANG2019589} adapt the local polynomial estimation with bias correction approach established in \cite{calonico_2014}. For $a \in \{0, 1\}$ let 
\begin{align*}
    F^{(k)}_a(y \! \mid \! x_0^{\pm}) = \left. \frac{\partial^k}{\partial x^k} F_{Y \mid X} (y \! \mid \! x) \right|_{x \to x_0^{\pm}}
\end{align*}
be the $k$th partial derivative of the conditional cumulative distribution function where $a = 1$ corresponds to the right limit (as $x \downarrow x_0$) and $a = 0$ corresponds to the left limit (as $x \uparrow x_0$). Then, under appropriate smoothness assumptions, it follows that we may define the following $p$th order one-sided Taylor expansions about $x = x_0$,
\begin{align*}
    F_{Y \mid X} (y \! \mid \! x) &\approx F_{Y \mid X} (y \! \mid \! x_0^+) + \cdots + \frac{F^{(p)}_{Y \mid X} (y \! \mid \! x_0^+)}{p!} (x - x_0)^p = r_p \!\left( \frac{x-x_0}{h} \right)^T \alpha_{1, p}(y)  \\
     F_{Y \mid X} (y \! \mid \! x) &\approx F_{Y \mid X} (y \! \mid \! x_0^-) + \cdots + \frac{F^{(p)}_{Y \mid X} (y \! \mid \! x_0^-)}{p!} (x - x_0)^p = r_p\! \left( \frac{x-x_0}{h} \right)^T \alpha_{0, p}(y)
\end{align*}
for $x > x_0$ and $x < x_0$ respectively, where we say $F_{Y \mid X}(y \! \mid \! x_0^+) = \text{lim}_{x \downarrow x_0} F_{Y \mid X}(y \! \mid \! x)$ and $F_{Y \mid X}(y \! \mid \! x_0^-) = \text{lim}_{x \uparrow x_0} F_{Y \mid X}(y \! \mid \! x)$ are the one-sided limits of $F_{Y \mid X}(y \! \mid \! x)$, we define $r_p(u) = (1, u, \ldots, u^p)^T$, and
\begin{align*}
    \alpha_{a, p}(y) = \left[F_{Y \mid X}(y \! \mid \! x_0^{\pm}), F^{(1)}_{Y \mid X}(y \! \mid \! x_0^{\pm})\frac{h}{1!}, \ldots, F^{(p)}_{Y \mid X}(y \! \mid \! x_0^{\pm})\frac{h^p}{p!}  \right]^T.
\end{align*}
Then, we may estimate the coefficients separately on each side of the treatment discontinuity by solving one-sided local weighted least squares problems, defined by
\begin{align*}
    \widehat{\alpha}_{1, p}(y) &= \underset{\alpha \in \mathbb{R}^{p + 1}}{\text{arg min}} \sum^n_{i=1} \mathbb{I}(X_i \geq x_0) \left(\mathbb{I}(Y_i \leq y) - r_p \!\left( \frac{X_i - x_0}{h} \right)^T \alpha \right)^2 K \! \left( \frac{X_i - x_0}{h} \right)
\end{align*}
where $K(\cdot)$ is some kernel function and the estimator for $\widehat{\alpha}_{0, p}(y)$ follows analogously with $\mathbb{I}(X_i \leq x_0)$. Clearly, if $e_0 = (1, 0, \ldots, 0)^T$ is a standard basis vector it follows that
\begin{align*}
    \widehat{F}_{Y \mid X}(y \! \mid \! x_0^+) = e^T_0 \widehat{\alpha}_{1, p}(y) \quad \text{and} \quad \widehat{F}_{Y \mid X}(y \! \mid \! x_0^-) = e^T_0 \widehat{\alpha}_{0, p}(y).
\end{align*}
However, we are not quite done defining our estimator. From here, \cite{CHIANG2019589} add a bias correction term in the style of \cite{calonico_2014} in order to allow for empirical bandwidth selection. To develop a deeper understanding of this calculation, observe that the bias of our local polynomial estimator is given by 
\begin{align*}
    \mathbb{E}\left[\widehat{F}_{Y \mid X}(y \! \mid \! x_0^{\pm}) \right] - F_{Y \mid X}(y \! \mid \! x_0^{\pm}) = \underbrace{h^{p+1} e^T_0 (\Gamma^{\pm}_p)^{-1} \Lambda^{\pm}_{p, p+1} \frac{F^{(p+1)}_{Y \mid X}(y \! \mid \! x_0^{\pm})}{(p+1)!}}_{\mathcal{B}^{\pm}(y, h, p)} + o(h^{p+1})
\end{align*}
where we define $\mathcal{B}^{\pm}(y, h, p)$ to be the bias such that
\begin{align*}
    \Gamma^{\pm}_p = \int_{\mathbb{R}_{\pm}} K(u) r_p(u) r_p(u)^T du \quad \text{and} \quad \Lambda^{\pm}_{p, q} = \int_{\mathbb{R}_{\pm}} u^q K(u) r_p(u) du.
\end{align*}
Intuitively, $\Gamma^{\pm}_p$ is a matrix that describes how the polynomial regressors interact under the kernel weights and $\Lambda^{\pm}_{p, q}$ captures how the next higher-order term in the Taylor expansion interacts with the regressors. Then, Lemma 1 of \cite{CHIANG2019589} shows that under some set of regularity conditions
\begin{align*}
   \Delta^{\pm}_{\mathcal{B}}(y) := \sqrt{nh}\left( \widehat{F}_{Y \mid X}(y \! \mid \! x_0^{\pm}) - F_{Y \mid X}(y \! \mid \! x_0^{\pm}) - \widehat{\mathcal{B}}^{\pm}(y, h, p) \right) 
\end{align*}
admits the uniform Bahadur representation
\begin{align*}
    \Delta^{\pm}_{\mathcal{B}}(y) &= \sum^n_{i=1} \frac{e^T_0 (\Gamma^{\pm}_p)^{-1} r_p \!\left( \frac{X_i - x_0}{h}\right) K \!\left( \frac{X_i - x_0}{h} \right)\left( \mathbb{I}(Y_i \leq y ) - F_{Y \mid X}(y \! \mid \! X_i) \right) \delta^{\pm}_i }{\sqrt{nh} f_X(x_0)} + o_{P \mid X}(1)
\end{align*}
where $\delta^+_i = \mathbb{I}(X_i \geq x_0)$ and $\delta^-_i = \mathbb{I}(X_i \leq x_0)$. We defer the reader to Assumption 1 of \cite{CHIANG2019589} for a comprehensive list of these regularity conditions. Notably, \cite{CHIANG2019589} require that the Kernel function $K(\cdot)$ is bounded and continuous and is of VC type, which allows for common kernels such as the uniform, triangular, biweight, triweight, and Epanechnikov kernels, but rules out the Gaussian kernel due to its unbounded support. Furthermore, for some bandwidth $h$ satisfying $h \to 0$, they require
\begin{align*}
    nh^2 \to \infty \quad \text{and} \quad nh^{2p+3} \to 0.
\end{align*}
The former condition is a stronger assumption than the typical $nh \to \infty$ in order to allow for uniform convergence of the quantile process, and the latter condition controls the bias relative to the variance.

Now that we have defined this machinery, we discuss the conditional weak convergence of our estimator. First, note that \cite{CHIANG2019589} consider convergence of the quantile process after trimming the left and right tails, such that $u \in [\varsigma, 1 - \varsigma]$ for some $\varsigma \in (0, 1/2)$. They do so since near the tails the conditional quantile function can be difficult to estimate reliably, so instead they establish weak convergence in $l^\infty([\varsigma, 1-\varsigma])$. However, in order to properly estimate the Wasserstein effect we need to extend the domain of the quantiles to the full support on $[0, 1]$. Therefore, for weak convergence we require the additional assumptions that:
\begin{enumerate}
    \item[$(i)$]  The potential outcomes are compactly supported. \vspace{-0.05in}
    \item[$(ii)$] $f_{Y(a) \mid X}(y \! \mid \! x)$ is uniformly bounded away from zero on that support.
\end{enumerate}
With these assumptions in place, let
\begin{align*}
    \nu^{\pm}_n(y) = \sum^n_{i=1} \frac{e^T_0 (\Gamma^{\pm}_p)^{-1} r_p \!\left( \frac{X_i - x_0}{h}\right) K \!\left( \frac{X_i - x_0}{h} \right)\left( \mathbb{I}(Y_i \leq y) - F_{Y \mid X}(y \! \mid \! X_i) \right) \delta^{\pm}_i }{\sqrt{nh} f_X(x_0)}.
\end{align*}
Then, by Theorem 1 of \cite{CHIANG2019589} it follows that $\nu^{\pm}_n \rightsquigarrow \mathbb{G}_{H^{\pm}}$ where $\mathbb{G}_{H^{\pm}}$ are zero mean Gaussian processes with some covariance function $H^{\pm}$. Now that we have established conditional weak convergence for the bias corrected cumulative distribution functions above and below the cutoff, we need to invert them in order to obtain weak convergence for the quantile processes. Simply put, we define
\begin{align*}
    \widehat{Q}_{1}(u) &= \text{inf} \, \{ y :  \widehat{F}_{Y \mid X}(y \! \mid \! x_0^{+}) - \widehat{\mathcal{B}}^{+}(y, h, p)\geq u\} \quad \text{and} \\
    \widehat{Q}_{0}(u) &= \text{inf} \, \{ y :  \widehat{F}_{Y \mid X}(y \! \mid \! x_0^{-}) - \widehat{\mathcal{B}}^{-}(y, h, p)\geq u \}
\end{align*}
such that the quantile treatment effect may be defined as $\Delta \widehat{Q}(u) = \widehat{Q}_{1}(u) - \widehat{Q}_{0}(u)$. From here, since the quantile map $F \mapsto F^{-1}$ is Hadamard differentiable, we may apply the functional delta method to see that
\begin{align*}
    \sqrt{nh}(\widehat{Q}_{1}(u) - Q_1(u)) &\rightsquigarrow - \frac{\mathbb{G}_{H^+}(Q_1(u))}{f_{Y \mid X}(Q_1(u) \! \mid \! x_0^+)} \quad \text{and} \\
    \sqrt{nh}(\widehat{Q}_{0}(u) - Q_0(u)) &\rightsquigarrow - \frac{\mathbb{G}_{H^-}(Q_0(u))}{f_{Y \mid X}(Q_0(u) \! \mid \! x_0^-)}.
\end{align*}
Consequently, it follows that
\begin{align*}
    \sqrt{nh}( \Delta \widehat{Q}(u) -  \Delta Q(u)) \rightsquigarrow  \frac{\mathbb{G}_{H^-}(Q_0(u))}{f_{Y \mid X}(Q_0(u) \! \mid \! x_0^-)} - \frac{\mathbb{G}_{H^+}(Q_1(u))}{f_{Y \mid X}(Q_1(u) \! \mid \! x_0^+)}.
\end{align*}
In practice, $p$ is often chosen to be two, yielding local quadratic polynomial estimators. Higher order polynomials can potentially reduce bias even further, but they also come with the risk of a larger variance due to sensitivity of the estimator near the boundary. In the next section, we discuss inference for the Wasserstein effect.

\begin{remark}[Conditioning on covariates]
    Although covariates are not required for identification, they are often of interest to practitioners both to obtain covariate indexed causal effects and to improve precision \citep{Frolich02102019, calonico_covariates_2019}. Let $W_i \in \mathbb{R}^d$ denote a vector of covariates and let $\widehat{\mu}_{a, W}(x)$ denote a local polynomial estimate of $\mathbb{E}[W \mid X = x]$; for example, following the same estimation procedure described in \cref{estimation_inference_section}. We may then define the centered covariates
    \begin{align*}
        \widetilde{W}_{a, i} = W_i - \widehat{\mu}_{a, W}(X_i)
    \end{align*}
    for $a \in \{0, 1\}$. Then, following \cite{CHIANG2019589}, for each $y$ we solve for
\begin{align*}
     (\widetilde{\alpha}_{1, p}(y), \widetilde{\vartheta}_1(y)) &= \underset{\alpha \in \mathbb{R}^{p + 1}, \vartheta \in \mathbb{R}^d}{\text{arg min}} \sum^n_{i=1}  \delta^+_i \left(\mathbb{I}(Y_i \leq y) - r_p \!\left( \frac{X_i - x_0}{h} \right)^T \alpha - \widetilde{W}^T_{1, i} \vartheta \right)^2 K \! \left( \frac{X_i - x_0}{h} \right)
\end{align*}
and analogously, for $(\widetilde{\alpha}_{0, p}(y), \widetilde{\vartheta}_0(y))$ by replacing $\delta^+_i$ with $\delta^-_i$ and $\widetilde{W}_{1, i}$ with $\widetilde{W}_{0, i}$. From here, if our goal is target $F_{Y \mid X}(y \! \mid \! x_0^+)$ (using the covariates only as variance-reducing nuisances), then we simply take 
\begin{align*}
    \widetilde{F}_{Y \mid X}(y \! \mid \! x_0^+) = e^T_0 \widetilde{\alpha}_{1, p}(y) \qquad \text{and} \qquad \widetilde{F}_{Y \mid X}(y \! \mid \! x_0^-) = e^T_0 \widetilde{\alpha}_{0, p}(y).
\end{align*}
Note that since we have centered our covariates we can now safely interpret each estimate as the cumulative distribution function at the average covariate value. If our goal is the conditional cumulative distribution function itself,
\begin{align*}
    F_{Y \mid X, W}(y \! \mid \! x^+_0, w) = \underset{x \downarrow x_0}{\mathrm{lim}} \big\{ \mathbb{P}(Y \leq y \! \mid \! X = x, W = w) \big\}
\end{align*}
then for any $w \in \mathbb{R}^d$ we define our estimator to be $\widetilde{F}_{Y \mid X, W}(y \! \mid \! x_0^{\pm}) = e^T_0 \widetilde{\alpha}_{a, p}(y) + w^T\widetilde{\vartheta}_a(y)$. With these estimators in place, we can now estimate the conditional Wasserstein effect
\begin{align*}
    \Psi(w) =  \left\{\int^1_0 \left( \widetilde{Q}_1(u; w) - \widetilde{Q}_0(u; w) \right)^2 du \right\}^{1/2}
\end{align*}
where $\widetilde{Q}_1(u; w)$ and $\widetilde{Q}_0(u; w)$ are the inverses of $\widetilde{F}_{Y \mid X, W}(y \! \mid \! x_0^{+})$ and $\widetilde{F}_{Y \mid X, W}(y \! \mid \! x_0^{-})$, respectively. Bias correction and the multiplier bootstrap can be implemented following the same covariate augmentation procedure with higher-order local fits.
\end{remark}

\subsection{Statistical Inference} \label{inference_section}

Now that we have established methods for estimation of quantile treatment effects as well as their limiting distributions, we turn to inference for the Wasserstein effect. Surprisingly, statistical inference in this setting is not as straightforward as one might expect since $\Psi$ is a quadratic parameter; here, the limiting distribution and rate of convergence change as $\Psi \to 0$. To illustrate this point broadly for quadratic parameters, \cite{verdinelli_wasserman_dvi_2024} consider a toy example where $X_1, \ldots, X_n \sim N(\mu, \sigma^2)$ and we are interested in estimating $\psi = \mu^2$. Using the estimator $\widehat{\psi} = \Bar{X}^2_n$, it follows that
\begin{align*}
    \sqrt{n}(\widehat{\psi} - \psi) &\rightsquigarrow  N(0, \eta^2) \hspace{-1in} &&\text{for some $\eta^2$ when $\mu \neq 0$, and} \\
    n \widehat{\psi} & \rightsquigarrow \sigma^2 \chi^2_1  \hspace{-1in} &&\text{when $\mu = 0$.}
\end{align*}
Moreover, when $\mu$ is close to zero, its distribution will be neither normal nor $\chi^2_1$, and its rate of convergence will be between $1 / n$ and $1 / \sqrt{n}$. This is a common (and perhaps understudied problem) in statistics; other parameters such as kernel two-sample statistics \citep{gretton12a} and Reproducing Kernel Hilbert Space corrections \citep{Sejdinovic_RKHS_test_2013} suffer from this misalignment of convergence around the null. In the context of distributional discontinuity design, the delta method fails for our functional $\Psi^2 = \int^1_0 \left[\Delta Q(u)\right]^2 du$, so we must construct our hypothesis tests and confidence intervals around this fact. We first consider testing the null hypothesis of no distributional change above and below the cutoff; that is, $\Delta Q(u) = 0$ for all $u \in (0, 1)$ (or equivalently that $\Psi = 0$). Then, we define two methods of constructing valid (but conservative) confidence intervals for $\Psi$.

\subsubsection{Testing the Null Hypothesis} \label{hypothesis_testing_section}

In this section, we test the null hypothesis of no causal effect. Under the null, it follows that $\Delta Q(u) = 0$ for all $u \in (0, 1)$. Furthermore, as discussed in \cref{estimation_inference_section} it follows that $\sqrt{nh} \Delta \widehat{Q}(u) \rightsquigarrow \mathbb{G}(u)$ where $\mathbb{G}(u)$ is a mean-zero Gaussian process with covariance kernel $\kappa$. From here, we may apply the Karhunen-Lo\`eve theorem \citep{Karhunen1946ZurSS, loeve1977probability} to expand $\mathbb{G}(u)$ as
\begin{align*}
    \mathbb{G}(u) = \sum^\infty_{k=1} \sqrt{\lambda_k} Z_k \phi_k(u)
\end{align*}
where $\{\phi_k \}^\infty_{k=1}$ are an orthonormal basis on $L^2([0, 1])$ defined by the eigenfunctions of the covariance operator induced by the kernel $\kappa(u, v)$ (with eigenvalues $\lambda_1, \lambda_2, \ldots$) and $Z_k \sim N(0, 1)$ for all $k$. Then, it follows that
\begin{align} \label{chaos_eq}
   nh \widehat{\Psi}^2_n =  \int^1_0 \left( \sqrt{nh}\Delta \widehat{Q}(u)\right)^2 du \rightsquigarrow \sum^\infty_{k=1} \lambda_k Z^2_k,
\end{align}
which is a second-order Gaussian (or Wiener-It\^o) Chaos \citep{Janson_1997}. From here, there are several ways we can go about conducting our hypothesis test. The first option is to directly estimate the eigenvalues of $\kappa(u, v) = \text{Cov}(\mathbb{G}(u), \mathbb{G}(v))$ and approximate \cref{chaos_eq} via Monte-Carlo simulation. Although in principle this appears to be a straightforward procedure, the validity of such a test is not automatic as we must estimate $\lambda_1, \lambda_2, \ldots$, truncate $\sum^K_{k=1} \lambda_k Z^2_k$ for some $K$, and approximate the null distribution via Monte-Carlo simulation. In the following theorem, we formally establish the conditions required to obtain a valid level-$\alpha$ test under this procedure.

\begin{theorem}[Eigenvalue Test] \label{eigenvalue_test}
    Suppose that $\sum^\infty_{k=1} \lambda_k < \infty$ with $\lambda_1 > 0$ and define the Monte-Carlo draws
    \begin{align*}
        \widehat{T}^*_{K_n, b} = \sum^{K_n}_{k=1} \widehat{\lambda}_{k, n} Z^2_{k, b}
    \end{align*} 
    where $\widehat{\lambda}_{k,n}$ are the estimated eigenvalues and $Z_{k, b} \sim N(0, 1)$ for $k \geq 1$ and $b = 1, \ldots, B_n$. Let $\widehat{c}^*_{n, \alpha}$ be the empirical $(1-\alpha)$ quantile computed from $\{\widehat{T}^*_{K_n,b}\}_{b=1}^{B_n}$. Then, supposing that $B_n \to \infty$ and $K_n \to \infty$, for any $\alpha \in (0, 1)$ it follows that
\begin{align*}
    \lim_{n\to\infty}
  \mathbb{P}_{H_0}(nh \widehat{\Psi}^2_n > \widehat{c}^*_{n,\alpha})
  = \alpha
\end{align*}
as long as $|| \widehat{\kappa}_n - \kappa||_{2} = o_{\mathbb{P}}(K^{-1/2}_n)$.
\end{theorem}

By \cref{eigenvalue_test}, we can see that the conditions required to obtain a level-$\alpha$ test using Monte-Carlo simulation depend crucially on the number of terms included in our truncation, $K_n$. Importantly, there are two errors introduced by simulating the critical value: the truncation error, controlled by $\sum_{k > K_n} \lambda_k$, and the estimation error, controlled by $\sqrt{K_n}|| \widehat{\kappa}_n - \kappa||_{2}$. Thus, $K_n$ must diverge to eliminate the truncation error, but not so fast that the estimation error fails to vanish. One way to obtain a rate for $K_n$ is to assume some kind of polynomial eigenvalue decay of the form $\lambda_k \lesssim k^{-\beta}$ for some $\beta > 1$; in the following corollary we formalize this notion.

\begin{corollary}[Eigenvalue Decay] \label{eigenvalue_decay}

Assume the conditions of \cref{eigenvalue_test} hold and suppose that $|| \widehat{\kappa}_n - \kappa||_{2} = O_p(r_n)$ for some $r_n \to 0$. Furthermore, suppose that there exist constants $C_\lambda > 0$ and $\beta > 1$ such that for all $k$, $\lambda_k \leq C_\lambda k^{-\beta}$. Then, it follows that letting
\begin{align*}
    K_n \asymp r^{-2/(2\beta - 1)}_n
\end{align*}
balances the truncation bias and estimation error, such that
\begin{align*}
    \sum_{k>K_n}\lambda_k = O\!\left(r_n^{\frac{2(\beta-1)}{2\beta-1}}\right) \quad \text{and} \quad \sqrt{K_n}|| \widehat{\kappa}_n - \kappa||_{2}
= O_p\!\left(r_n^{\frac{2(\beta-1)}{2\beta-1}}\right).
\end{align*}
    
\end{corollary}

\cref{eigenvalue_decay} clarifies the relationship between both the truncation bias and estimation error, as well as the $K_n$ and the rate of eigenvalue decay. Clearly, faster eigenvalue decay (i.e. a larger $\beta$) allows for a smaller $K_n$; in this setting there will be less sensitivity to estimating $\kappa$. Conversely, slower decay requires a larger $K_n$ and therefore requires more accurate estimation of the covariance operator. A natural choice for the rate is $r_n \asymp (nh)^{-1/2}$, as this aligns with the effective sample size in a discontinuity design setting.

While \cref{eigenvalue_test} and \cref{eigenvalue_decay} establish a useful testing framework, choosing $K_n$ in practice can be tricky. That, combined with the computational burden of Monte-Carlo simulation, suggests the eigenvalue test may be less than desirable for practitioners. Alternatively, one could leverage Theorem 5 of \cite{luedtke_omnibus_2018} to obtain a conservative, but computationally simple statistical test. Specifically, \cite{luedtke_omnibus_2018} derive non-parametric tests of equality in distribution between unknown functions; they show that such a test also manifests as a Gaussian chaos, which can be easily bounded by applying a one-sided Chebyshev inequality. In the following proposition, we leverage their results to obtain a conservative test for no causal effect. 

\begin{proposition}[Conservative Test] \label{conservative_test}

Let $\mu =  \int^1_0 \kappa(u, u) du$ and $\sigma^2 =  2 \int^1_0 \!  \int^1_0 \kappa(u, v)^2 du \, dv$. Fix $\alpha \in (0, 1)$ and define $c^{\text{ub}}_{1- \alpha} = \mu + \sigma \sqrt{(1 - \alpha) / \alpha}$. Then, by \cite{luedtke_omnibus_2018} it follows that
\begin{align*}
    \underset{n \to \infty}{\text{lim} \, \text{sup}} \ \mathbb{P}_{H_0}(nh \widehat{\Psi}^2_n > \widehat{c}^{\text{ub}}_{n, 1- \alpha}) \leq \alpha
\end{align*}
where $\widehat{c}^{\text{ub}}_{n, 1- \alpha} = \widehat{\mu} + \widehat{\sigma} \sqrt{(1 - \alpha) / \alpha}$ for any estimators such that $\widehat{\mu} \overset{p}{\rightarrow} \mu$ and $\widehat{\sigma} \overset{p}{\rightarrow} \sigma$.
    
\end{proposition}
Note that we may equivalently define $\mu$ and $\sigma^2$ as $\sum^\infty_{k=1} \lambda_k$ and $2 \sum^\infty_{k=1} \lambda^2_k$, respectively. \cref{conservative_test} provides us with a more convenient statistical test that requires fewer assumptions on the estimation error of $\kappa$. Practically speaking, the condition $K_n^{1/2} ||\widehat{\kappa}_n-\kappa||_2 =o_{\mathbb{P}}(1)$ required in \cref{eigenvalue_test} means the eigenvalue test is only trustworthy when $\widehat\kappa_n$ is estimated accurately enough that one can include many eigenvalues without the simulated critical value becoming sensitive to $K_n$. With a small sample size, $\widehat\kappa_n$ may only support a small $K_n$, making the test fragile to the truncation choice and potentially anti-conservative if $K_n$ is pushed too large. In such settings the conservative test is preferable; it avoids estimating the full eigenspectrum and instead requires only consistent estimation of $\mu$ and $\sigma$. Finally, we note that one may reject the null hypothesis using a one-sided $1 - \alpha$ upper confidence bound for $\Psi$ using the intervals defined in the following section.

\subsubsection{Constructing Confidence Intervals} \label{confidence_intervals_section}

As discussed in \cite{verdinelli_wasserman_dvi_2024}, constructing confidence intervals for quadratic parameters with uniformly correct coverage (with length $n^{-1/2}$ away from the null and length $n^{-1}$ at the null) is an unsolved problem in statistics. In practice, we deal with this problem by constructing intervals that are conservative near the null. We consider two approaches for constructing such intervals. Later, in \cref{application_section}, we compare the coverage and width of both methods via simulation.

First, we consider constructing a confidence interval for $\Psi$ using the uniform confidence band defined for the quantile treatment effect. As shown by \cite{CHIANG2019589}, we can construct a multiplier bootstrap process $\mathbb{G}^*_n$ such that $\mathbb{G}^*_n \rightsquigarrow \mathbb{G}$. Therefore, if we let $\widehat{c}_{n,\alpha}$ be the $1-\alpha$ conditional quantile of $\sup_{u}|\mathbb{G}_n^*(u)|$, it follows that $\Delta \widehat{Q}(u) \pm \frac{1}{\sqrt{nh}} \widehat{c}_{n, \alpha}$ yields a $1 - \alpha$ confidence band. With that in mind, let
\begin{align*}
    a_n(u) = \Delta\widehat{Q}(u) - \frac{\widehat{c}_{n, \alpha}}{\sqrt{nh}} \quad \text{and} \quad   b_n(u) = \Delta\widehat{Q}(u) + \frac{\widehat{c}_{n, \alpha}}{\sqrt{nh}}.
\end{align*}
Then, it is clear that over the interval $[a, b]$ that
\begin{align*}
    \underset{x \in [a, b]}{\mathrm{max}} \, x^2 = \begin{cases}
b^2,& a\geq 0 \\
a^2,& b\leq 0,\\
\mathrm{max}\{a^2,b^2\},& a<0<b
\end{cases}
\qquad \text{and} \qquad 
\underset{x \in [a, b]}{\mathrm{min}} \, x^2 =
\begin{cases}
a^2, & a \geq 0 \\
b^2, & b \leq 0 \\
0 & a<0<b.
\end{cases}
\end{align*}
Therefore, if we define the upper and lower bounds
\begin{align*}
    \overline{M}_n(u) = \text{max}\{ a_n^2(u), b_n^2(u) \} \quad \text{and} \quad \underline{M}_n(u) = (\text{max}\{a_n(u), 0\})^2 + (\text{min}\{b_n(u), 0\})^2
\end{align*}
then it becomes straightforward to construct the interval $C_n = [\int^1_0 \underline{M}_n(u)du , \int^1_0 \overline{M}_n(u) du ]$. Under the regularity conditions established in \cite{CHIANG2019589} it immediately follows that
\begin{align*}
    \underset{n \to \infty}{\text{lim} \, \text{inf}} \  \mathbb{P}(\Psi^2 \in C_n) \geq 1 - \alpha.
\end{align*}

Alternatively, we can artificially widen our confidence interval following the approach of \cite{verdinelli_wasserman_dvi_2024}. Specifically, we could define
\begin{align} \label{conservative_interval}
    C^\prime_n = \left[ \widehat{\Psi}^2_n \pm z_{1-\alpha/2}\sqrt{\widehat{s}^2_n + \frac{c^2}{nh}} \, \right]
\end{align}
where $\widehat{s}_n$ is the estimated standard deviation of $\Psi^2$, $z_{1-\alpha/2}$ is the $1 - \alpha / 2$ quantile of a standard Normal distribution, and $c$ is some constant, such as $\mathbb{V}(Y)$.  We now confirm that this provides a valid, but possibly conservative, confidence interval.

\begin{lemma}[Conservative Interval] \label{coverage_lemma}
Let $C^\prime_n$ be the interval defined in \cref{conservative_interval} for some constant $c$. Suppose that $\mathbb{E}[\widehat{\Psi}^2_n - \Psi^2] = o((nh)^{-1/2})$ and $\mathbb{V}(\widehat{\Psi}^2_n) = o((nh)^{-1})$. Then, it follows that 
\begin{align*}
      \mathbb{P}(\Psi^2 \not \in C^
    \prime_n) = o(1).
\end{align*}

\end{lemma}

In practice, either of the proposed methods for constructing confidence intervals for $\Psi$ is reasonable; their empirical widths are further discussed in \cref{application_section}. Additionally, as noted in \cref{hypothesis_testing_section}, we can check the null hypothesis of no causal effect by checking if zero is in $C_n$ or $C^\prime_n$. In the following section, we extend our analysis to the fuzzy treatment assignment setting.

\subsection{Fuzzy Distributional Discontinuity Design} \label{fuzzy_rdd_section}

In many applications treatment assignment above and below the cutoff is not perfectly sharp. That is, although the probability of receiving treatment jumps discontinuously at the threshold, some units below the threshold may receive treatment, and some above may not; such settings are referred to as ``fuzzy'' regression discontinuity designs \citep{hahn2001}. Intuitively, in this setting the cutoff acts as an instrument for treatment status; crossing the threshold changes the likelihood of treatment but does not fix it. Now, it no longer makes sense to directly compare outcome distributions above and below the cutoff because these groups differ in more than treatment status. Notably, \cite{FRANDSEN2012382} extend the framework proposed by \cite{Angrist01061996} to define local always-takers, never-takers, compliers, defiers, and indefinites. To do so, let $X_i$ be the running variable with cutoff $x_0$ and now let $A_i(x)$ denote unit $i$'s potential treatment status if the running variable were $x$. Observed treatment is then $A_i = A_i(X_i)$. Then, we define the one-sided treatment limits (when they exist) as
\begin{align*}
    A^-_i = \underset{x \uparrow x_0}{\text{lim}} \{ A_i(x) \} \quad \text{and} \quad A^+_i = \underset{x \downarrow x_0}{\text{lim}} \{ A_i(x) \} 
\end{align*}
such that $A^-_i$ is the treatment status that would be received if the running variable approaches the cutoff from the left and $A^+_i$ the treatment that would be received from the right, then we may define the following mutually exclusive groups:
\begin{itemize}
    \item \textit{Local Always-Takers}: $\mathcal{AT} = \{i : A^-_i = 1, A^+_i = 1 \}$.
    \item \textit{Local Never-Takers}: $\mathcal{NT} = \{i : A^-_i = 0, A^+_i = 0 \}$.
    \item \textit{Local Compliers}: $\mathcal{C} = \{i : A^-_i = 0, A^+_i = 1 \}$
    \item \textit{Local Defiers}: $\mathcal{D} = \{i : A^-_i = 1, A^+_i = 0 \}$
    \item \textit{Local Indefinites}: $\mathcal{I} = \{i : \text{one or both of $(A^-_i, A^+_i)$ do not exist}\}$.
\end{itemize}

Now, we focus on the sub-population of local compliers when defining a fuzzy distributional effect, as this is the group whose treatment status is actually changed by the treatment discontinuity. In this setting, we need require additional assumptions for causal identification. Beyond the assumptions discussed in \cref{identification_section}, we need to assume
\begin{enumerate}
    \item[$(v)$] \textit{Treatment Discontinuity:} $\text{lim}_{x \downarrow x_0} \mathbb{P}(A = 1 \mid X = x) > \text{lim}_{x \uparrow x_0} \mathbb{P}(A = 1 \mid X = x)$.
    \item[$(vi)$] \textit{Local Smoothness:} $\mathbb{E}[A^\pm \mid X = x]$ and $F_{Y(a) \mid G=g, X}(y \! \mid \! g, x)$ are continuous at $X = x_0$, the latter for all $y$, each $a \in \{0, 1\}$, and each $g \in \{\mathcal{AT}, \mathcal{NT}, \mathcal{C}, \mathcal{D} \}$.
    \item[$(vii)$] \textit{Monotonicity}: $\text{lim}_{x \to x_0} \mathbb{P}( A^+ \geq A^- \mid X = x) = 1$ and $\mathbb{P}(\mathcal{I}) = 0$.
\end{enumerate}

Assumption $(v)$ simply requires that the probability of treatment changes discontinuously at the threshold $X = x_0$. Assumption $(vi)$ requires that the fraction of units that would take treatment evolves smoothly as the treatment cutoff is approached. This guarantees that the only discontinuity of the observed treatment assignment is through the discontinuity at $X = x_0$ and not through some hidden break in the treatment assignment mechanism. Furthermore, $(vi)$ requires that (within each compliance group) the distribution of the potential outcomes varies smoothly with the running variable at the cutoff. This again ensures that  any discontinuity in observed outcome distributions is attributable to the change in the probability of treatment, rather than a discontinuity in the potential outcome distributions themselves. Assumption $(vii)$ rules out the existence of defiers (units that always go against their treatment assignment) and indefinites (units with ill-defined treatment limits) in a neighborhood of the treatment discontinuity. Intuitively, this assumption implies that all units weakly comply with the treatment assignment mechanism, such that moving from below to above the cutoff can only increase the chance of treatment. Furthermore, it ensures that every unit has well-defined potential treatment statuses. With these assumptions in place, it now follows that the group identified by the discontinuity are genuine compliers.

With assumptions $(v)$-$(vii)$ in place \cite{FRANDSEN2012382} show that the cumulative distribution functions for compliers above and below the cutoff are identified as
\begin{align*}
    F_{1 \mid \mathcal{C}}(y) = \frac{\text{lim}_{x \downarrow x_0} \mathbb{E}[\mathbb{I}(Y \leq y) A \mid X = x] - \text{lim}_{x \uparrow x_0} \mathbb{E}[\mathbb{I}(Y \leq y) A \mid X = x]}{\text{lim}_{x \downarrow x_0} \mathbb{E}[A \mid X = x] - \text{lim}_{x \uparrow x_0} \mathbb{E}[A \mid X = x] } 
\end{align*}
and
\begin{align*}
    F_{0 \mid \mathcal{C}}(y) = \frac{\text{lim}_{x \downarrow x_0} \mathbb{E}[\mathbb{I}(Y \leq y)(1 - A) \mid X = x] - \text{lim}_{x \uparrow x_0} \mathbb{E}[\mathbb{I}(Y \leq y) (1 - A) \mid X = x]}{\text{lim}_{x \downarrow x_0} \mathbb{E}[(1 - A) \mid X = x] - \text{lim}_{x \uparrow x_0} \mathbb{E}[(1 -A) \mid X = x] }.
\end{align*}
Therefore, it follows that the Wasserstein effect for compliers is defined and identified as
\begin{align*}
    \Psi_{\mathcal{C}} = \left\{\int^1_0 \left( Q_{1 \mid \mathcal{C}}(u) - Q_{0 \mid \mathcal{C}}(u) \right)^2 du\right\}^{1/2}
\end{align*}
where $Q_{a \mid \mathcal{C}}(u) = \text{inf} \{ y : F_{a \mid \mathcal{C}}(y) \geq u \}$ for $a \in \{0, 1\}$ are the quantiles of the complier cumulative distribution functions above and below the treatment discontinuity. Interpretation of the fuzzy Wasserstein effect follows analogously to the sharp case; $\Psi^2_{\mathcal{C}}$ still acts as a distributional analogue upper bounding the local average treatment effect at the cutoff. Moreover, the same inequalities and decompositions can be extended to $\Psi^2_{\mathcal{C}}$. Now that we have defined the Wasserstein effect in a fuzzy distributional discontinuity design framework, we discuss estimation.

\subsubsection{Estimation of the Fuzzy Wasserstein Effect}

To estimate the fuzzy Wasserstein effect, we can directly extend the procedure described in \cref{estimation_inference_section}, again following the work of \cite{CHIANG2019589}. For $a \in \{0, 1\}$ we define $G_a(y \mid x)= \mathbb{E}\left[\mathbb{I}(Y\leq y)\mathbb{I}(A=a)\mid X=x\right]$ and $\pi_a(x) = \mathbb{E}\left[\mathbb{I}(A=a)\mid X=x\right]$. Then, taking the corresponding one-sided limits we have
\begin{align*}
 F_{a\mid\mathcal{C}}(y)= \frac{G_a(y \! \mid \! x_0^+)-G_a(y \! \mid \! x_0^-)}{\pi_a(x_0^+)-\pi_a(x_0^-)}.
\end{align*}
To estimate $G_a(y \! \mid \! x_0^\pm)$ and $\pi_a(x_0^{\pm})$ we use one-sided local polynomial regression, mirroring the sharp case. Specifically, for $G_a(y \! \mid \! x)$ we solve
\begin{align*}
    \widehat{\alpha}^{G}_{a,p}(y) = \underset{\alpha \in \mathbb{R}^{p + 1}}{\text{arg min}}
\sum_{i=1}^n \delta_i^{\pm}\Big(
\mathbb{I}(Y_i \leq y) \mathbb{I}(A_i=a) -r_p \!\left(\frac{X_i-x_0}{h}\right)^T \alpha
\Big)^2 K \!\left(\frac{X_i-x_0}{h}\right).
\end{align*}
Analogously, for $\pi_a(x_0^{\pm})$ we solve $\widehat{\alpha}^{\pi}_{a,p}$ by letting $\mathbb{I}(A_i=a)$ be the dependent variable. Then, it follows that $\widehat{G}_a(y \! \mid \! x_0^\pm) = e^T_0\widehat{\alpha}^{G}_{a,p}(y)$ and $\widehat{\pi}_a(x_0^\pm)= e^T_0\widehat{\alpha}^{\pi}_{a,p}$, which yields local Wald estimator
\begin{align*}
    \widehat{F}_{a\mid\mathcal{C}}(y)=\frac{\widehat{G}_a(y \! \mid \! x_0^+) - \widehat{G}_a(y \! \mid \! x_0^-)}{\widehat{\pi}_a(x_0^+)-\widehat{\pi}_a(x_0^-)}.
\end{align*}
Finally, the complier quantile function is given by $\widehat{Q}_{a\mid\mathcal{C}}(u)=\inf\{y:\widehat{F}_{a\mid\mathcal{C}}(y)\ge u\}$. Statistical inference for the fuzzy Wasserstein effect follows exactly as in the sharp case, as described in \cref{inference_section}.

\section{Distributional Kink Designs} \label{kink_designs_section}

In many practical settings, we do not observe a discontinuity in the treatment assignment, but rather a kink or change in slope of the policy. The idea here is the same as under regression discontinuity designs; it is assumed that units arbitrarily close to either side of the policy kink are comparable, and therefore a causal interpretation can be justified. A canonical application of regression kink designs is that of \cite{nielsen2010estimating}, who estimate the causal effect of student grants on college enrollment in Denmark. Here, their running variable $X$ is some continuous measure of parental income and the benefit $b(X)$ exhibits a kink at different eligibility thresholds; that is, the full grant is offered up to some level $X=x_1$, a linear phaseout occurs between $x_1$ and $x_2$ (with a decreasing benefit in $x$ given), and zero benefit is offered for parental incomes greater than $x_2$. Other notable early applications of regression kink designs can be found in \cite{NBERw8269} and \cite{DAHLBERG20082320}; interested readers should refer to \cite{card_theory_practice_2016} for a review, and to \cite{ando2017trust, Ganong03042018} for discussions of inference and robustness in finite samples.
\begin{figure}[h]
\centering
\begin{tikzpicture}
\definecolor{plotlyBlue}{HTML}{2E86C1}
\definecolor{plotlyRed}{HTML}{C0392B}

\pgfplotstableread[col sep=comma]{files/rkd_points.csv}\pts
\pgfplotstableread[col sep=comma]{files/rkd_true_line.csv}\line

\begin{axis}[
    width=0.8\textwidth,
    height=0.4\textwidth,
    scale only axis,  
  grid=major,
  major grid style={line width=0.2pt, draw=black!12},
  minor tick num=1,
  minor grid style={line width=0.1pt, draw=black!6}, 
  legend cell align=left, xtick pos=bottom, ytick pos=left,
  xlabel={$b(x)$},
  ylabel={$y$},
  xmin=0, xmax=6,
  ymin=-0.3, ymax=2.0,
    xtick={0,1, 2,3,4,5, 6},
  xticklabels={0,1, 2,3,4,5, 6},
     xtick pos=bottom,
    ytick pos=left,
]

\addplot[
  only marks,
  mark=*,
  mark size=0.55pt,
  draw=none,
  fill=black,
  opacity=0.38
]
table[x="X", y="Y"]{\pts};

\addplot[
  no markers,
  line width=5.2pt,
  draw=white,
  line cap=round
]
table[x="X", y="trueY"]{\line};

\addplot[
  no markers,
  line width=2pt,
  color=plotlyBlue,
  line cap=round
]
table[x="X", y="trueY"]{\line};

\addplot[
  densely dashed,
  line width=1.5pt,
  plotlyRed
]
coordinates {(2.5,-0.3) (2.5,2.0)};

\end{axis}
\end{tikzpicture}
\caption{Example of a hypothetical regression kink design.}
\label{kink_design}
\end{figure}
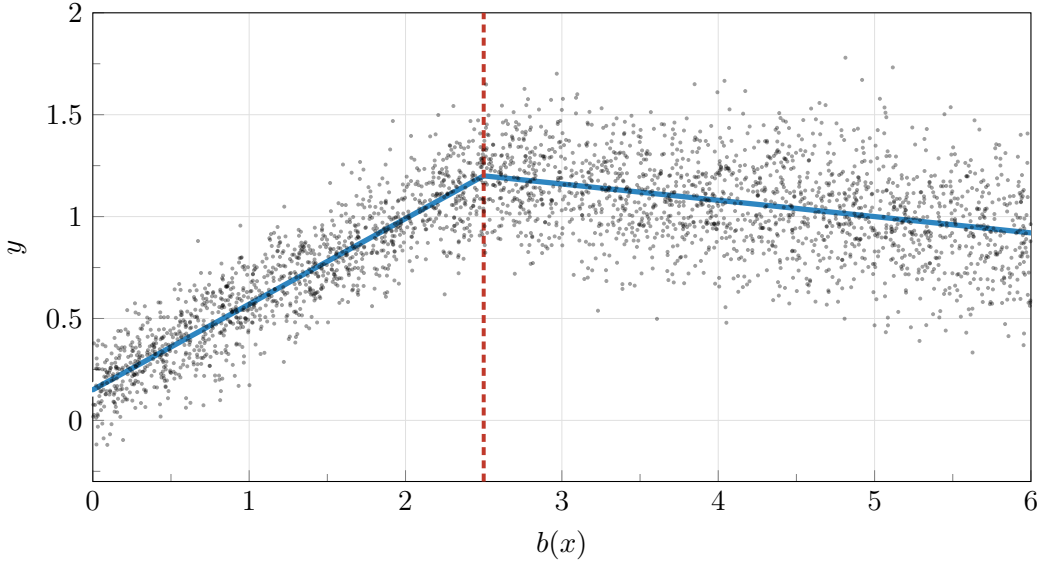

In a sharp regression kink design the benefit is set deterministically according to the known assignment rule $b(\cdot)$. One causal target is the local average marginal effect of the benefit on the outcomes; that is, the \textit{slope} of the dose-response curve at the policy kink,
\begin{align*}
\tau^\prime = \frac{\partial}{\partial t} \mathbb{E}[Y(t) \mid X = x_0] \, \Big|_{t=b(x_0)} \overset{(i)}{=} \frac{\mu^\prime_Y(x_0^+) - \mu^\prime_Y(x_0^-)}{b^\prime(x_0^+) - b^\prime(x_0^-)}
\end{align*}
where equality $(i)$ follows by the identifying assumptions outlined in \cite{card_rkd_2015} and
\begin{align*}
    \mu^\prime_Y(x_0^+) &= \lim_{x\downarrow x_0} \left\{ \frac{ \partial }{ \partial x} \mathbb{E}[Y\mid X=x] \right\}, \\
    b^\prime(x_0^+) &=\lim_{x\downarrow x_0} \left\{\frac{\partial}{ \partial x}b(x) \right\},
\end{align*}
and analogous definitions are given for $\mu^\prime_Y(x_0^-)$ and $b^\prime(x_0^-)$. For example, in \cref{kink_design} we can see an example of a regression kink design, where there is a clear kink in the dose-response curve at $X = x_0$. As was the case for regression discontinuity design, far more information can be gained by considering distributional causal effects. Notably, quantile treatment effects in kinked designs have been explored by \cite{CHIANG2019405}, \cite{CHIANG2019589}, \cite{Chen_Chiang_Sasaki_2020}, and \cite{wang2025unifiedframeworkidentificationinference}. However, these approaches suffer from the same set of drawbacks as before --- namely, difficulty in implementation and interpretation. Thus, in what follows, we show that the Wasserstein derivative at the policy kink provides a clean generalization of traditional kink design effects.

Let $g(t, x, \varepsilon)$ be a function of the benefit, running variable, and unobservables. Then, we may define the counterfactual $Y(t) = g(t, X, \varepsilon)$ and the observed outcome $Y = g(b(X), X, \varepsilon)$. Again, we let $P_{t\mid x}$ denote the conditional distribution of $Y(t) \mid X = x$ under some benefit or treatment level $t = b(x)$. In this setting, we can think of $P_{t \mid x}$ as a distribution along some absolutely continuous path of distributions in the running variable $x$. This allows us to define the Wasserstein derivative at the policy kink $X=x_0$ as
\begin{align*}
\Psi^\prime = \lim_{\delta \to 0} 
\left\{ \frac{W_2(P_{t_0 + \delta \mid x_0},P_{t_0 \mid x_0})}{|\delta|} \right\} = \left\{ \int^1_0 \left( \left. \frac{\partial}{\partial t} Q_{Y(t) \mid X=x_0}(u) \right|_{t=b(x_0)} \right)^2 du \right\}^{1/2}
\end{align*}
where we define $t_0 = b(x_0)$. Intuitively, $\Psi^\prime$ represents the instantaneous rate at which probability mass moves or \textit{flows} at the policy kink \citep{AmbrosioGigliSavare2005GradientFlows}. While the traditional regression kink design estimand $\tau^\prime$ measures how the center of mass moves or drifts through the kink, $\Psi^\prime$ measures how the entire distribution moves.

Now we consider identification of the Wasserstein derivative at the kink. Assume that $b^\prime(x_0^+)\neq b^\prime(x_0^-)$ and $b(\cdot)$ is a known function. Naturally, we expect the behavior of
$F_{Y \mid X}(y \! \mid \! x)$ near $x_0$
to provide information
about the causal effect; however, making this intuition rigorous is subtle. For identification and interpretation, the classical approaches of \cite{card_rkd_2015} and \cite{CHIANG2019405} can be surprisingly difficult to work with. In the case of mean effects, \cite{card_rkd_2015} show that $\tau^\prime$ can be written as a weighted average of individual-level marginal effects where the weights depend on unobservables. For quantile effects, \cite{CHIANG2019405} obtain an analogous weighted-average of structural derivatives evaluated along some latent boundary set. Both settings can be hard to translate into standard treatment-effect language, e.g. ``what is the effect of a marginal increase in the benefit.'' Furthermore, these identification strategies can require additional regularity conditions to make the weights well-defined; for example, \cite{CHIANG2019405} suggest a rank-invariance assumption. More recently, 
\cite{wang2025unifiedframeworkidentificationinference}, proposed a more direct identification strategy that leads to a cleaner interpretation. Specifically, \cite{wang2025unifiedframeworkidentificationinference} define the local treatment effect at the kink to be
\begin{align*}
  \Delta_\phi =  \frac{\partial}{\partial t} \left( \phi(F_{Y(t) \mid X = x_0})\right) \Big|_{t = b(x_0)} = \lim_{\delta\to 0} \left\{\frac{\phi(F_{Y(t_0 + \delta) \mid X = x_0})-\phi(F_{Y(t_0) \mid X = x_0})}{\delta} \right\}
\end{align*}
where $\phi$ is some Hadamard differentiable functional. Now, the estimand describes a genuinely local average marginal effect of a small policy-induced change in the benefit level around $b(x_0)$ for units at $X=x_0$. Under some regularity conditions, \cite{wang2025unifiedframeworkidentificationinference} show that the causal effect $\Delta_\phi$ is identified as $\phi^\prime_{F_{Y \mid X = x_0}}( \mathrm{DRKD}(\cdot))$ where $\mathrm{DRKD}(\cdot)$ is the distributional regression kink design estimand,
\begin{align*}
    \mathrm{DRKD}(y) =
\frac{\frac{\partial}{\partial x} F_{Y\mid X}(y \! \mid \! x_0^+)-\frac{\partial}{\partial x} F_{Y \mid X}(y \! \mid \! x_0^-)}
{b^\prime(x_0^+) - b^\prime(x_0^-)}. 
\end{align*}
In the case of distributional kink designs, we let $\phi_u(F) = F^{-1}(u)$ denote the $u$-quantile functional. Then, since $Y$ is univariate, it follows that the Wasserstein derivative at the policy kink is identified as
\begin{align*}
    \Psi^\prime = \left\{\int^1_0 \left( \frac{\frac{\partial}{\partial x} Q_{Y\mid X}(u \! \mid \! x_0^+)-\frac{\partial}{\partial x} Q_{Y \mid X}(u \! \mid \! x_0^-)}
{b^\prime(x_0^+) - b^\prime(x_0^-)}\right)^2 du\right\}^{1/2}.
\end{align*}
Notably, if $Y$ is multivariate this identification strategy does not work as the Wasserstein distance is no longer a function of the limiting conditional quantiles; establishing and identifying such distributional causal effects in high dimensional settings is an interesting and open question. In the following section, we extend the work of \cite{wang2025unifiedframeworkidentificationinference} to handle identification of fuzzy treatment assignment in kink designs.

\subsection{Fuzzy Distributional Kink Designs} \label{fuzzy_kink_design_section}

Although \cite{wang2025unifiedframeworkidentificationinference} establish a clean and interpretable framework for identification of causal effects in sharp kink designs, they do not consider the fuzzy treatment assignment setting, where the running variable induces a kink in treatment propensities rather than deterministically setting a benefit level. Here, we observe a noisy analogue of $b(x)$ due to imperfect compliance, measurement error, or some other unobserved determinants of behavior. More formally, suppose we observe some $b(X, \eta)$ where $\eta$ captures unobserved variation in the treatment assignment. Now, there is no single baseline level of treatment, so we must define an analogous version of $\Delta_\phi$ in the fuzzy setting, and we must establish additional structure/conditions to ensure that the unobserved determinants of treatment evolve smoothly in $x$ around the kink. To formalize this, we first define a nonseparable outcome model and establish a fuzzy kink design characterization. \vspace{0.1in} 

\begin{assumption}[Nonseparable Model] \label{nonsep_model_assumption}
Suppose there exist unobservables $(\varepsilon,\eta)$ and a measurable structural function
$g:\R\times\mathcal X\times\mathcal E\to \R$ such that:
\begin{enumerate}
\item[$(i)$] (Potential outcomes) For each $t \in \mathbb{R}$, $Y(t)=g(t,X,\varepsilon)$.
\item[$(ii)$] (Fuzzy assignment) $T = b(X,\eta)$ for some measurable $b:\mathcal X\times\mathcal H\to \mathbb{R}$.
\item[$(iii)$] (Consistency) $Y = Y(T) = g(b(X,\eta),X,\varepsilon)$
\end{enumerate}
where $\varepsilon \in \mathcal{E} \subset \mathbb{R}^{d_\varepsilon}$, $\eta \in \mathcal{H}  \subset \mathbb{R}^{d_\eta}$, and $\mathcal{X} \subset \mathbb{R}$ is the support of the running variable.
\end{assumption} \medskip

\begin{assumption}[Fuzzy Kink Characterization]\label{fuzzy_characterization} Let $I_{x_0}$ be a closed interval containing the kink point $x_0$. Then, suppose that:
\begin{enumerate}
\item[$(i)$] For a.e.\ $\eta$, the map $x \mapsto b(x,\eta)$ is continuous on $I_{x_0}$ and continuously differentiable on $I_{x_0}\setminus\{x_0\}$, with finite one-sided derivatives
$b^\prime(x_0^+,\eta)$ and $b^\prime(x_0^-,\eta)$.
\item[$(ii)$] Let $\mu_B(x) = \mathbb{E}[b(X, \eta) \mid X = x]$ and assume that
\begin{align*}
    \Delta_B := \mu^\prime_B(x^+_0) - \mu^\prime_B(x^-_0) \neq 0
\end{align*}
where $\mu^\prime_B(x^+_0) = \underset{x\downarrow x_0}{\mathrm{lim}} \left\{\frac{\partial}{ \partial x} \mathbb{E}[b(x, \eta) \mid X = x]  \right\}$ and $\mu^\prime_B(x^-_0) = \underset{x\uparrow x_0}{\mathrm{lim}} \left\{\frac{\partial}{ \partial x} \mathbb{E}[b(x, \eta) \mid X = x]  \right\}$.
\end{enumerate}
\end{assumption}

By \cref{fuzzy_characterization}, we guarantee that there is continuity at the kink point, but a discontinuity in the first order derivative. Moreover, under both \cref{nonsep_model_assumption} and \cref{fuzzy_characterization}, letting $T_0 = b(x_0, \eta)$ and $\omega(\eta) = b^\prime(x_0^+,\eta) - b^\prime(x_0^-,\eta)$, we can define the counterfactual treatment at the kink
\begin{align*}
    T_\delta = T_0 + \delta \! \left( \frac{\omega(\eta)}{\Delta_B} \right),
\end{align*}
and the associated counterfactual outcome $Y_\delta = g(T_\delta ,x_0,\varepsilon)$. \medskip 

\begin{remark}[Counterfactual Definition]
    The definition of the counterfactual $T_\delta$ may seem counterintuitive at first glance since the intervention depends on the kink-responsiveness $\omega(\eta) = b^\prime(x_0^+,\eta) - b^\prime(x_0^-,\eta)$; it might seem more natural to define an intervention such as $T_\delta = T_0 + \delta$, which shifts every unit by the same amount. The reason the weighting $\omega(\eta) / \Delta_B$ is required is that in a fuzzy kink design, the identifying variation comes from the change in the slope of $\mathbb{E}[b(X, \eta) \mid X]$ at $X = x_0$. A small change in the running variable shifts a unit's treatment by an amount proportional to $\omega(\eta)$, so units more responsive to the kink contribute more to the local change actually observed in the data. The reason we normalize by $\Delta_B = \mathbb{E}[\omega(\eta) \mid X = x_0]$ is to allow $\delta$ to be interpreted as a one-unit change in the average treatment at the kink, since now $\mathbb{E}[T_\delta - T_0 \mid X = x_0] = \delta$. Thus, even though $T_\delta$ doesn't correspond to a uniform shift in treatment, it is the counterfactual that aligns with the observed kink-variation, with weights determined by each unit's kink-responsiveness.
\end{remark} \medskip

With this counterfactual outcome in place, we may now define the fuzzy local treatment effect at the kink for the functional $\phi$ as
\begin{align*}
    \Delta_\phi^{F} = \left.\frac{\partial }{\partial \delta} \phi\left(F_{Y_\delta\mid X=x_0}\right)\right|_{\delta=0} = 
\lim_{\delta\to 0} \left\{
\frac{\phi\big(F_{Y_\delta\mid X=x_0}\big)-\phi \big(F_{Y_0\mid X=x_0}\big)}{\delta} \right\},
\end{align*}
provided the limit exists. Next, our goal is to obtain a structural representation of $\Delta_\phi^{F}$ analogous to Lemma 1 of \cite{wang2025unifiedframeworkidentificationinference}. However, before doing so we must outline a few additional assumptions. First, we need identical smoothness assumptions to those required in \cite{wang2025unifiedframeworkidentificationinference}; for completeness, we write them out in what follows. \medskip

\begin{assumption}[Smooth Functional] \label{hadamard_differentiability}

Let $\mathcal{F}$ be the space of all one-dimensional distribution functions. Then, assume the functional $\phi : \mathcal{F} \to \mathbb{R}$ is Hadamard differentiable at $F_{Y \mid X = x_0}$, with its Hadamard derivative denoted by $\phi^\prime_{F_{Y \mid X = x_0}}$.
\end{assumption} \medskip

\begin{assumption}[Smooth Structural Functions] \label{smooth_g_assumption}
    The function $g(t,x, e)$ is continuously differentiable in $(t,x)$ for each $e \in \mathcal{E}$, with continuous partial derivatives
    \begin{align*}
        g_1(t,x, e) = \frac{\partial}{\partial t } g(t,x,e) \qquad \text{and} \qquad g_2(t,x, e) = \frac{\partial}{\partial x} g(t,x,e).
    \end{align*}
\end{assumption} \medskip

As discussed in \cite{wang2025unifiedframeworkidentificationinference}, \cref{smooth_g_assumption} is analogous to the smoothness conditions imposed in \cite{card_rkd_2015}, but weaker than those required by the identification strategy of \cite{CHIANG2019405}. Under this smoothness condition, the partial derivative of $h(x, e, u) := g(b(x, u), x, e)$ with respect to $x$ is given by
\begin{align} \label{partial_deriv}
    \frac{\partial}{\partial x} h(x, e, u) := h_x(x, e, u) =  b^\prime(x, u) g_1(b(x, u), x, e) + g_2(b(x, u), x, e).
\end{align}

An implication of \cref{partial_deriv} is that $x \mapsto g_1(b(x, u), x, e)$ and $x \mapsto g_2(b(x, u), x, e)$ are continuous at $x_0$, but $x \mapsto h(x, e, u)$ is not continuously differentiable at $x_0$ due to the discontinuity in $b^\prime(x, u)$. Finally, we must establish a few regularity conditions in the spirit of conditions $R1(i)$ and $R1(ii)$ of \cite{wang2025unifiedframeworkidentificationinference}. 

\begin{assumption}[Regularity 1] \label{regularity_conditions} Let $Z = \left(\omega(\eta) /\Delta_B \right) g_1(T_0,x_0,\varepsilon)$ and assume the following conditions hold:
    \begin{enumerate}
        \item[$(i)$] For each $c > 0$,
        \begin{align*}
            P \left(\big|Y_\delta-Y_0-\delta Z\big|\ge c|\delta| \mid X=x_0 \right)=o(|\delta|)
        \end{align*}
        as $\delta \to 0$. 
        
        \item[$(ii)$] The conditional distribution of $(Y,Z)$ given $X=x_0$ is absolutely continuous with respect to the Lebesgue measure and has a joint density $f_{Y,Z\mid X}(y,y^\prime \! \mid \! x_0)$ that is continuous in $y$ for all $y^\prime$. Furthermore, assume there exists a Lebesgue integrable function $\varpi : \mathbb{R} \to \mathbb{R}$ with $\int | y^\prime \varpi(y^\prime) | dy^\prime < \infty$ such that for all $(y, y^\prime)$, 
        \begin{align*}
            f_{Y,Z\mid X}(y, y^\prime \! \mid \! x_0) \leq |\varpi(y^\prime)|.
        \end{align*}
    \end{enumerate}
\end{assumption}

\cref{regularity_conditions}$(i)$ is a stochastic differentiability condition along the counterfactual path $\delta\mapsto Y_\delta$ induced by the fuzzy kink. It requires that, conditional on $X = x_0$, the change in outcomes from shifting treatment from $T_0$ to $T_0+\delta(\omega(\eta)/\Delta_B)$ admits a first-order expansion with a remainder that is small enough to control. \cref{regularity_conditions}$(ii)$ ensures that the joint density $(Y,Z)\mid X=x_0$ is well-behaved. Relative to the regularity conditions required in \cite{wang2025unifiedframeworkidentificationinference}, the difference now is that the derivative direction can vary. That is, $Z$ includes the random compliance weight $\omega(\eta)/\Delta_B$, so the domination and integrability requirements must control the weighted marginal effect $(\omega(\eta)/\Delta_B) g_1(T_0,x_0,\varepsilon)$ and not just $g_1(T_0,x_0,\varepsilon)$ alone. Note that we don't explicitly need to require an additional smooth outcome distribution assumption like Assumption S4 of \cite{wang2025unifiedframeworkidentificationinference} since \cref{regularity_conditions}$(ii)$ already implies continuity of $y$ in $f_{Y \mid X}( \cdot \! \mid \! x_0)$. With these conditions in place, we can now establish a structural representation of $\Delta^F_\phi$.

\begin{lemma}[Fuzzy Structural Representation]
\label{fuzzy_WZ_lemma1}
Suppose that Assumptions 1-5 hold. Then the fuzzy local treatment effect at the kink $\Delta^F_\phi$ admits the representation
\begin{align*}
\Delta_\phi^F = 
\phi^{\prime}_{F_{Y\mid X=x_0}}\left( \mathbb{E}\left[ -f_{Y\mid X}(\, \cdot \mid x_0) \left(\frac{\omega(\eta)}{\Delta_B} \right) g_1(T_0,x_0,\varepsilon) \mid Y= \, \cdot \, , X=x_0\right]\right).
\end{align*}
\end{lemma}

As discussed in \cite{wang2025unifiedframeworkidentificationinference}, \cref{fuzzy_WZ_lemma1} shows that the fuzzy local treatment effect at the kink can be expressed as the Hadamard derivative of $\phi$ applied to a conditional expectation. Notably, this conditional expectation is analogous to the local average structural derivative discussed in \cite{hoderlein2007, hoderlein2009}. The primary difference between \cref{fuzzy_WZ_lemma1} and Lemma 1 of \cite{wang2025unifiedframeworkidentificationinference} is that now the conditional expectation is compliance weighted; it averages $g_1$ evaluated at each unit's (random) baseline treatment $T_0 = b(x_0, \eta)$, weighted by the unit’s kink-responsiveness $\omega(\eta)$, such that units whose treatment is more strongly shifted by the kink contribute more to the identified marginal effect. With this structural representation in place, we can now finalize proof for causal identification of the fuzzy local treatment effect at the kink; however, we first need to establish a few more assumptions and regularity conditions. \medskip

\begin{assumption}[Smooth Disturbance Distributions] \label{smoothness_unobservables}
    The conditional distribution of $(\varepsilon,\eta)$ given $X=x$ is absolutely continuous with respect to Lebesgue measure. Furthermore, it admits a density $f_{\varepsilon,\eta \mid X}(e,u \mid x)$ that is continuously differentiable in $x$ on $I_{x_0}$ for all $(e,u)$. Finally, assume there exists some Lebesgue integrable function $\varpi(e, u)$ such that 
    \begin{align*}
        \underset{x \in I_{x_0}}{\mathrm{sup}} \ \left| \frac{\partial}{\partial x} f_{\varepsilon, \eta \mid X}(e, u \mid x) \right| \leq |\varpi(e, u) |.
    \end{align*}
    Finally, for each $y$ assume $\mathbb{I}(h(x_0 + t, e, u) \leq y) \to \mathbb{I}(h(x_0, e, u) \leq y)$ as $t \to 0$ for all $(e, u)$.
\end{assumption} \medskip

\begin{assumption}[Regularity 2] \label{theorem_regularity_conditions}
    Recall that we use the notational shorthand $h_x(x^{\pm}_0, \varepsilon, \eta) := \frac{\partial}{\partial x } h(x_0^{\pm},\varepsilon,\eta)$. Assume the following conditions hold:
    \begin{enumerate}
        \item[$(i)$] For each $c > 0$,
        \begin{align*}
            P\left(\, \left|h(x_0+\delta,\varepsilon,\eta)-h(x_0,\varepsilon,\eta)-\delta  h_x(x_0^+,\varepsilon,\eta) \right| \geq c|\delta| \mid X=x_0 \right) &= o(|\delta|) \\
P\left( \, \left|h(x_0+\delta,\varepsilon,\eta)-h(x_0,\varepsilon,\eta)-\delta h_x(x_0^-,\varepsilon,\eta)  \right| \geq c|\delta| \mid X=x_0 \right) &= o(|\delta|)
\end{align*}
as $\delta \downarrow 0$ and $\delta \uparrow 0$, respectively.
    \item[$(ii)$] The conditional distributions of $(Y, h_x(x_0^\pm,\varepsilon,\eta))$ given $X=x_0$ are absolutely continuous with respect to the Lebesgue measure with densities $f_{Y, h^\pm_x\mid X}(y, y^\prime \mid x_0)$  that are continuous in $y$ for each fixed $y^\prime$. Moreover, there exists a Lebesgue integrable function $\varpi_h : \mathbb{R} \to \mathbb{R}$ with $\int |y^\prime \varpi_h(y^\prime)| dy^\prime <\infty$ such that for all $y, y^\prime$,
    \begin{align*}
        f_{Y, h^\pm_x\mid X}(y,y^\prime \! \mid \! x_0) \leq |\varpi_h(y^\prime)|.
    \end{align*}
    \item[$(iii)$] Suppose that the conditional distribution of $\eta$ given $X=x$ is absolutely continuous with respect to Lebesgue measure, with conditional density $f_{\eta\mid X}(u\mid x)$ that is continuously differentiable in $x$ on $I_{x_0}$. Furthermore, that there exists some Lebesgue integrable function $\varpi_\eta(u)$ such that
    \begin{align*}
        \sup_{x\in I_{x_0}}\left| \frac{\partial}{\partial x}f_{\eta\mid X}(u \! \mid \! x) \right|\leq \varpi_\eta(u).
    \end{align*}
    \item[$(iv)$] Assume the function $b(x,u)$ is continuous in $x$ at $x_0$ for each $u$ and differentiable on each side of $x_0$ with one-sided derivatives $b^\prime(x_0^\pm,u)$. Finally, assume there exist Lebesgue integrable functions $\kappa_0,\kappa_1$ such that
    \begin{align*}
\sup_{x\in I_{x_0}}|b(x,u)|\le \kappa_0(u)
\qquad \text{and} \qquad 
\sup_{x\in I_{x_0} \setminus \{x_0\}}\left|\frac{\partial}{\partial x}b(x,u)\right|\le \kappa_1(u)
    \end{align*}
together with $\int \!  \kappa_0(u)\varpi_\eta(u)du<\infty$ and $\mathbb{E}[\kappa_1(\eta)\mid X=x_0]<\infty$.
    \end{enumerate}
\end{assumption}

\cref{smoothness_unobservables} allows for $(\varepsilon,\eta)$ to be both correlated with $X$ and to vary with $x$, however, it requires this variation be smooth on $I_{x_0}$. \cref{theorem_regularity_conditions}$(i)$ is a local linearization requirement for $h(x,\varepsilon,\eta)$ around $x_0$ that ensures that small changes in $x$ induce approximately linear shifts in $Y$ that are controlled by the one-sided derivatives $h_x(x_0^\pm,\varepsilon,\eta)$. Assumption $(ii)$ is another regularity condition ensuring the joint distribution of $(Y,h_x(x_0^\pm,\varepsilon,\eta))\mid X=x_0$ is well behaved at the kink. Finally, $(iii)$ ensures that any selection effect arising from $x\mapsto f_{\eta\mid X}(\cdot\mid x)$ is smooth and therefore does not itself generate a kink and $(iv)$ adds integrability conditions on $b(x,\eta)$ (and its derivative). With these assumptions in place, we now establish causal identification of the fuzzy local treatment effect at the kink.

\begin{theorem}[Fuzzy Kink Design Identification] \label{fuzzy_identification}
    Suppose the conditions of \cref{fuzzy_WZ_lemma1} and Assumptions 6-7 hold. Then, the fuzzy local treatment effect at the kink is identified as
    \begin{align*}
        \Delta_\phi^{F} = \phi'_{F_{Y 
        \mid X=x_0}}\big(\mathrm{FDRKD}(\cdot)\big),
    \end{align*}
    where $\mathrm{FDRKD}(\cdot)$ is the fuzzy distributional regression kink design estimand,
    \begin{align*}
        \mathrm{FDRKD}(y) = \frac{\frac{\partial}{\partial x} F_{Y \mid X}(y \! \mid \! x_0^+) - \frac{\partial}{\partial x} F_{Y \mid X}(y \! \mid \! x_0^-)}
{\mu^\prime_B(x_0^+) - \mu^\prime_B(x_0^-)}.
    \end{align*}
\end{theorem}

Analogously to the identification results of \cite{wang2025unifiedframeworkidentificationinference} in sharp kink designs, \cref{fuzzy_identification} shows that the fuzzy local treatment effect at the kink is identified by applying the Hadamard derivative of the functional $\phi$ in the direction of the FDRKD estimand. Importantly, $\mathrm{FDRKD}(y)$ is interpretable as the local distributional effect per unit of the kink-induced treatment change, represented as a distributional local Wald ratio. Clearly, in the case of distributional kink designs the fuzzy Wasserstein derivative at the kink is identified as
\begin{align*}
     \Psi^\prime_{\mathcal{C}} = \left\{\int^1_0 \left( \frac{\frac{\partial}{\partial x} Q_{Y\mid X}(u \! \mid \! x_0^+)-\frac{\partial}{\partial x} Q_{Y \mid X}(u \! \mid \! x_0^-)}
{\mu^\prime_B(x_0^+) - \mu^\prime_B(x_0^-)}\right)^2 du\right\}^{1/2}
\end{align*}
after letting $\phi$ be the quantile function. Now that we have established identification in the fuzzy kink design setting, in the next section we discuss estimation.

\subsection{Estimation and Inference for the Kinked Wasserstein Effect}

In this section, we discuss estimation of the Wasserstein derivative at a policy kink. Our strategy will build off of the work of \cite{CHIANG2019589} and the framework established in \cref{estimation_inference_section}. First, note that $Q_{Y\mid X}(u \! \mid \! x_0^+) = Q_{Y\mid X}(u \! \mid \! x_0^-) =: Q_{Y\mid X}(u \! \mid \! x_0)$ due to the continuity at $x_0$. Second, recall that the derivative with respect to $x$ of the quantile function can be written as
\begin{align*}
    \frac{\partial}{\partial x}Q_{Y\mid X}(u \! \mid \! x_0^\pm) = -\frac{\frac{\partial}{\partial x}F_{Y\mid X} \! \left(Q_{Y\mid X}(u \! \mid \!  x_0) \! \mid \! x_0^\pm\right)}
{f_{Y\mid X} \! \left(Q_{Y\mid X}(u \! \mid \! x_0) \! \mid \! x_0\right)}.
\end{align*}
Thus, after taking the difference $\frac{\partial}{\partial x}Q_{Y\mid X}(u \! \mid \! x_0^+)-\frac{\partial}{\partial x}Q_{Y\mid X}(u \! \mid \! x_0^-)$ and then dividing by the first-stage kink $\mu^\prime_B(x_0^+)-\mu^\prime_B(x_0^-)$, it is clear that
\begin{align} \label{fuzzy_eq}
    \frac{\frac{\partial}{\partial x}Q_{Y\mid X}(u \! \mid \! x_0^+)-\frac{\partial}{\partial x}Q_{Y\mid X}(u \! \mid \! x_0^-)}
{\mu^\prime_B(x_0^+)-\mu^\prime_B(x_0^-)} = -\frac{\mathrm{FDRKD} \! \left(Q_{Y\mid X}(u \! \mid \! x_0)\right)}
{f_{Y\mid X} \!\left(Q_{Y\mid X}(u \! \mid \! x_0) \! \mid \! x_0\right)}.
\end{align}
With this in mind, we can see that estimation of $Q_{Y\mid X}(u \! \mid \! x_0)$ is the same as in \cref{estimation_inference_section}; the only additional terms to estimate are $\frac{\partial}{\partial x} F_{Y \mid X}(y \! \mid \! x_0^\pm)$, $\mu^\prime_B(x_0^\pm)$, and $f_{Y \mid X}( \cdot \mid x_0)$. We begin by considering estimation of the first two terms. Recall that in \cref{estimation_inference_section} we considered a one-sided Taylor expansion of $F_{Y \mid X}(y \! \mid \! x)$ about $x = x_0$. Specifically, we defined
\begin{align*}
    \alpha_{a, p}(y) = \left[F_{Y \mid X}(y \! \mid \! x_0^{\pm}), F^{(1)}_{Y \mid X}(y \! \mid \! x_0^{\pm})\frac{h}{1!}, \ldots, F^{(p)}_{Y \mid X}(y \! \mid \! x_0^{\pm})\frac{h^p}{p!}  \right]^T.
\end{align*}
and then estimated $\alpha_{a, p}(y)$ via one-sided local weighted least squares. Consequently, leveraging this exact approach it follows that
\begin{align*}
    \frac{\partial}{\partial x} \widehat{F}_{Y\mid X}(y \! \mid \! x_0^\pm) = \frac{1}{h} e_1^T \widehat{\alpha}_{\pm,p}(y)
\end{align*}
where $e_1 = (0, 1, 0,  \ldots, 0)^T$. Notably, we can repeat this approach to estimate $\mu^\prime_B(x_0^\pm)$. Specifically, if we use the same local polynomial estimation, now with outcomes $T_i = b(X_i, \eta_i)$, i.e.
\begin{align} \label{beta_equation}
    \widehat{\beta}_{\pm, p} &= \underset{\beta \in \mathbb{R}^{p + 1}}{\text{arg min}} \sum^n_{i=1} \delta_i^\pm \left(T_i - r_p \!\left( \frac{X_i - x_0}{h} \right)^T \beta \right)^2 K \! \left( \frac{X_i - x_0}{h} \right),
\end{align}
then we can similarly obtain the estimator $ \widehat{\mu}^{\prime}_B(x^\pm_0) = \frac{1}{h} e_1^T \widehat{\beta}_{\pm,p}$. Bias correction can similarly be implemented following the same steps discussed in \cref{estimation_inference_section}. Finally, we note that there are many ways one could estimate $f_{Y \mid X}( \cdot \mid x_0)$. One simple method would be to define a local polynomial conditional density estimator by replacing $T_i$ in \cref{beta_equation} with a kernel in $y$, i.e. $h^{-1}_y K \! \left((Y_i-y)/ h_y\right)$. Putting everything together, if we plug-in all of our estimators into \cref{fuzzy_eq} then squaring and numerically integrating over $(0, 1)$ yields an estimate of $\Psi^\prime_\mathcal{C}$.

Inference for the Wasserstein derivative at the policy kink follows largely in the same manner as discussed in \cref{inference_section}; the only major difference is the scaling. As discussed in \cite{calonico_2014, card_rkd_2015}, estimating a derivative at a boundary introduces an additional $1 / h$ scaling, so its variance now scales as $(nh)^{-1} (h^2)^{-1} = (nh^3)^{-1}$. Consequently, if we wanted to construct a confidence interval for $\Psi^\prime$ we simply need to correct this scaling. Following  \cref{conservative_interval}, we can obtain an analogous interval of
\begin{align*}
    C^{ \prime \prime}_n = \left[ (\widehat{\Psi}^\prime_{\mathcal{C}})^2 \pm z_{1-\alpha/2}\sqrt{\widehat{s}^2_n + \frac{c^2}{nh^3}} \, \right]
\end{align*}
where $\widehat{s}_n$ is the estimated standard deviation of $(\widehat{\Psi}^\prime_{\mathcal{C}})^2$, $z_{1-\alpha/2}$ is the $1 - \alpha / 2$ quantile of a standard Normal distribution, and $c$ is some constant, such as $\mathbb{V}(Y)$. 

\subsection{Interpretation of the Kinked Wasserstein Effect}

Interpretation of the Wasserstein derivative at a policy kink follows analogously to the interpretation established in \cref{interpretation_section} for discontinuity designs. To see this, define the quantile effect curve at the kink by
\begin{align*}
    \Delta Q^\prime(u) =  \frac{\frac{\partial}{\partial x}Q_{Y\mid X}(u \! \mid \! x_0^+)-\frac{\partial}{\partial x}Q_{Y\mid X}(u \! \mid \! x_0^-)}
{\mu^\prime_B(x_0^+)-\mu^\prime_B(x_0^-)}.
\end{align*}
Then, again we can immediately see that $\tau^\prime = \int_0^1 \Delta Q^\prime(u)du$ and $(\Psi^\prime)^2 =\int_0^1 [\Delta Q^\prime(u)]^2du$, so letting $U \sim \mathrm{Uniform}(0, 1)$ we again obtain the same variance decomposition
\begin{align*}
    (\Psi^\prime)^2=(\tau^\prime)^2+ \mathbb{V}(\Delta Q^\prime(U)).
\end{align*}
Thus, $\Psi^\prime$ captures both the mean drift through the kink (as measured by $\tau^\prime$) and the heterogeneity of the treatment effect across quantiles (as measured by $\mathbb{V}(\Delta Q^\prime(U))$). Similarly, by applying the Cauchy-Schwarz inequality we can obtain the kink analogue of \cref{effect_inequality}, $|\tau^\prime|\leq \Psi^\prime$. Equality holds if and only if the marginal effect is purely additive, or equivalently that the quantile effect curve is flat, i.e. $\Delta Q^\prime(u)= \delta$ for all $u\in(0,1)$. Finally, it is possible to obtain an analogous version of \cref{series_representation} for the kinked Wasserstein effect. Let $\lambda_k(x) =\int_0^1 Q_{Y\mid X}(u \! \mid \! x) P_{k-1}^*(u)du$ be the conditional $L$-moment evaluated at $x$ and define its one-sided derivatives as
\begin{align*}
    \lambda^\prime_k(x_0^\pm)=\int_0^1 \frac{\partial}{\partial x} \! \left\{ Q_{Y\mid X}(u \! \mid \! x_0^\pm) \right\}P_{k-1}^*(u)du.
\end{align*}
Then, following the same arguments as in the proof of \cref{series_representation}, it can be shown that the Wasserstein derivative at the kink may be decomposed into derivatives of $L$-moments. We formalize this decomposition in the following theorem.
\begin{theorem}[$L$-Moment Derivative Decomposition] \label{wass_deriv_decomp}
    Suppose that $\int^1_0 [\Delta Q^\prime(u)]^2 du < \infty$. Then, 
    \begin{align*}
    \Psi^\prime_{\mathcal{C}} = \left\{\sum_{k=1}^\infty (2k-1) \left( \frac{\lambda^\prime_k(x_0^+)-\lambda^\prime_k(x_0^-)}{\mu^\prime_B(x_0^+)-\mu^\prime_B(x_0^-)} \right)^2 \right\}^{1/2}.
\end{align*}
\end{theorem}
Now, each $k$ in the series representation of $\Psi^\prime_{\mathcal{C}}$ represents the instantaneous change in $L$-location, $L$-scale, $L$-skewness, etc. at the kink. With these representations and interpretations established, in the next section we analyze real data sets to see how these methods can be implemented in practice.

\section{Simulations and Data Analysis} \label{application_section}

In this section we consider the practical implementation of distributional discontinuity designs and distributional kink designs. First, we compare the empirical coverage and interval width of the two confidence intervals proposed in \cref{confidence_intervals_section}. Next, we re-analyze two natural experiments: one regression discontinuity design and one regression kink design. Our goal is to compare traditional mean-based effects to our proposed distributional effects.

\subsection{Simulations}

In what follows we conduct a simulation study to compare the empirical widths and coverage of the two conservative confidence intervals for $\Psi$ described in \cref{confidence_intervals_section}. We consider three data generating processes, all of which feature a running variable drawn $X_i \sim \mathrm{Uniform}(-1,1)$, treatment sharply assigned such that $A_i = \mathbb{I}(X_i \geq 0)$, and the function $m(x) = 0.5x+x^2$:
\begin{enumerate}
    \item[$(i)$] \textit{Additive Effect:} Let $Y_i=m(X_i)+\tau A_i+\varepsilon_i$ with $\varepsilon_i\sim N(0,1)$.
    \item[$(ii)$] \textit{Differing Variances:} Let $Y_i=m(X_i)+\sigma(A_i)\varepsilon_i$ with $\sigma(0)=1$ and $\sigma(1)=2$.
    \item[$(iii)$] \textit{Heavy Tailed:} Let $Y_i =m(X_i) + \tau A_i + (1 + 0.3|X_i|) (1 + 0.6A_i) (\mathrm{exp}(\varepsilon_i) - \mathrm{exp}(1/2))$.
\end{enumerate}

The three data generating processes are chosen to represent increasingly rich forms of treatment heterogeneity. Setting $(i)$ is a simple additive treatment effect model where treatment shifts the conditional distribution by a constant $\tau = 1/2$ at every quantile. Setting $(ii)$ leaves the mean unchanged at the cutoff, but doubles the standard deviation. Finally, in setting $(iii)$ we introduce errors that are skewed and heavy tailed via $
\mathrm{exp}(\varepsilon_i) - \mathrm{exp}(1/2)$; furthermore, we introduce the factors $(1+0.3|X_i|)$ and $(1+0.6A_i)$ to induce heteroskedasticity in the running variable and an explicitly non-additive treatment effect.

In the simulation we consider $n \in \{10^3, 10^4, 10^5, 10^6\}$, which corresponds to one-sided within-bandwidth sample sizes of roughly 185, 1185, 7500, and 47300. We also consider different truncation profiles for $\Psi$, where instead of integrating over the full quantile grid $u \in (0, 1)$ we consider $u \in (\gamma, 1 - \gamma)$ for $\gamma \in \{0.05, 0.1, 0.25\}$. We find that for small to modest sample sizes some degree of truncation is useful for numerical stability. For each replication and choice of $n$, we estimate $\Psi$ using the bias-corrected local polynomial procedure described in \cref{estimation_inference_section}. We use the default bandwidth rule $h_n \propto n^{-1/5}$, 1,000 bootstrap replications, and 10,000 overall Monte-Carlo simulations.

\begin{figure}[h]
\centering

\definecolor{gammaRed}{HTML}{F8766D}   
\definecolor{gammaGreen}{HTML}{00BA38} 
\definecolor{gammaBlue}{HTML}{619CFF}  

\pgfplotsset{
  CnSimple/.style={solid, thick, mark=*, mark options={solid}, mark size=2.0},
  CnBand/.style={dash pattern=on 5pt off 3pt, thick, mark=triangle*, mark options={solid}, mark size=2.2},
  RefLine/.style={black!60, dashed, thick},
}

\begin{tikzpicture}
\begin{groupplot}[
  group style={group size=3 by 2, horizontal sep=1.3cm, vertical sep=1.3cm},
  width=0.34\textwidth,
  height=0.30\textwidth,
  xmode=log, log basis x=10,
  xtick={1000,10000,100000, 1000000},
  xlabel={$n$},
  ymajorgrids=true,
  xmajorgrids=true,
  grid style={line width=0.1pt, draw=gray!20, opacity=0.5},
  xtick pos=bottom,
  ytick pos=left
]

\nextgroupplot[
  title={$(i)$ Additive}, xlabel={},
    ymode=log, log basis y=10,
  ylabel={Interval width},
  legend to name=sharedlegend_w,
  legend columns=3,
  legend cell align=left,
  legend image post style={mark=none} 
]

\addlegendimage{mark=none, thick, color=gammaRed}
\addlegendentry{$\gamma=0.05 \, $ }
\addlegendimage{mark=none, thick, color=gammaGreen}
\addlegendentry{$\gamma=0.10 \, $}
\addlegendimage{mark=none, thick, color=gammaBlue}
\addlegendentry{$\gamma=0.25$}

\addplot[CnSimple, color=gammaRed]
  table[col sep=comma, x="n", y="width"]
  {files/simulations/A_baseline_width_gamma0p05_Cnprime_simple.csv};

\addplot[CnBand, color=gammaRed]
  table[col sep=comma, x="n", y="width"]
  {files/simulations/A_baseline_width_gamma0p05_Cn_band.csv};

\addplot[CnSimple, color=gammaGreen]
  table[col sep=comma, x="n", y="width"]
  {files/simulations/A_baseline_width_gamma0p10_Cnprime_simple.csv};

\addplot[CnBand, color=gammaGreen]
  table[col sep=comma, x="n", y="width"]
  {files/simulations/A_baseline_width_gamma0p10_Cn_band.csv};

\addplot[CnSimple, color=gammaBlue]
  table[col sep=comma, x="n", y="width"]
  {files/simulations/A_baseline_width_gamma0p25_Cnprime_simple.csv};

\addplot[CnBand, color=gammaBlue]
  table[col sep=comma, x="n", y="width"]
  {files/simulations/A_baseline_width_gamma0p25_Cn_band.csv};

\nextgroupplot[title={$(ii)$ Diff. Variances}, ylabel={},
  ymode=log, log basis y=10]

\addplot[CnSimple, color=gammaRed]
  table[col sep=comma, x="n", y="width"]
  {files/simulations/B_scale_width_gamma0p05_Cnprime_simple.csv};

\addplot[CnBand, color=gammaRed]
  table[col sep=comma, x="n", y="width"]
  {files/simulations/B_scale_width_gamma0p05_Cn_band.csv};

\addplot[CnSimple, color=gammaGreen]
  table[col sep=comma, x="n", y="width"]
  {files/simulations/B_scale_width_gamma0p10_Cnprime_simple.csv};

\addplot[CnBand, color=gammaGreen]
  table[col sep=comma, x="n", y="width"]
  {files/simulations/B_scale_width_gamma0p10_Cn_band.csv};

\addplot[CnSimple, color=gammaBlue]
  table[col sep=comma, x="n", y="width"]
  {files/simulations/B_scale_width_gamma0p25_Cnprime_simple.csv};

\addplot[CnBand, color=gammaBlue]
  table[col sep=comma, x="n", y="width"]
  {files/simulations/B_scale_width_gamma0p25_Cn_band.csv};

\nextgroupplot[title={$(iii)$ Heavy}, ylabel={}, xlabel={},
  ymode=log, log basis y=10]

\addplot[CnSimple, color=gammaRed]
  table[col sep=comma, x="n", y="width"]
  {files/simulations/C_heavy_width_gamma0p05_Cnprime_simple.csv};

\addplot[CnBand, color=gammaRed]
  table[col sep=comma, x="n", y="width"]
  {files/simulations/C_heavy_width_gamma0p05_Cn_band.csv};

\addplot[CnSimple, color=gammaGreen]
  table[col sep=comma, x="n", y="width"]
  {files/simulations/C_heavy_width_gamma0p10_Cnprime_simple.csv};

\addplot[CnBand, color=gammaGreen]
  table[col sep=comma, x="n", y="width"]
  {files/simulations/C_heavy_width_gamma0p10_Cn_band.csv};

\addplot[CnSimple, color=gammaBlue]
  table[col sep=comma, x="n", y="width"]
  {files/simulations/C_heavy_width_gamma0p25_Cnprime_simple.csv};

\addplot[CnBand, color=gammaBlue]
  table[col sep=comma, x="n", y="width"]
  {files/simulations/C_heavy_width_gamma0p25_Cn_band.csv};

\nextgroupplot[
  title={}, xlabel={},
    ylabel={Coverage}, ymin=0.8
]

\addplot[RefLine, domain=900:1100000] {0.95};

\addplot[CnSimple, color=gammaRed]
  table[col sep=comma, x="n", y="coverage"]
  {files/simulations/A_baseline_coverage_gamma0p05_Cnprime_simple.csv};
\addplot[CnBand, color=gammaRed]
  table[col sep=comma, x="n", y="coverage"]
  {files/simulations/A_baseline_coverage_gamma0p05_Cn_band.csv};

\addplot[CnSimple, color=gammaGreen]
  table[col sep=comma, x="n", y="coverage"]
  {files/simulations/A_baseline_coverage_gamma0p10_Cnprime_simple.csv};
\addplot[CnBand, color=gammaGreen]
  table[col sep=comma, x="n", y="coverage"]
  {files/simulations/A_baseline_coverage_gamma0p10_Cn_band.csv};

\addplot[CnSimple, color=gammaBlue]
  table[col sep=comma, x="n", y="coverage"]
  {files/simulations/A_baseline_coverage_gamma0p25_Cnprime_simple.csv};
\addplot[CnBand, color=gammaBlue]
  table[col sep=comma, x="n", y="coverage"]
  {files/simulations/A_baseline_coverage_gamma0p25_Cn_band.csv};

\nextgroupplot[title={}, ylabel={}, ymin=0.8,]

\addplot[RefLine, domain=900:1100000] {0.95};

\addplot[CnSimple, color=gammaRed]
  table[col sep=comma, x="n", y="coverage"]
  {files/simulations/B_scale_coverage_gamma0p05_Cnprime_simple.csv};
\addplot[CnBand, color=gammaRed]
  table[col sep=comma, x="n", y="coverage"]
  {files/simulations/B_scale_coverage_gamma0p05_Cn_band.csv};

\addplot[CnSimple, color=gammaGreen]
  table[col sep=comma, x="n", y="coverage"]
  {files/simulations/B_scale_coverage_gamma0p10_Cnprime_simple.csv};
\addplot[CnBand, color=gammaGreen]
  table[col sep=comma, x="n", y="coverage"]
  {files/simulations/B_scale_coverage_gamma0p10_Cn_band.csv};

\addplot[CnSimple, color=gammaBlue]
  table[col sep=comma, x="n", y="coverage"]
  {files/simulations/B_scale_coverage_gamma0p25_Cnprime_simple.csv};
\addplot[CnBand, color=gammaBlue]
  table[col sep=comma, x="n", y="coverage"]
  {files/simulations/B_scale_coverage_gamma0p25_Cn_band.csv};

\nextgroupplot[title={}, ylabel={}, xlabel={}, ymin=0.8]

\addplot[RefLine, domain=900:1100000] {0.95};

\addplot[CnSimple, color=gammaRed]
  table[col sep=comma, x="n", y="coverage"]
  {files/simulations/C_heavy_coverage_gamma0p05_Cnprime_simple.csv};
\addplot[CnBand, color=gammaRed]
  table[col sep=comma, x="n", y="coverage"]
  {files/simulations/C_heavy_coverage_gamma0p05_Cn_band.csv};

\addplot[CnSimple, color=gammaGreen]
  table[col sep=comma, x="n", y="coverage"]
  {files/simulations/C_heavy_coverage_gamma0p10_Cnprime_simple.csv};
\addplot[CnBand, color=gammaGreen]
  table[col sep=comma, x="n", y="coverage"]
  {files/simulations/C_heavy_coverage_gamma0p10_Cn_band.csv};

\addplot[CnSimple, color=gammaBlue]
  table[col sep=comma, x="n", y="coverage"]
  {files/simulations/C_heavy_coverage_gamma0p25_Cnprime_simple.csv};
\addplot[CnBand, color=gammaBlue]
  table[col sep=comma, x="n", y="coverage"]
  {files/simulations/C_heavy_coverage_gamma0p25_Cn_band.csv};

\end{groupplot}

\path (group c1r2.south west) -- (group c3r2.south east)
  node[midway, yshift=-1.5cm] {\pgfplotslegendfromname{sharedlegend_w}};

\end{tikzpicture}

\caption{Monte Carlo confidence interval widths and coverage for bootstrap intervals (dashed) and simple intervals (solid) across trimming levels $\gamma$ and sample sizes $n$.}
\label{simulation_results}
\end{figure}
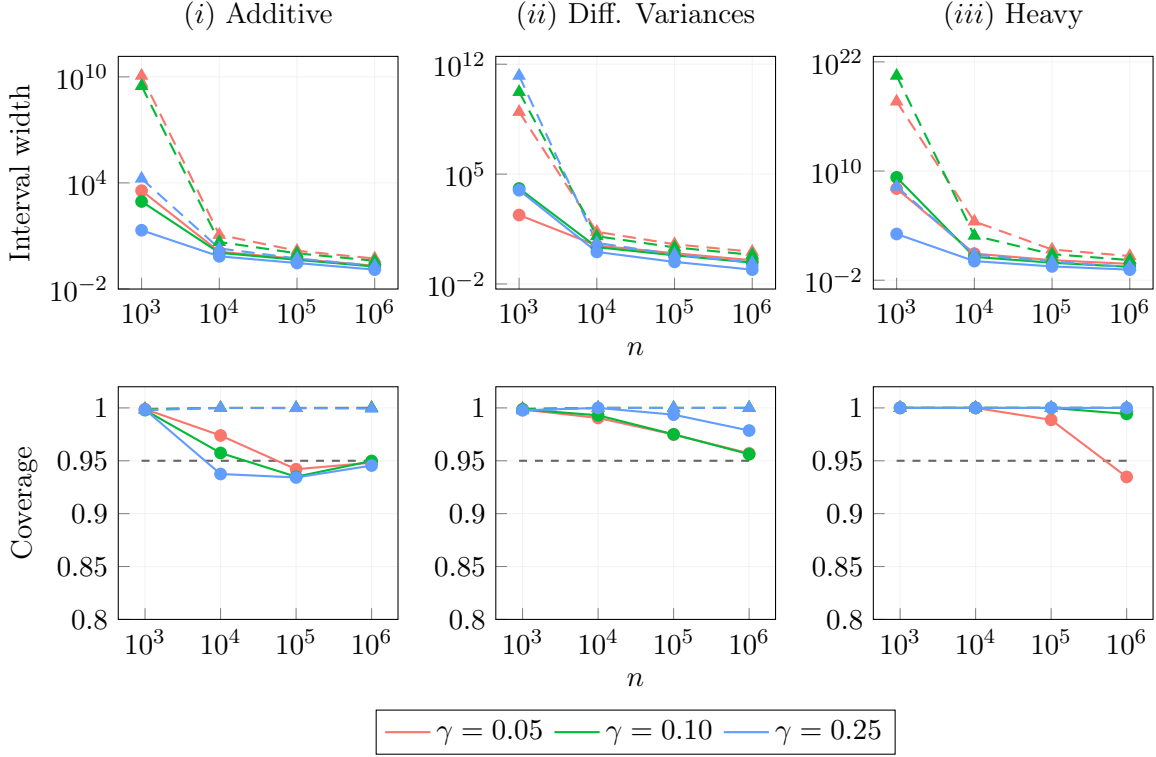

Our simulation results are visualized in \cref{simulation_results}. Broadly speaking, we find that for small sample sizes, the conservative intervals defined in \cref{conservative_interval} are an order of magnitude smaller than the bootstrap intervals; as the sample size increases this difference becomes less pronounced, but does not go away. This is likely because, as outlined in \cite{CHIANG2019589}, the bootstrap intervals estimate both the running variable density at the cutoff as well as conditional outcome densities evaluated at estimated quantiles; both can be unstable with small sample sizes. Furthermore, because \cite{CHIANG2019589} define a uniform confidence band over $u$, if any quantiles are poorly estimated then the bands can blow up in width. Furthermore, we can see that while coverage is theoretically conservative for both methods, as the sample size increased, the conservative intervals attained approximate $1 - \alpha$ coverage; meanwhile, the bootstrap intervals always overcovered.

\subsection{Distributional Discontinuity Design Analysis}

In this section we re-analyze a canonical regression discontinuity design analysis in order to compare the Wasserstein effect to the conventional mean effect at the cutoff. We consider the work of \cite{lee2008}, who studied the causal effect of electoral incumbency in U.S. house elections by leveraging the idea that elections decided by very small margins are ``as good as randomized;'' their data is publicly available via the \texttt{R} package \texttt{RDHonest}.  The running variable is the Democratic vote share margin of victory in a given election (defined by the Democratic vote share minus the vote share of the strongest opponent), with a corresponding discontinuity at zero. The primary outcome is the Democratic party's vote share in the subsequent election. Empirically, \cite{lee2008} finds clear evidence of an incumbency advantage, where barely winning an election leads to a statistically significant jump in the following election, to the tune of a 7-8\% increase in vote share. In what follows, we consider whether or not there were interesting distributional effects not visible by considering the mean alone. 

In our re-analysis, to keep things simple we choose $p=1$ for our local polynomial estimator, a triangular kernel, and we set $h = n^{-1/5}$. Under these settings, using the \texttt{rdrobust} package we estimate the mean effect to be a 7.099 increase in vote share with a 95\% confidence interval of [2.648 , 8.038] --- these results replicate the findings of \cite{lee2008}. Next, we consider the Wasserstein effect without truncating quantiles: we obtain an estimate for $\Psi$ of 7.544 with a 95\% confidence interval of [5.023, 9.412]. The fact that $\widehat{\tau}$ and $\widehat{\Psi}$ are so close to each other suggests there was not much heterogeneity in the treatment effect. Indeed, if we break down $\widehat{\Psi}$ into the distributional $R^2$ table as shown in \cref{comparison_table_lee08}, 
\begin{table}[h]
\centering
\begin{tabular}{r|c}
\textit{Moment} & \textit{Explained Distance}\\ \hline
$k = 1$ & 0.5598  \\
$k = 2$ & 0.0413  \\
$k = 3$ & 0.1118 \\
$k \geq 4$ & 0.2871  \\
\end{tabular}
\caption{Explained distributional variation for the incumbency advantage}
\label{comparison_table_lee08}
\end{table}

we can see that most of the effect is explained by the variation in the $L$-location, with some notable changes as well in $L$-skewness and higher-order decompositions. Note that as shown in \cref{rk_equation}, because $\lambda_1$ is the mean, it follows that
\begin{align*}
    R^2_1 = \frac{\tau^2}{\Psi^2} = 1 - \gamma,
\end{align*}
where $\gamma$ is the heterogeneity index discussed in \cref{ate_comparison_sec}. Thus, if we were to plug in each estimate, we'd find $7.099^2 / 7.544^2 \approx 0.886$ as the explained distributional distance coming from the first moment. The gap between 0.5598 and 0.886 is likely due to finite-sample differences in mean vs quantile effect estimation. Our findings are further validated by considering the estimated Wasserstein dominance. Here, we find $\widehat{\rho} = 0.5777$, suggesting that winning a close election was pretty uniformly beneficial, with little quantile crossing. By combining the traditional mean effects analysis with our distributional analysis, we were able to obtain a much more complete understanding of the causal effect of the incumbency advantage.

\subsection{Distributional Kink Design}

In this section we re-analyze an existing regression kink design analysis in order to compare the Wasserstein derivative to the mean-effect at the kink. Specifically, we consider the work of \cite{Lundqvist_2014} who study whether general intergovernmental grants increase local public employment using known kinks in the Swedish grant system; the data used is publicly available via the \texttt{R} package \texttt{causalweight}. The running variable is the net out-migration rate in a given municipality, $m_{it}=100(1-n_{i,t-2}/n_{i,t-12})$ where $n_{i, t}$ is the population in the $i$th municipality at time $t$. That is to say, the percentage population decrease over a ten-year window with a two-year lag. The policy rule for out-migration grants is given by
\begin{align*}
   g^{m}_{it}=
\begin{cases}
a \, (m_{it}-2), & m_{it}>2,\\
0, & m_{it}\le 2,
\end{cases} 
\end{align*}
where the kink is at $2\%$ and $a$ is a constant (100 Swedish krona per capita per additional percentage point above $2\%$). In their analysis \cite{Lundqvist_2014} find no statistically significant effect of grants on total local public employment, making their study a good point of comparison to distributional effects that consider more than just the mean. 

In our re-analysis, we consider $h \in \{5, 10, 15\}$ and a uniform kernel to match the analysis of \cite{Lundqvist_2014}; we report the $h = 10$ results, although they are all qualitatively similar. Furthermore, we demean the outcome and benefit by year and cluster our standard errors at the municipality. Using the \texttt{rdrobust} package we estimate the local average slope $\widehat{\tau}^\prime_{\mathcal{C}}$ to be -0.050, with a 95\% confidence interval of $[-0.378 , 0.277]$, matching the null effect found in \cite{Lundqvist_2014}. Our estimated value for the Wasserstein derivative $\widehat{\Psi}^\prime_{\mathcal{C}}$ is 0.6713, with a 95\% confidence interval of $[0.000, 1.432]$, indicating a null distributional effect. However, we do find an interesting characterization of the effect in the $L$-moment decomposition, as shown in \cref{comparison_table_lundqvist}. It appears that most of the distance explained in $\widehat{\Psi}^\prime_{\mathcal{C}}$ comes from higher-order moments; notably, almost none comes from the mean effect. This may suggest there is more to the story worth looking into: perhaps there were a few outlier municipalities that used their grants extremely well (or poorly). Practitioners may then consider looking into targeted hypothesis tests on specific $L$-moments to further explore whether an effect exists at these levels.

\begin{table}[h]
\centering
\begin{tabular}{r|c}
\textit{Moment} & \textit{Explained Distance}\\ \hline
$k = 1$ & 0.0007  \\
$k = 2$ & 0.1750  \\
$k = 3$ & 0.1362 \\
$k \geq 4$ & 0.6881  \\
\end{tabular}
\caption{Explained distributional variation for the grant effect.}
\label{comparison_table_lundqvist}
\end{table}

\section{Discussion and Conclusion} \label{discussion_section}

In this paper we introduced distributional discontinuity designs and distributional kink designs, a framework for studying distributional causal effects for a scalar outcome at the boundary of a discontinuity or kink in treatment assignment. A key practical motivation for this approach is that many applied regression discontinuity and kink analyses remain centered on mean effects, despite the fact that distributional changes are often of substantive interest. However, it is not our intention to replace these classical tools; rather, we show that distributional causal effects play a complementary role. The Wasserstein effect establishes a natural reference point for both mean and quantiles effects. Since we show that $\Psi$ weakly upper bounds the magnitude of the average treatment effect, practitioners now have an interpretable index of treatment effect heterogeneity whenever $\Psi$ is meaningfully larger than $|\tau|$. Furthermore, we show that the Wasserstein distance admits an orthogonal decomposition into squared differences in $L$-moments. In practice, this decomposition provides a principled way to answer questions like ``is the effect mostly a shift in the distribution means, or is it driven by changes in dispersion, asymmetry, or tail behavior?'' 

Although this work primarily focuses on regression discontinuity and kink designs, the principles outlined here extend beyond these specific applications. One could easily estimate the Wasserstein effect in a randomized controlled trial, for example, or any setting where the exchangeability assumption holds. The exact same interpretations and decompositions described in \cref{interpretation_section} will still hold, all that would change is the way quantile effects are estimated. Thus, we see this work as opening the door toward considering distributional distances as interesting causal effects in their own right. It would be interesting to extend this distributional analysis more broadly to other quasi-experimental designs, such the difference-in-differences framework.

Finally, we note that although this work establishes a novel framework for distributional causal effects at a treatment discontinuity, we only consider univariate outcomes. Future work could establish methods for estimating the Wasserstein distance between multivariate outcome distributions, where identification and estimation no longer reduces to a distance between quantile functions. Furthermore, it would be interesting to consider distributional distances beyond the Wasserstein distance; perhaps other distributional measures capture other underlying phenomena in the data, and admit their own useful decompositions.

\section*{Acknowledgments}

This paper is a product of the Iowa Agriculture and Home Economics Experiment Station, Ames, Iowa. Project No. IOW03717 is supported by USDA/NIFA and State of Iowa funds. Any opinions, findings, conclusions, or recommendations expressed in this publication are those of the authors and do not necessarily reflect the views of the U.S. Department of Agriculture. The authors would like to sincerely thank Emileigh Harrison and Terence Chau for many helpful discussions and Zhaohua Zeng for writing \texttt{R} code to estimate quantile treatment effects.

\section*{References}
\vspace{-1cm}
\bibliographystyle{abbrvnat}
\bibliography{references}

\begin{thebibliography}{56}
\providecommand{\natexlab}[1]{#1}
\providecommand{\url}[1]{\texttt{#1}}
\expandafter\ifx\csname urlstyle\endcsname\relax
  \providecommand{\doi}[1]{doi: #1}\else
  \providecommand{\doi}{doi: \begingroup \urlstyle{rm}\Url}\fi

\bibitem[Ambrosio et~al.(2005)Ambrosio, Gigli, and {Savar\'e}]{AmbrosioGigliSavare2005GradientFlows}
L.~Ambrosio, N.~Gigli, and G.~{Savar\'e}.
\newblock \emph{Gradient Flows}.
\newblock Lectures in Mathematics. ETH Z{\"u}rich. Birkh{\"a}user Basel, 1 edition, 2005.
\newblock ISBN 978-3-7643-7309-2.
\newblock \doi{10.1007/b137080}.

\bibitem[Ando(2017)]{ando2017trust}
M.~Ando.
\newblock How much should we trust regression-kink-design estimates?
\newblock \emph{Empirical Economics}, 53\penalty0 (3):\penalty0 1287--1322, November 2017.
\newblock \doi{10.1007/s00181-016-1155-8}.
\newblock URL \url{https://ideas.repec.org/a/spr/empeco/v53y2017i3d10.1007_s00181-016-1155-8.html}.

\bibitem[Angrist et~al.(1996)Angrist, Imbens, and Rubin]{Angrist01061996}
J.~D. Angrist, G.~W. Imbens, and D.~B. Rubin.
\newblock Identification of causal effects using instrumental variables.
\newblock \emph{Journal of the American Statistical Association}, 91\penalty0 (434):\penalty0 444--455, 1996.
\newblock \doi{10.1080/01621459.1996.10476902}.
\newblock URL \url{https://www.tandfonline.com/doi/abs/10.1080/01621459.1996.10476902}.

\bibitem[Calonico et~al.(2014)Calonico, Cattaneo, and Titiunik]{calonico_2014}
S.~Calonico, M.~D. Cattaneo, and R.~Titiunik.
\newblock Robust nonparametric confidence intervals for regression-discontinuity designs.
\newblock \emph{Econometrica}, 82\penalty0 (6):\penalty0 2295--2326, 2014.
\newblock \doi{https://doi.org/10.3982/ECTA11757}.
\newblock URL \url{https://onlinelibrary.wiley.com/doi/abs/10.3982/ECTA11757}.

\bibitem[Calonico et~al.(2019{\natexlab{a}})Calonico, Cattaneo, and Farrell]{calonico_2019}
S.~Calonico, M.~D. Cattaneo, and M.~H. Farrell.
\newblock Optimal bandwidth choice for robust bias-corrected inference in regression discontinuity designs.
\newblock \emph{The Econometrics Journal}, 23\penalty0 (2):\penalty0 192--210, 11 2019{\natexlab{a}}.
\newblock ISSN 1368-4221.
\newblock \doi{10.1093/ectj/utz022}.
\newblock URL \url{https://doi.org/10.1093/ectj/utz022}.

\bibitem[Calonico et~al.(2019{\natexlab{b}})Calonico, Cattaneo, Farrell, and Titiunik]{calonico_covariates_2019}
S.~Calonico, M.~D. Cattaneo, M.~H. Farrell, and R.~Titiunik.
\newblock Regression discontinuity designs using covariates.
\newblock \emph{The Review of Economics and Statistics}, 101\penalty0 (3):\penalty0 442--451, July 2019{\natexlab{b}}.
\newblock \doi{None}.
\newblock URL \url{https://ideas.repec.org/a/tpr/restat/v101y2019i3p442-451.html}.

\bibitem[Card et~al.(2015)Card, Lee, Pei, and Weber]{card_rkd_2015}
D.~Card, D.~S. Lee, Z.~Pei, and A.~Weber.
\newblock Inference on causal effects in a generalized regression kink design.
\newblock \emph{Econometrica}, 83\penalty0 (6):\penalty0 2453--2483, 2015.
\newblock \doi{https://doi.org/10.3982/ECTA11224}.
\newblock URL \url{https://onlinelibrary.wiley.com/doi/abs/10.3982/ECTA11224}.

\bibitem[Card et~al.(2016)Card, Lee, Pei, and Weber]{card_theory_practice_2016}
D.~Card, D.~S. Lee, Z.~Pei, and A.~Weber.
\newblock Regression kink design: Theory and practice.
\newblock Working Paper 22781, National Bureau of Economic Research, October 2016.
\newblock URL \url{http://www.nber.org/papers/w22781}.

\bibitem[Cattaneo and Titiunik(2022)]{Cattaneo2022}
M.~D. Cattaneo and R.~Titiunik.
\newblock Regression discontinuity designs.
\newblock \emph{Annual Review of Economics}, 14\penalty0 (Volume 14, 2022):\penalty0 821--851, 2022.
\newblock ISSN 1941-1391.
\newblock \doi{https://doi.org/10.1146/annurev-economics-051520-021409}.
\newblock URL \url{https://www.annualreviews.org/content/journals/10.1146/annurev-economics-051520-021409}.

\bibitem[Cattaneo et~al.(2020)Cattaneo, Jansson, and Ma]{Cattaneo02072020}
M.~D. Cattaneo, M.~Jansson, and X.~Ma.
\newblock Simple local polynomial density estimators.
\newblock \emph{Journal of the American Statistical Association}, 115\penalty0 (531):\penalty0 1449--1455, 2020.
\newblock \doi{10.1080/01621459.2019.1635480}.
\newblock URL \url{https://doi.org/10.1080/01621459.2019.1635480}.

\bibitem[Chen et~al.(2020)Chen, Chiang, and Sasaki]{Chen_Chiang_Sasaki_2020}
H.~Chen, H.~D. Chiang, and Y.~Sasaki.
\newblock Quantile treatment effects in regression kink designs.
\newblock \emph{Econometric Theory}, 36\penalty0 (6):\penalty0 1167–1191, 2020.
\newblock \doi{10.1017/S0266466619000409}.

\bibitem[Chiang and Sasaki(2019)]{CHIANG2019405}
H.~D. Chiang and Y.~Sasaki.
\newblock Causal inference by quantile regression kink designs.
\newblock \emph{Journal of Econometrics}, 210\penalty0 (2):\penalty0 405--433, 2019.
\newblock ISSN 0304-4076.
\newblock \doi{https://doi.org/10.1016/j.jeconom.2019.02.005}.
\newblock URL \url{https://www.sciencedirect.com/science/article/pii/S0304407619300387}.

\bibitem[Chiang et~al.(2019)Chiang, Hsu, and Sasaki]{CHIANG2019589}
H.~D. Chiang, Y.-C. Hsu, and Y.~Sasaki.
\newblock Robust uniform inference for quantile treatment effects in regression discontinuity designs.
\newblock \emph{Journal of Econometrics}, 211\penalty0 (2):\penalty0 589--618, 2019.
\newblock ISSN 0304-4076.
\newblock \doi{https://doi.org/10.1016/j.jeconom.2019.03.006}.
\newblock URL \url{https://www.sciencedirect.com/science/article/pii/S0304407619300569}.

\bibitem[Conway(1990)]{Conway1990CourseFunctionalAnalysis}
J.~B. Conway.
\newblock \emph{A Course in Functional Analysis}, volume~96 of \emph{Graduate Texts in Mathematics}.
\newblock Springer, New York, NY, 2 edition, 1990.
\newblock ISBN 978-0-387-97245-9.
\newblock \doi{10.1007/978-1-4757-4383-8}.

\bibitem[Cook(2008)]{COOK2008636}
T.~D. Cook.
\newblock “waiting for life to arrive”: A history of the regression-discontinuity design in psychology, statistics and economics.
\newblock \emph{Journal of Econometrics}, 142\penalty0 (2):\penalty0 636--654, 2008.
\newblock ISSN 0304-4076.
\newblock \doi{https://doi.org/10.1016/j.jeconom.2007.05.002}.
\newblock URL \url{https://www.sciencedirect.com/science/article/pii/S0304407607001108}.
\newblock The regression discontinuity design: Theory and applications.

\bibitem[Dahlberg et~al.(2008)Dahlberg, Mörk, Rattsø, and Ågren]{DAHLBERG20082320}
M.~Dahlberg, E.~Mörk, J.~Rattsø, and H.~Ågren.
\newblock Using a discontinuous grant rule to identify the effect of grants on local taxes and spending.
\newblock \emph{Journal of Public Economics}, 92\penalty0 (12):\penalty0 2320--2335, 2008.
\newblock ISSN 0047-2727.
\newblock \doi{https://doi.org/10.1016/j.jpubeco.2007.05.004}.
\newblock URL \url{https://www.sciencedirect.com/science/article/pii/S0047272707000886}.
\newblock New Directions in Fiscal Federalism.

\bibitem[Dijcke(2025)]{vandijcke2025regressiondiscontinuitydesigndistributionvalued}
D.~V. Dijcke.
\newblock Regression discontinuity design with distribution-valued outcomes, 2025.
\newblock URL \url{https://arxiv.org/abs/2504.03992}.

\bibitem[Frandsen et~al.(2012)Frandsen, Frölich, and Melly]{FRANDSEN2012382}
B.~R. Frandsen, M.~Frölich, and B.~Melly.
\newblock Quantile treatment effects in the regression discontinuity design.
\newblock \emph{Journal of Econometrics}, 168\penalty0 (2):\penalty0 382--395, 2012.
\newblock ISSN 0304-4076.
\newblock \doi{https://doi.org/10.1016/j.jeconom.2012.02.004}.
\newblock URL \url{https://www.sciencedirect.com/science/article/pii/S0304407612000607}.

\bibitem[Frölich and Huber(2019)]{Frolich02102019}
M.~Frölich and M.~Huber.
\newblock Including covariates in the regression discontinuity design.
\newblock \emph{Journal of Business \& Economic Statistics}, 37\penalty0 (4):\penalty0 736--748, 2019.
\newblock \doi{10.1080/07350015.2017.1421544}.
\newblock URL \url{https://doi.org/10.1080/07350015.2017.1421544}.

\bibitem[Ganong and Jäger(2018)]{Ganong03042018}
P.~Ganong and S.~Jäger.
\newblock A permutation test for the regression kink design.
\newblock \emph{Journal of the American Statistical Association}, 113\penalty0 (522):\penalty0 494--504, 2018.
\newblock \doi{10.1080/01621459.2017.1328356}.
\newblock URL \url{https://doi.org/10.1080/01621459.2017.1328356}.

\bibitem[Gretton et~al.(2012)Gretton, Borgwardt, Rasch, Sch{{\"o}}lkopf, and Smola]{gretton12a}
A.~Gretton, K.~M. Borgwardt, M.~J. Rasch, B.~Sch{{\"o}}lkopf, and A.~Smola.
\newblock A kernel two-sample test.
\newblock \emph{Journal of Machine Learning Research}, 13\penalty0 (25):\penalty0 723--773, 2012.
\newblock URL \url{http://jmlr.org/papers/v13/gretton12a.html}.

\bibitem[Gunsilius(2023)]{gunsilius_synthetic_controls2023}
F.~F. Gunsilius.
\newblock Distributional synthetic controls.
\newblock \emph{Econometrica}, 91\penalty0 (3):\penalty0 1105--1117, 2023.
\newblock \doi{https://doi.org/10.3982/ECTA18260}.
\newblock URL \url{https://onlinelibrary.wiley.com/doi/abs/10.3982/ECTA18260}.

\bibitem[Gunsilius(2025)]{gunsilius2025primeroptimaltransportcausal}
F.~F. Gunsilius.
\newblock A primer on optimal transport for causal inference with observational data, 2025.
\newblock URL \url{https://arxiv.org/abs/2503.07811}.

\bibitem[Guryan(2001)]{NBERw8269}
J.~Guryan.
\newblock Does money matter? regression-discontinuity estimates from education finance reform in massachusetts.
\newblock Working Paper 8269, National Bureau of Economic Research, May 2001.
\newblock URL \url{http://www.nber.org/papers/w8269}.

\bibitem[Hahn et~al.(2001)Hahn, Todd, and der Klaauw]{hahn2001}
J.~Hahn, P.~Todd, and W.~V. der Klaauw.
\newblock Identification and estimation of treatment effects with a regression-discontinuity design.
\newblock \emph{Econometrica}, 69\penalty0 (1):\penalty0 201--209, 2001.
\newblock ISSN 00129682, 14680262.
\newblock URL \url{http://www.jstor.org/stable/2692190}.

\bibitem[Hoderlein and Mammen(2007)]{hoderlein2007}
S.~Hoderlein and E.~Mammen.
\newblock Identification of marginal effects in nonseparable models without monotonicity.
\newblock \emph{Econometrica}, 75\penalty0 (5):\penalty0 1513--1518, 2007.
\newblock \doi{https://doi.org/10.1111/j.1468-0262.2007.00801.x}.
\newblock URL \url{https://onlinelibrary.wiley.com/doi/abs/10.1111/j.1468-0262.2007.00801.x}.

\bibitem[Hoderlein and Mammen(2009)]{hoderlein2009}
S.~Hoderlein and E.~Mammen.
\newblock Identification and estimation of local average derivatives in non-separable models without monotonicity.
\newblock \emph{The Econometrics Journal}, 12\penalty0 (1):\penalty0 1--25, 2009.
\newblock \doi{https://doi.org/10.1111/j.1368-423X.2008.00273.x}.
\newblock URL \url{https://onlinelibrary.wiley.com/doi/abs/10.1111/j.1368-423X.2008.00273.x}.

\bibitem[Hosking(1990)]{hosking1990}
J.~R.~M. Hosking.
\newblock L-moments: Analysis and estimation of distributions using linear combinations of order statistics.
\newblock \emph{Journal of the Royal Statistical Society. Series B (Methodological)}, 52\penalty0 (1):\penalty0 105--124, 1990.
\newblock ISSN 00359246.
\newblock URL \url{http://www.jstor.org/stable/2345653}.

\bibitem[Imbens and Kalyanaraman(2012)]{imbens_kalyanaraman_2012}
G.~Imbens and K.~Kalyanaraman.
\newblock Optimal bandwidth choice for the regression discontinuity estimator.
\newblock \emph{The Review of Economic Studies}, 79\penalty0 (3):\penalty0 933--959, 2012.
\newblock ISSN 00346527, 1467937X.
\newblock URL \url{http://www.jstor.org/stable/23261375}.

\bibitem[Imbens and Lemieux(2008)]{IMBENS2008615}
G.~W. Imbens and T.~Lemieux.
\newblock Regression discontinuity designs: A guide to practice.
\newblock \emph{Journal of Econometrics}, 142\penalty0 (2):\penalty0 615--635, 2008.
\newblock ISSN 0304-4076.
\newblock \doi{https://doi.org/10.1016/j.jeconom.2007.05.001}.
\newblock URL \url{https://www.sciencedirect.com/science/article/pii/S0304407607001091}.
\newblock The regression discontinuity design: Theory and applications.

\bibitem[Janson(1997)]{Janson_1997}
S.~Janson.
\newblock \emph{Wiener chaos}, page 17–22.
\newblock Cambridge Tracts in Mathematics. Cambridge University Press, 1997.

\bibitem[Jin et~al.(2025)Jin, Zhang, Zhang, and Zhou]{Jin_Zhang_Zhang_Zhou_2025}
Z.~Jin, Y.~Zhang, Z.~Zhang, and Y.~Zhou.
\newblock Identification and inference in a quantile regression discontinuity design under rank similarity with covariates.
\newblock \emph{Econometric Theory}, 41\penalty0 (1):\penalty0 172–217, 2025.
\newblock \doi{10.1017/S026646662300021X}.

\bibitem[Karhunen(1946)]{Karhunen1946ZurSS}
K.~Karhunen.
\newblock Zur spektraltheorie stochastischer prozesse.
\newblock 1946.
\newblock URL \url{https://api.semanticscholar.org/CorpusID:118738283}.

\bibitem[Kim et~al.(2024)Kim, Kim, and Kennedy]{kim2024causaleffectsbaseddistributional}
K.~Kim, J.~Kim, and E.~H. Kennedy.
\newblock Causal effects based on distributional distances, 2024.
\newblock URL \url{https://arxiv.org/abs/1806.02935}.

\bibitem[Kurisu et~al.(2025)Kurisu, Zhou, Otsu, and Müller]{kurisu2025geodesiccausalinference}
D.~Kurisu, Y.~Zhou, T.~Otsu, and H.-G. Müller.
\newblock Geodesic causal inference, 2025.
\newblock URL \url{https://arxiv.org/abs/2406.19604}.

\bibitem[Lee(2008)]{lee2008}
D.~S. Lee.
\newblock Randomized experiments from non-random selection in u.s. house elections.
\newblock \emph{Journal of Econometrics}, 142\penalty0 (2):\penalty0 675--697, February 2008.
\newblock \doi{None}.
\newblock URL \url{https://ideas.repec.org/a/eee/econom/v142y2008i2p675-697.html}.

\bibitem[Lee and Lemieux(2010)]{leelemieux2010}
D.~S. Lee and T.~Lemieux.
\newblock Regression discontinuity designs in economics.
\newblock \emph{Journal of Economic Literature}, 48\penalty0 (2):\penalty0 281–355, June 2010.
\newblock \doi{10.1257/jel.48.2.281}.
\newblock URL \url{https://www.aeaweb.org/articles?id=10.1257/jel.48.2.281}.

\bibitem[Lo\`eve(1977)]{loeve1977probability}
M.~Lo\`eve.
\newblock \emph{Probability Theory I}, volume~45 of \emph{Graduate Texts in Mathematics}.
\newblock Springer, New York, NY, 4 edition, 1977.
\newblock ISBN 978-1-4684-9464-8.
\newblock \doi{10.1007/978-1-4684-9464-8}.

\bibitem[Luedtke et~al.(2018)Luedtke, Carone, and van~der Laan]{luedtke_omnibus_2018}
A.~Luedtke, M.~Carone, and M.~J. van~der Laan.
\newblock An omnibus non-parametric test of equality in distribution for unknown functions.
\newblock \emph{Journal of the Royal Statistical Society Series B: Statistical Methodology}, 81\penalty0 (1):\penalty0 75--99, 11 2018.
\newblock ISSN 1369-7412.
\newblock \doi{10.1111/rssb.12299}.
\newblock URL \url{https://doi.org/10.1111/rssb.12299}.

\bibitem[Lundqvist et~al.(2014)Lundqvist, Dahlberg, and Mörk]{Lundqvist_2014}
H.~Lundqvist, M.~Dahlberg, and E.~Mörk.
\newblock Stimulating local public employment: Do general grants work?
\newblock \emph{American Economic Journal: Economic Policy}, 6\penalty0 (1):\penalty0 167--192, 2014.
\newblock ISSN 19457731, 1945774X.
\newblock URL \url{http://www.jstor.org/stable/43189370}.

\bibitem[McCrary(2008)]{MCCRARY2008698}
J.~McCrary.
\newblock Manipulation of the running variable in the regression discontinuity design: A density test.
\newblock \emph{Journal of Econometrics}, 142\penalty0 (2):\penalty0 698--714, 2008.
\newblock ISSN 0304-4076.
\newblock \doi{https://doi.org/10.1016/j.jeconom.2007.05.005}.
\newblock URL \url{https://www.sciencedirect.com/science/article/pii/S0304407607001133}.
\newblock The regression discontinuity design: Theory and applications.

\bibitem[Nielsen et~al.(2010)Nielsen, S{\o}rensen, and Taber]{nielsen2010estimating}
H.~S. Nielsen, T.~S{\o}rensen, and C.~Taber.
\newblock Estimating the effect of student aid on college enrollment: Evidence from a government grant policy reform.
\newblock \emph{American Economic Journal: Economic Policy}, 2\penalty0 (2):\penalty0 185--215, 2010.

\bibitem[Qu and Yoon(2015)]{QU20151}
Z.~Qu and J.~Yoon.
\newblock Nonparametric estimation and inference on conditional quantile processes.
\newblock \emph{Journal of Econometrics}, 185\penalty0 (1):\penalty0 1--19, 2015.
\newblock ISSN 0304-4076.
\newblock \doi{https://doi.org/10.1016/j.jeconom.2014.10.008}.
\newblock URL \url{https://www.sciencedirect.com/science/article/pii/S0304407614002462}.

\bibitem[Qu and Yoon(2019)]{Qu02102019}
Z.~Qu and J.~Yoon.
\newblock Uniform inference on quantile effects under sharp regression discontinuity designs.
\newblock \emph{Journal of Business \& Economic Statistics}, 37\penalty0 (4):\penalty0 625--647, 2019.
\newblock \doi{10.1080/07350015.2017.1407323}.
\newblock URL \url{https://doi.org/10.1080/07350015.2017.1407323}.

\bibitem[Schindl and Wasserman(2025)]{schindl2025causalgeodesycounterfactualestimation}
K.~Schindl and L.~Wasserman.
\newblock Causal geodesy: Counterfactual estimation along the path between correlation and causation, 2025.
\newblock URL \url{https://arxiv.org/abs/2508.08499}.

\bibitem[Sejdinovic et~al.(2013)Sejdinovic, Sriperumbudur, Gretton, and Fukumizu]{Sejdinovic_RKHS_test_2013}
D.~Sejdinovic, B.~Sriperumbudur, A.~Gretton, and K.~Fukumizu.
\newblock {Equivalence of distance-based and RKHS-based statistics in hypothesis testing}.
\newblock \emph{The Annals of Statistics}, 41\penalty0 (5):\penalty0 2263 -- 2291, 2013.
\newblock \doi{10.1214/13-AOS1140}.
\newblock URL \url{https://doi.org/10.1214/13-AOS1140}.

\bibitem[Sillitto(1969)]{sillitto69}
G.~P. Sillitto.
\newblock Derivation of approximants to the inverse distribution function of a continuous univariate population from the order statistics of a sample.
\newblock \emph{Biometrika}, 56\penalty0 (3):\penalty0 641--650, 12 1969.
\newblock ISSN 0006-3444.
\newblock \doi{10.1093/biomet/56.3.641}.
\newblock URL \url{https://doi.org/10.1093/biomet/56.3.641}.

\bibitem[Thistlethwaite and Campbell(1960)]{thistlethwaite1960regression}
D.~L. Thistlethwaite and D.~T. Campbell.
\newblock Regression-discontinuity analysis: An alternative to the ex post facto experiment.
\newblock \emph{Journal of Educational psychology}, 51\penalty0 (6):\penalty0 309, 1960.

\bibitem[Torous et~al.(2024)Torous, Gunsilius, and Rigollet]{TorousGunsiliusRigollet+2024}
W.~Torous, F.~Gunsilius, and P.~Rigollet.
\newblock An optimal transport approach to estimating causal effects via nonlinear difference-in-differences.
\newblock \emph{Journal of Causal Inference}, 12\penalty0 (1):\penalty0 20230004, 2024.
\newblock \doi{doi:10.1515/jci-2023-0004}.
\newblock URL \url{https://doi.org/10.1515/jci-2023-0004}.

\bibitem[Vaart(1998)]{Vaart_1998}
A.~W. v.~d. Vaart.
\newblock \emph{Asymptotic Statistics}.
\newblock Cambridge Series in Statistical and Probabilistic Mathematics. Cambridge University Press, 1998.

\bibitem[Vallender(1974)]{vallender1974}
S.~S. Vallender.
\newblock Calculation of the wasserstein distance between probability distributions on the line.
\newblock \emph{Theory of Probability \& Its Applications}, 18\penalty0 (4):\penalty0 784--786, 1974.
\newblock \doi{10.1137/1118101}.
\newblock URL \url{https://doi.org/10.1137/1118101}.

\bibitem[Verdinelli and Wasserman(2024)]{verdinelli_wasserman_dvi_2024}
I.~Verdinelli and L.~Wasserman.
\newblock Decorrelated variable importance.
\newblock \emph{Journal of Machine Learning Research}, 25\penalty0 (7):\penalty0 1--27, 2024.
\newblock URL \url{http://jmlr.org/papers/v25/22-0801.html}.

\bibitem[Villani et~al.(2009)]{villani2009optimal}
C.~Villani et~al.
\newblock \emph{Optimal transport: old and new}, volume 338.
\newblock Springer, 2009.

\bibitem[Wang and Zhang(2025)]{wang2025unifiedframeworkidentificationinference}
Z.~Wang and Z.~Zhang.
\newblock A unified framework for identification and inference of local treatment effects in sharp regression kink designs, 2025.
\newblock URL \url{https://arxiv.org/abs/2506.11663}.

\bibitem[Yu and Jones(1998)]{Yu01031998}
K.~Yu and M.~C. Jones.
\newblock Local linear quantile regression.
\newblock \emph{Journal of the American Statistical Association}, 93\penalty0 (441):\penalty0 228--237, 1998.
\newblock \doi{10.1080/01621459.1998.10474104}.
\newblock URL \url{https://www.tandfonline.com/doi/abs/10.1080/01621459.1998.10474104}.

\bibitem[Zhou et~al.(2025)Zhou, Kurisu, Otsu, and Müller]{zhou2025geodesicdifferenceindifferences}
Y.~Zhou, D.~Kurisu, T.~Otsu, and H.-G. Müller.
\newblock Geodesic difference-in-differences, 2025.
\newblock URL \url{https://arxiv.org/abs/2501.17436}.

\end{thebibliography}

\setlength{\parindent}{0cm}
\appendix

\begin{center}
{\large\bf SUPPLEMENTARY MATERIAL}
\end{center}

\begin{description}

\item \cref{sec:proofs}: Contains all proofs from the main text and supplementary material, including:
\begin{description}
    \item \cref{effect_inequality_proof}: Proof of \cref{effect_inequality}.
    \item \cref{series_representation_proof}: Proof of \cref{series_representation}.
    \item \cref{eigenvalue_test_proof}: Proof of \cref{eigenvalue_test}.
    \item \cref{eigenvalue_decay_proof}: Proof of \cref{eigenvalue_decay}.
    \item \cref{conservative_test_proof}: Proof of \cref{conservative_test}.
    \item \cref{coverage_lemma_proof}: Proof of \cref{coverage_lemma}.
    \item \cref{fuzzy_WZ_lemma1_proof}: Proof of \cref{fuzzy_WZ_lemma1}.
    \item \cref{fuzzy_identification_proof}: Proof of \cref{fuzzy_identification}.
    \item \cref{wass_deriv_decomp_proof}: Proof of \cref{wass_deriv_decomp}.
\end{description}
\end{description}

\newpage 

\section{Proofs}
\label{sec:proofs}

\subsection{Proof of \cref{effect_inequality}} \label{effect_inequality_proof} 
\begin{proof}[\textbf{Proof:}] First, observe that we may write the average treatment effect at the cutoff, $\tau$, in terms of quantile functions, i.e.,
\begin{align*}
    \tau = \int^1_0 \left(Q_1(u) - Q_0(u) \right) du
\end{align*}
where $Q_{1}(u) = \text{inf} \{ y \! : \! \text{lim}_{x \downarrow x_0} F_{Y \mid X} (y \! \mid \! x) \geq u\}$ and $Q_{0}(u) = \text{inf} \{ y \! : \! \text{lim}_{x \uparrow x_0} F_{Y \mid X} (y \! \mid \! x) \geq u\}$. Then, we immediately obtain our desired inequality by applying the Cauchy-Schwarz inequality, 
\begin{align*}
    |\tau| = \left|\int^1_0 \left(Q_1(u) - Q_0(u) \right) du \right| \leq \left(\int^1_0 \left(Q_1(u) - Q_0(u) \right)^2 du \right)^{1/2} \left( \int^1_0 1^2 du\right)^{1/2} = \Psi.
\end{align*}
Next, we show that $|\tau| = \Psi$ under an additive treatment effect $Q_1(u) = Q_0(u) + \delta$, where the quantiles of the limiting counterfactual distributions above and below the cutoff only differ by a translation in $\delta$. Immediately, this yields
\begin{align*}
    \tau = \int^1_0 \left(Q_1(u) - Q_0(u) \right) du = \int^1_0 \delta \, du = \delta 
\end{align*}
and furthermore that
\begin{align*}
    \Psi = \left(\int^1_0 \left(Q_1(u) - Q_0(u) \right)^2 du \right)^{1/2} = \left(\int^1_0 \delta^2 du \right)^{1/2} = |\delta|.
\end{align*}
This proves one direction, i.e. that under an additive treatment effect then $|\tau| = \Psi$. One easy way to prove the other direction is to consider the additive decomposition of the Wasserstein distance, $\Psi^2 = \tau^2 + \mathbb{V}\! \left(\Delta Q(U) \right)$. If $\Psi = |\tau|$, this implies that $\mathbb{V}\! \left(\Delta Q(U) \right) = 0$, and therefore that the quantile treatment effect $\Delta Q(U)$ is constant.
    
\end{proof}

\subsection{Proof of \cref{series_representation}} \label{series_representation_proof} 

\begin{proof}[\textbf{Proof:}] First, recall by \cref{identification_theorem} that the Wasserstein effect is identified as
\begin{align*}
     \Psi^2 = \int^1_0 (Q_1(u) - Q_0(u))^2 du
\end{align*}
where $Q_{1}(u) = \text{inf} \{ y \! : \! \text{lim}_{x \downarrow x_0} F_{Y \mid X} (y \! \mid \! x) \geq u\}$ and $Q_{0}(u) = \text{inf} \{ y \! : \! \text{lim}_{x \uparrow x_0} F_{Y \mid X} (y \! \mid \! x) \geq u\}$. Next, let $\{P^*_k\}^\infty_{k=0}$ be the orthogonal basis of $L^2(0, 1)$ defined by the shifted Legendre polynomials such that the $k$th shifted Legendre polynomial is defined as
\begin{align*}
    P^*_k(x) = (-1)^k \sum^k_{j=0} \binom{k}{j} \binom{k + j}{j} (-x)^j.
\end{align*}
For $a \in \{0, 1\}$ we define the $L$-moments under the limiting quantiles $Q_1(u)$ and $Q_0(u)$ to be
\begin{align*}
    \lambda^{(a)}_k = \int^1_0 Q_a(u) P^*_{k-1}(u) du.
\end{align*}
From here, under the assumption that $P_{a \mid x} \in \mathcal{P}_2(\mathbb{R})$ for $a \in \{0, 1\}$, by \cite{hosking1990} and \cite{sillitto69} it follows that
\begin{align*}
     Q_a(u) = \sum^\infty_{k=1} (2k - 1) \lambda^{(a)}_k P^*_{k-1}(u).
\end{align*}
and consequently,
\begin{align*}
    f(u) = Q_1(u) - Q_0(u) = \sum^\infty_{k=1} (2k - 1)\left(\lambda^{(1)}_k - \lambda^{(0)}_k\right) P^*_{k-1}(u).
\end{align*}
Now, let $S_K(u) = \sum^K_{k=1} (2k-1)(\lambda^{(1)}_k - \lambda^{(0)}_k) P^*_{k-1}(u)$ be a partial summation of $f$, and note that under the mean square convergence established by \cite{sillitto69},
\begin{align*}
    ||f||^2_2 = \underset{K \to \infty}{\text{lim}} ||S_K(u)||^2_2.
\end{align*}
By the orthogonality of the shifted Legendre polynomials it follows that
\begin{align*}
    ||S_K(u)||^2_2 &= \sum^K_{k=1} (2k - 1)^2 \left(\lambda^{(1)}_k - \lambda^{(0)}_k \right)^2 || P^*_{k-1} ||^2_2 \\
    &= \sum^K_{k=1} (2k - 1) \left(\lambda^{(1)}_k - \lambda^{(0)}_k \right)^2
\end{align*}
since by \cite{hosking1990} we know that $ ||P^*_k||^2_2 = (2k + 1)^{-1}$. Therefore, taking the limit as $K \to \infty$ we can see that
\begin{align*}
    \Psi^2 = \sum^\infty_{k=1} (2k - 1) \left(\lambda^{(1)}_k - \lambda^{(0)}_k \right)^2
\end{align*}
thereby completing the proof.

\end{proof}

\subsection{Proof of \cref{eigenvalue_test}} \label{eigenvalue_test_proof}

\begin{proof}[\textbf{Proof:}] 
To begin, let
\begin{align*}
    T_n =  nh \int^1_0 [\Delta \widehat{Q}(u)]^2 du
\end{align*}
be our test statistic where $\Delta \widehat{Q}(u) = \widehat{Q}_{1}(u) - \widehat{Q}_{0}(u)$ are the local polynomial estimators of the conditional quantile functions described in \cref{estimation_inference_section} and \cite{CHIANG2019589}. Furthermore, assume the regularity conditions discussed in \cref{estimation_inference_section} and \cite{CHIANG2019589} hold. Then, under the null hypothesis,  $\sqrt{nh} \Delta \widehat{Q}(u) \rightsquigarrow \mathbb{G}(u)$ in $L^2([0,1])$ where $\mathbb{G}$ is a mean-zero Gaussian process with covariance kernel
\begin{align*}
    \kappa(u,v) = \operatorname{Cov}\bigl(\mathbb{G}(u),\mathbb{G}(v)\bigr).
\end{align*}
Let $\{\lambda_k,\phi_k\}^{\infty}_{k = 1}$ denote the eigenvalue-eigenfunction pairs of the covariance operator induced by $\kappa$, such that $\{\phi_k\}^{\infty}_{k = 1}$ forms an orthonormal basis of $L^2([0,1])$. Suppose that $\sum_{k=1}^\infty \lambda_k < \infty$ and $\lambda_1 > 0$. Then we may apply the Karhunen-Lo\`eve theorem \citep{Karhunen1946ZurSS, loeve1977probability} to expand $\mathbb{G}(u)$ as
\begin{align*}
    \mathbb{G}(u) = \sum^\infty_{k=1} \sqrt{\lambda_k} Z_k \phi_k(u)
\end{align*}
where $Z_k \overset{iid}{\sim} N(0, 1)$ for all $k$. Next, by the continuous mapping theorem, it follows that
\begin{align*}
    T_n \rightsquigarrow T := \int^1_0 [\mathbb{G}(u)]^2 du =  \sum^\infty_{k=1} \lambda_k Z^2_k
\end{align*}
where the final equality holds by an application of Parseval's identity to $\{\phi_k\}^\infty_{k=1}$. From here, let $F(t)=\mathbb{P}(T\le t)$ and $c_\alpha = \inf\{t\in\mathbb{R}:F(t)\ge 1-\alpha\}$ denote the $(1-\alpha)$ quantile of $T$. Note that $\mathbb{P}(T=c_\alpha)=0$. With this machinery established, we first consider how to deal with truncation of the series $\sum^\infty_{k=1} \lambda_k Z^2_k$. Let
\begin{align*}
    T_{K} = \sum^{K}_{k=1} \lambda_k Z^2_k.
\end{align*}
Importantly, we cannot achieve a level-$\alpha$ test for a fixed value of $K$ unless it is the case that $\lambda_k = 0$ for all $k > K$. To see this, let $c_{K, \alpha}$ be the $(1 - \alpha)$ quantile under $T_K$. Then, because $T_K \leq T$ almost surely, it follows that $c_{K, \alpha} \leq c_\alpha$ and therefore,
\begin{align*}
  \mathbb{P} \! \left(T > c_{K,\alpha}\right) \geq \mathbb{P} \!\left(T > c_{\alpha}\right).
\end{align*}
Consequently, for a fixed $K$ our tests will be anti-conservative. Thus, we must let $K \to \infty$; consequently, we give $K$ an index in $n$ through the remainder of the proof. Let
\begin{align*}
    F_{K_n}(t) = \mathbb{P}(T_{K_n} \leq t)
\end{align*}
and note that since $T_{K_n} \to T$ almost surely, it follows that for each fixed $t$, $\mathbb{I}(T_{K_n} \leq t) \to \mathbb{I}(T \leq t)$. Thus, by the dominated convergence theorem, for all $t$, $F_{K_n}(t) \to F(t)$. Clearly, following standard quantile convergence arguments it then follows that $c_{K_n,\alpha}\to c_\alpha$ as $K_n \to \infty$ \citep{Vaart_1998}. \\

We now control the effect of estimating the eigenvalues. Let $\widehat{\kappa}_n$ denote an estimator of $\kappa$ such that $\widehat{\lambda}_{k,n}$ denotes the corresponding estimated eigenvalues. Notably, it is important to relate how well $\kappa$ is estimated to the number of terms we include in our truncation $T_{K_n}$. To that end, suppose that
\begin{align*}
    \left| \left| \widehat{\kappa}_n - \kappa \right| \right|_{2} = o_{\mathbb{P}}(K^{-1/2}_n )
\end{align*}
where $\left| \left| \kappa \right| \right|^2_{2} = \int^1_0 \int^1_0 \kappa(u, v)^2 du \, dv$ is the Hilbert-Schmidt norm. Then, it follows that
\begin{align} \label{eigenvalue_inequality}
    \sum^{K_n}_{k=1} |\widehat{\lambda}_{k, n} - \lambda_k| \leq \sqrt{K_n} \left( \sum^{K_n}_{k=1} (\widehat{\lambda}_{k, n} - \lambda_k)^2 \right)^{1/2} \leq \sqrt{K_n} \left| \left| \widehat{\kappa}_n - \kappa \right| \right|_{2} = o_{\mathbb{P}}(1)
\end{align}
where the first inequality follows by applying Cauchy-Schwarz and the second inequality follows by applying the Hoffman-Wielandt inequality for operators. Now, define $\widehat{T}_{K_n} = \sum^{K_n}_{k=1} \widehat{\lambda}_{k, n} Z^2_k$. Our goal here is to show that $|\widehat{T}_{K_n} - T_{K_n}|=o_{\mathbb{P}}(1)$. To do so, let $\mathcal{Z}_n = \{(X_i, A_i, Y_i) \}^n_{i=1}$ and define the event
\begin{align*}
    E_n = \left\{ \mathbb{E} \!\left[ |\widehat{T}_{K_n} - T_{K_n} | \mid \mathcal{Z}_n \right] > \delta  \right\}
\end{align*}
for some $\delta > 0$. Then, note that for any $ \varepsilon > 0$ it follows that
\begin{align*}
    \mathbb{P} \!\left( | \widehat{T}_{K_n} - T_{K_n} | > \varepsilon \right) &= \mathbb{P} \! \left( | \widehat{T}_{K_n} - T_{K_n} | > \varepsilon, E_n \right) + \mathbb{P} \!\left( | \widehat{T}_{K_n} - T_{K_n} | > \varepsilon, E^c_n \right) \\
    &\leq \mathbb{P} \!\left( \mathbb{E} \!\left[ |\widehat{T}_{K_n} - T_{K_n} | \mid \mathcal{Z}_n \right] > \delta \right) + \mathbb{P} \!\left( | \widehat{T}_{K_n} - T_{K_n} | > \varepsilon, E^c_n \right) \\
    &\overset{(i)}{\leq}  \mathbb{P} \!\left( \mathbb{E} \!\left[ |\widehat{T}_{K_n} - T_{K_n} | \mid \mathcal{Z}_n \right] > \delta \right) + \frac{\delta}{\varepsilon}
\end{align*}
where $(i)$ follows by applying a conditional Markov inequality. Then, observe by \cref{eigenvalue_inequality},
\begin{align*}
    \mathbb{E} \!\left[ |\widehat{T}_{K_n} - T_{K_n}| \mid \mathcal{Z}_n \right]
    \leq \sum_{k=1}^{K_n}|\widehat{\lambda}_{k,n}-\lambda_k| \mathbb{E}[Z_k^2]
    = o_{\mathbb{P}}(1).
\end{align*}
Thus, for each fixed $\delta$ it follows that $\mathbb{P}( \mathbb{E}[ |\widehat{T}_{K_n} - T_{K_n} | \mid \mathcal{Z}_n ] > \delta) \to 0$ and consequently,
\begin{align*}
    \underset{n}{\text{lim} \, \text{sup}} \, \mathbb{P} \!\left( | \widehat{T}_{K_n} - T_{K_n} | > \varepsilon \right) \leq \frac{\delta}{\varepsilon}.
\end{align*}
Then, by taking $\delta \to 0$ we can see that $\mathbb{P}( | \widehat{T}_{K_n} - T_{K_n} | > \varepsilon ) \to 0$. From here, define the conditional distribution function
\begin{align*}
    \widehat{F}_n(t) =\mathbb{P} \!\left(\sum_{k=1}^{K_n}\widehat{\lambda}_{k,n} Z_k^2 \leq t \mid \mathcal{Z}_n \right)
\end{align*}
let $\widehat{c}_{n,\alpha}$ be the corresponding conditional $(1-\alpha)$ quantile, and define
\begin{align*}
    p_{n, \varepsilon} = \mathbb{P} \!\left( | \widehat{T}_{K_n} - T_{K_n} | > \varepsilon \mid \mathcal{Z}_n \right).
\end{align*}
Then, on the event $\{|\widehat{T}_{K_n} - T_{K_n}| \leq \varepsilon \}$, for any $t \in \mathbb{R}$ and $\varepsilon > 0$ it follows that 
\begin{align*}
    F_{K_n}(t - \varepsilon) -  p_{n, \varepsilon} \leq \widehat{F}_{n}(t) \leq  F_{K_n}(t + \varepsilon) + p_{n, \varepsilon}
\end{align*}
and consequently,
\begin{align*}
      \underset{t \in \mathbb{R}}{\text{sup}} \, | \widehat{F}_n(t) - F_{K_n}(t)| \leq  p_{n, \varepsilon} + \underset{t \in \mathbb{R}}{\text{sup}} \big[ F_{K_n}(t + \varepsilon) - F_{K_n}(t - \varepsilon) \big].
\end{align*}
Thus, since we have already shown $F_{K_n}(t) \to F(t)$ for each $t$ (so that $\text{sup}_t |F_{K_n}(t) - F(t)| \to 0$ by P\'olya's theorem)  and since $p_{n, \varepsilon} = o_\mathbb{P}(1)$ (which follows by applying Markov's inequality) it follows that after taking $\varepsilon \to 0$,
\begin{align*}
    \underset{t \in \mathbb{R}}{\text{sup}} \, | \widehat{F}_n(t) - F_{K_n}(t)| = o_\mathbb{P}(1).
\end{align*}

Then, since $F_{K_n}$ is continuous and strictly increasing at $c_{K_n,\alpha}$ it again follows that $\widehat{c}_{n,\alpha} - c_{K_n,\alpha} = o_{\mathbb{P}}(1)$ from standard arguments for convergence of quantiles. \\

Finally, we consider the effect of approximating the critical value via Monte-Carlo simulation. The argument here is standard. Define the Monte-Carlo draws
\begin{align*}
    \widehat{T}^*_{K_n,b} = \sum_{k=1}^{K_n} \widehat{\lambda}_{k,n} Z_{k,b}^2
\end{align*}
for $b = 1, \ldots, B_n$ and let $\widehat{c}^*_{n,\alpha}$ denote the empirical $(1-\alpha)$ quantile computed from 
$\{\widehat{T}^*_{K_n,b}\}_{b=1}^{B_n}$.  Let the Monte-Carlo empirical distribution function be
\begin{align*}
    \widehat F_{n,B_n}(t) = \frac{1}{B_n}\sum_{b=1}^{B_n}\mathbb{I}(\widehat T^{*}_{K_n,b}\leq t)
\end{align*}
On the event $A_n = \{ \widehat{\lambda}_{1, n} > 0 \}$, the conditional distribution $\widehat{F}_n$ is continuous, and thus continuous at its quantiles $\widehat{c}_{n, \alpha}$. Therefore, applying the Glivenko-Cantelli theorem conditionally on $(X_1, \ldots, X_n)$, it follows that as $B_n \to \infty$,
\begin{align*}
    \underset{t \in \mathbb{R}}{\text{sup}} \, | \widehat F_{n,B_n}(t) - \widehat{F}_n(t) | \to 0
\end{align*}
and therefore, $\widehat{c}_{n,\alpha}^* - \widehat{c}_{n, \alpha} = o_{\mathbb{P}}(1)$. Finally, since $\widehat{\lambda}_{1, n} \to \lambda_1 > 0$, it follows that $\mathbb{P}(A_n) \to 1$, so this result holds unconditionally. Putting everything together, it follows that under $H_0$,
\begin{align*}
    \lim_{n\to\infty}\mathbb{P}\bigl(T_n > \widehat{c}^*_{n,\alpha}\bigr)
    =\mathbb{P}(T>c_\alpha)
    =\alpha.
\end{align*}
\end{proof}

\subsection{Proof of \cref{eigenvalue_decay}} \label{eigenvalue_decay_proof}

\begin{proof}[\textbf{Proof:}]

First, recall that we assume $K_n \asymp r^{-2/(2\beta - 1)}_n$. Thus, we assume there exist constants $0 < c_1 \leq c_2 < \infty$ and $n_0$ such that for all $n \geq n_0$,
\begin{align*}
     c_1 r^{-2/(2\beta - 1)}_n \leq K_n \leq c_2r^{-2/(2\beta - 1)}_n.
\end{align*}
Next, recall that we assume a polynomial eigenvalue decay. That is, there exist constants $C_\lambda > 0$ and $\beta > 1$ such that for all $k$, $\lambda_k \leq C_\lambda k^{-\beta}$. With this in mind, we proceed with the truncation bias. Observe that for any fixed $K_n$,
\begin{align*}
    \sum_{k > K_n} \lambda_k \leq C_\lambda \sum_{k > K_n} k^{-\beta} \overset{(i)}{\leq} C_\lambda \int^{\infty}_{K_n} x^{- \beta} dx =  \left( \frac{C_\lambda}{\beta - 1} \right) K^{1 - \beta}_n
\end{align*}
where $(i)$ follows since $f(x) = x^{-\beta}$ is a decreasing function. From here, it follows that for all $n \geq n_0$
\begin{align*}
     \left( \frac{C_\lambda}{\beta - 1} \right) K^{1 - \beta}_n \leq   \frac{C_\lambda}{\beta - 1}  \left(c_1 r^{-2/(2\beta - 1)}_n \right)^{1 - \beta} =  \frac{C_\lambda c^{1 - \beta}_1}{\beta - 1}  \left( r^{2(\beta - 1)/(2\beta - 1)}_n \right)
\end{align*}
and therefore,
\begin{align*}
    \sum_{k > K_n} \lambda_k = O \! \left( r^{\frac{2(\beta - 1)}{(2\beta - 1)}}_n\right).
\end{align*}
Next, we consider the estimation error. Recall that we assumed $|| \widehat{\kappa}_n - \kappa||_{2} = O_p(r_n)$ for some $r_n \to 0$. Thus, it is clear that $\sqrt{K_n}|| \widehat{\kappa}_n - \kappa||_{2} = O_p( \sqrt{K_n}r_n)$. Then, it follows that
\begin{align*}
    \sqrt{K_n}r_n \leq \left( c_2r^{-2/(2\beta - 1)}_n \right)^{1/2} r_n = \sqrt{c_2}  r^{-\frac{1}{2\beta - 1}}_n r_n = \sqrt{c_2 } r^{\frac{2(\beta - 1)}{(2\beta - 1)}}_n
\end{align*}
and consequently,
\begin{align*}
    \sqrt{K_n}|| \widehat{\kappa}_n - \kappa||_{2} = O_p \!\left( r^{\frac{2(\beta - 1)}{(2\beta - 1)}}_n\right)
\end{align*}
thereby completing the proof. Note: the choice $K_n \asymp r^{-2/(2\beta - 1)}_n$ can easily be seen by noting that the truncation bias scales like $K^{1-\beta}_n$. Thus, if we set
\begin{align*}
    K^{1-\beta}_n \asymp \sqrt{K_n} r_n
\end{align*}
and solve, we obtain the aforementioned rate.

\end{proof}

\subsection{Proof of \cref{conservative_test}} \label{conservative_test_proof}

\begin{proof}[\textbf{Proof:}] To begin, let $ T_n = nh \int_0^1 [\Delta \widehat{Q}(u)]^2 du$ be our test statistic. Then, recall that under the conditions described in \cref{estimation_inference_section} and \cref{eigenvalue_test_proof} it follows that
\begin{align*}
    T_n \rightsquigarrow T = \sum^{\infty}_{k=1} \lambda_k Z^2_k
\end{align*}
where $Z_k \sim N(0, 1)$ for all $k$ and $\lambda_k$ are the eigenvalues of the covariance kernel $\kappa(u, v)$. From here, observe that
\begin{align*}
    \mathbb{E}[T] = \mathbb{E} \!\left[\sum^{\infty}_{k=1} \lambda_k Z^2_k \right] = \sum^{\infty}_{k=1} \lambda_k \mathbb{E} \!\left[Z^2_k \right] = \sum^{\infty}_{k=1} \lambda_k =  \int^1_0 \kappa(u, u) du
\end{align*}
and
\begin{align*}
    \mathbb{V} \! \left(T \right) &=  \mathbb{V} \! \left(\sum^{\infty}_{k=1} \lambda_k Z^2_k \right) = \sum^{\infty}_{k=1} \lambda^2_k\mathbb{V}\! \left( Z^2_k \right) 
    = 2 \sum^\infty_{k=1} \lambda^2_k 
    = 2 \int^1_0 \!  \! \int^1_0 \kappa(u, v)^2 du \, dv.
\end{align*}
Thus, it follows that $T$ has mean $\mu := \int^1_0 \kappa(u, u) du$ and variance $\sigma^2 := 2 \int^1_0 \!  \int^1_0 \kappa(u, v)^2 du \, dv$. Consequently, it follows that $(T - \mu) / \sigma$ has mean zero and unit variance. Furthermore, by Slutsky's Theorem it follows that as $n \to \infty$
\begin{align*}
    \frac{T_n - \widehat{\mu}}{\widehat{\sigma}} \longrightarrow \frac{T - \mu}{\sigma}
\end{align*}
under the assumption that $\widehat{\mu} \overset{p}{\rightarrow} \mu$ and $\widehat{\sigma} \overset{p}{\rightarrow} \sigma$. From here, following the one-sided Chebyshev inequality (i.e. Cantelli's inequalty) discussed in \cite{luedtke_omnibus_2018}, we note that for any mean zero, unit variance random variable $X$ and $t > 0$, it follows that
\begin{align*}
    \mathbb{P}(X \geq t) \leq \frac{1}{1 + t^2}.
\end{align*}
Then, it is easy to see by the Portmanteau Theorem,
\begin{align*}
     \underset{n \to \infty}{\text{lim} \, \text{sup}} \ \mathbb{P}_{H_0}(T_n > \widehat{c}^{\text{ub}}_{n, 1- \alpha}) &\leq \mathbb{P}_{H_0} \!\left( \frac{T - \mu}{\sigma} >  \sqrt{\frac{1 - \alpha}{\alpha}}\right) \leq \frac{1}{1 + \frac{1 - \alpha}{\alpha}} = \alpha.
\end{align*}
\end{proof}

\subsection{Proof of \cref{coverage_lemma}} \label{coverage_lemma_proof}

\begin{proof}[\textbf{Proof:}]

To begin, let $\widehat{\Psi}^2_n = \int^1_0 [\Delta \widehat{Q}(u)]^2 du$. Then, recall that we define our interval as 
\begin{align*} 
    C^\prime_n = \left[\widehat{\Psi}^2_n \pm z_{1-\alpha/2}\sqrt{\widehat{s}^2_n + \frac{c^2}{nh}} \, \right]
\end{align*}
where $\widehat{s}_n$ is the estimated standard deviation of $\Psi^2$, $z_{1-\alpha/2}$ is the $1 - \alpha / 2$ quantile of a standard Normal distribution, and $c$ is some constant. Then, observe that
\begin{align*}
    \mathbb{P} \!\left( \Psi^2 \not \in C^\prime_n \right) &= \mathbb{P} \!\left( \left| \widehat{\Psi}^2_n - \Psi^2 \right| > z_{1-\alpha/2}\sqrt{\widehat{s}^2_n + \frac{c^2}{nh}} \right) \\
    &\leq \mathbb{P} \!\left( \left| \widehat{\Psi}^2_n - \Psi^2 \right| > z_{1-\alpha/2}\sqrt{ \frac{c^2}{nh}} \right) \\
    &= \mathbb{P} \!\left( \left( \widehat{\Psi}^2_n - \Psi^2 \right)^2 >  \frac{ z^2_{1-\alpha/2} c^2}{nh} \right).
\end{align*}
From here, we apply Markov's inequality to see that
\begin{align*}
    \mathbb{P} \! \left( \left( \widehat{\Psi}^2_n - \Psi^2 \right)^2 >  \frac{ z^2_{1-\alpha/2} c^2}{nh} \right) &\leq \frac{nh}{z^2_{1-\alpha/2} c^2} \mathbb{E} \!\left[( \widehat{\Psi}^2_n - \Psi^2)^2 \right] \\
    &= \frac{nh}{z^2_{1-\alpha/2} c^2} \left\{\mathbb{E} \!\left[\widehat{\Psi}^2_n - \Psi^2 \right]^2 + \mathbb{V}(\widehat{\Psi}^2_n) \right\} \\
    &= o(1)
\end{align*}
where the last equality holds under the assumption that $\mathbb{E}[\widehat{\Psi}^2_n - \Psi^2] = o((nh)^{-1/2})$ and $\mathbb{V}(\widehat{\Psi}^2_n) = o((nh)^{-1})$.
Therefore, as $n \to \infty$ it follows that $ \mathbb{P} \!\left( \Psi^2 \not \in C^\prime_n \right) \to 0$.
\end{proof}

\subsection{Proof of \cref{fuzzy_WZ_lemma1}} \label{fuzzy_WZ_lemma1_proof}

\begin{proof}[\textbf{Proof:}] The proof of \cref{fuzzy_WZ_lemma1} follows analogously to Lemma 1 of \cite{wang2025unifiedframeworkidentificationinference} with two modifications: one, the baseline treatment at the kink is random ($T_0=b(x_0,\eta)$) rather than a constant $t_0 = b(x_0)$; and two, the perturbation direction is $\delta (\omega(\eta)/\Delta_B) $ rather than the constant shift $\delta$. With that in mind, let $F_\delta(\cdot)=F_{Y_\delta\mid X=x_0}(\cdot)$ and $F_0(\cdot)=F_{Y_0\mid X=x_0}(\cdot)$. Furthermore, define $h_\delta(\cdot)= (F_\delta-F_0)/\delta$. Then, under \cref{hadamard_differentiability}, we have that as $\delta \to 0$,
\begin{align*}
    \frac{\phi(F_\delta)-\phi(F_0)}{\delta} = \phi^{\prime}_{F_0}(\Delta^F_{Id})+o(1).
\end{align*}
where $\Delta^F_{Id} = \text{lim}_{\delta \to 0} \{ h_\delta \}$. From here, let $Y_\delta = g(T_{0} + \delta (\omega(\eta)/\Delta_B),x_0,\varepsilon)$, $Y_0 =g(T_0,x_0,\varepsilon)$, and define
\begin{align*}
    Z = \left(\frac{\omega(\eta)}{\Delta_B} \right) g_1(T_0,x_0,\varepsilon).
\end{align*}
Then, we may define the remainder term $R_\delta = Y_\delta - Y_0 - \delta Z$ such that $Y_\delta = Y_0 + \delta Z + R_\delta$. 
We can now see that
\begin{align*}
    h_\delta(y) = \frac{F_\delta(y)-F_0(y)}{\delta} =
\frac{1}{\delta} \Big( \mathbb{E}\big[\mathbb{I}(Y_0+\delta Z+R_\delta\leq y) - \mathbb{I}(Y_0\leq y) \mid  X=x_0\big] \Big).
\end{align*}
Our first goal is to show that the remainder term $R_\delta$ drops out. To that end, we define
\begin{align*}
    \widetilde{h}_\delta(y) =
\frac{1}{\delta}\Big( \mathbb{E}\left[\mathbb{I}(Y_0+\delta Z\leq y) - \mathbb{I}(Y_0\leq y) \mid  X=x_0\right]\Big).
\end{align*}
Then, observe that
\begin{align*}
    \left|h_\delta(y) - \widetilde{h}_\delta(y)\right| &\leq \frac{1}{|\delta|} \Big( \mathbb{E}\big[ \left|\mathbb{I}(Y_0+\delta Z+R_\delta\leq y) - \mathbb{I}(Y_0+\delta Z\leq y) \right| \mid  X=x_0\big] \Big).
\end{align*}
From here, define the events $U = \{Y_0+\delta Z+R_\delta\leq y \} $ and $V = \{Y_0+\delta Z\leq y \}$ and note that
\begin{align*}
    |\mathbb{I}(U) - \mathbb{I}(V)| = \mathbb{I}(U \triangle V)
\end{align*}
where $U \triangle V$ denotes the symmetric difference. Our goal now is to show the set inclusion $U \triangle V \subseteq \{ | y - (Y_0 + \delta Z)| \leq |R_\delta| \}$ holds. First, consider the case where $R_\delta \geq 0$. Then, it is clear that $U \subseteq V$. Furthermore, the set difference $V \setminus U$ occurs when
\begin{align*}
    \{ Y_0 + \delta Z \leq y \} \cap \{Y_0 + \delta Z + R_\delta > y \} \quad \iff \quad \{y - R_\delta < Y_0 + \delta Z \leq y \},
\end{align*}
and so it follows that $|y - (Y_0 + \delta Z)| \leq R_\delta = |R_\delta|$. Next, suppose that $R_\delta < 0$. Now we have that $V \subseteq U$ and the set difference $U \setminus V$ occurs when
\begin{align*}
    \{Y_0 + \delta Z > y \} \cap \{ Y_0 + \delta Z + R_\delta \leq y \} \quad \iff \quad \{y < Y_0 + \delta Z \leq y - R_\delta\}.
\end{align*}
Thus, $|y - (Y_0 + \delta Z)| \leq - R_\delta = |R_\delta|$. Putting both cases together, it follows that
\begin{align*}
    \left|h_\delta(y) - \widetilde{h}_\delta(y)\right| &\leq \frac{1}{|\delta|} \Big( \mathbb{E} \! \left[ \mathbb{I} \left( \big|y-(Y_0+\delta Z)\big|\leq |R_\delta| \right) \mid  X=x_0\right] \Big) \\
    &\leq \frac{1}{|\delta|} \Big( \mathbb{P}(|R_\delta|\geq c|\delta| \mid X=x_0) + \mathbb{P}( |Y_0+\delta Z-y |\leq c|\delta| \mid X=x_0) \Big).
\end{align*}
where the second inequality follows after fixing some $c > 0$ and splitting on the event that $|R_\delta| \geq c |\delta|$ or $|R_\delta| < c |\delta|$. From here, it follows by \cref{regularity_conditions}$(i)$ that as $\delta \to 0$, then
\begin{align*}
    \frac{1}{|\delta|} \mathbb{P}(|R_\delta|\geq c|\delta| \mid X=x_0) = o(1).
\end{align*}
In the case of the second term, we now leverage \cref{regularity_conditions}$(ii)$ to see that
\begin{align*}
    \mathbb{P} \! \left(\big|Y_0+\delta Z-y\big|\le c|\delta| \mid X=x_0\right) = \int \! \! \int_{y-\delta z-c|\delta|}^{y-\delta z+c|\delta|} f_{Y_0,Z\mid X}( a,z\mid x_0) da\,dz \leq 2c|\delta|\int |\varpi(z)| dz,
\end{align*}
and so, consequently, it follows that 
\begin{align*}
    \left|h_\delta(y) - \widetilde{h}_\delta(y)\right| = o(1) + O(c).
\end{align*}
Then, since the choice of $c > 0$ was arbitrary, we can see that $h_\delta(y) - \widetilde{h}_\delta(y) \to 0$. Next, we evaluate the limit of $\widetilde{h}_\delta(y)$ as $\delta \to 0$ using the joint density of $(Y_0,Z)$. Recall that
\begin{align*}
    \widetilde h_\delta(y) = \frac{1}{\delta} \Big( \mathbb{P}(Y_0+\delta Z\leq y\mid X=x_0)- \mathbb{P}(Y_0\leq y\mid X=x_0) \Big).
\end{align*}
Thus, using the identity
\begin{align*}
    \mathbb{I}(U\le v)-\mathbb{I}(U\le w) = \mathbb{I}(w< U\leq v)-\mathbb{I}(v < U\leq w)
\end{align*}
with $U=Y_0$, $v=y-\delta Z$, and $w=y$, we obtain the more convenient expression
\begin{align} \label{h_delta_eq}
    \widetilde h_\delta(y) = \frac{1}{\delta}\Big( \mathbb{P}(y<Y_0\leq y-\delta Z\mid X=x_0) - \mathbb{P}(y-\delta Z<Y_0\leq y\mid X=x_0) \Big).
\end{align}
Suppose $\delta>0$ and let $f(a,z) =f_{Y_0,Z\mid X}(a,z\mid x_0)$. Since $y-\delta Z>y$ requires $Z<0$ and
$y-\delta Z<y$ requires $Z>0$, we can write
\begin{align*}
\mathbb{P} \! \left( y< Y_0\leq y-\delta Z\mid X=x_0 \right)
&=\int_{-\infty}^{0}\int_{y}^{y-\delta z} f(a,z)da\,dz,\\
\mathbb{P} \! \left( y-\delta Z<Y_0\leq y\mid X=x_0 \right)
&=\int_{0}^{\infty}\int_{y-\delta z}^{y} f(a,z)da\,dz.
\end{align*}
Then, applying the change of variables $u=(y-a)/\delta$ (such that $a=y-\delta u$ and $da=-\delta\,du$) it follows that
\begin{align*}
    \frac{1}{\delta}\mathbb{P} \! \left( y<Y_0\leq y-\delta Z\mid X=x_0 \right) &= \int_{-\infty}^{0}\int_{z}^{0} f(y-\delta u,z) du\,dz, \\
\frac{1}{\delta}\mathbb{P} \! \left( y-\delta Z<Y_0\leq y\mid X=x_0 \right) &= \int_{0}^{\infty}\int_{0}^{z} f(y-\delta u,z) du\,dz.
\end{align*}
Substituting both into \cref{h_delta_eq} yields
\begin{align*}
    \widetilde h_\delta(y) = \int_{-\infty}^{0}\int_{z}^{0} f(y-\delta u,z)du\,dz - \int_{0}^{\infty}\int_{0}^{z} f(y-\delta u,z)du\,dz.
\end{align*}
From here, continuity of $f(\cdot,z)$ in its first argument and the domination condition specified in \cref{regularity_conditions}$(ii)$ implies that
\begin{align*}
    \lim_{\delta \to 0} \left\{ \widetilde h_\delta(y) \right\} = \int_{-\infty}^{0}\int_{z}^{0} f(y,z) du\,dz - \int_{0}^{\infty}\int_{0}^{z} f(y,z) du\,dz = \int (-z) f(y,z) dz.
\end{align*}
Repeating analogous calculations in the case where $\delta < 0$ yields the same limit. Therefore, the two-sided derivative exists and we can say that
\begin{align*}
    \Delta^F_{Id}(y) &= \left.\frac{\partial}{\partial \delta}F_\delta(y)\right|_{\delta=0} \\
    &= \int(-z) f_{Y_0,Z\mid X}(y,z\mid x_0)dz \\
    &= \int(-z) f_{Y\mid X}(y\mid x_0)\,f_{Z\mid Y,X}(z\mid y,x_0) dz \\
    &=  -f_{Y\mid X}(y\mid x_0) \mathbb{E}\! \left[ \left(\frac{\omega(\eta)}{\Delta_B} \right) g_1(T_0,x_0,\varepsilon) \mid Y= y, X=x_0\right]
\end{align*}
which is the desired equation, and thus completes the proof.
\end{proof}

\subsection{Proof of \cref{fuzzy_identification}} \label{fuzzy_identification_proof}

\begin{proof}[\textbf{Proof:}] First, recall by \cref{fuzzy_WZ_lemma1} that the fuzzy local treatment effect at the kink admits the representation $\Delta_\phi^F = \phi'_{F_{Y\mid X=x_0}}(\Delta^F_{Id}(\cdot))$
where
\begin{align*}
    \Delta^F_{Id}(y) = - f_{Y\mid X}(y\mid x_0) \mathbb{E} \! \left[\frac{\omega(\eta)}{\Delta_B}\,g_1(T_0,x_0,\varepsilon) \mid Y=y,X=x_0\right]
\end{align*}
such that $T_0 = b(x_0,\eta)$, $\omega(\eta) =b^\prime(x_0^+,\eta)-b^\prime(x_0^-,\eta)$, and $\Delta_B = \mu^\prime_B(x^+_0) - \mu^\prime_B(x^-_0) \neq 0$. From here, our goal is to show that for all $y$, $\Delta^F_{Id}(y)=\mathrm{FDRKD}(y)$ where 
\begin{align*}
    \mathrm{FDRKD}(y) = \frac{\frac{\partial}{\partial x} F_{Y \mid X}(y \! \mid \! x_0^+) - \frac{\partial}{\partial x} F_{Y \mid X}(y \! \mid \! x_0^-)}
{\mu^\prime_B(x_0^+) - \mu^\prime_B(x_0^-)}
\end{align*}
and $\mu_B(x) = \mathbb{E}[b(x,\eta)\mid X=x]$. To that end, following \cite{wang2025unifiedframeworkidentificationinference} define $h(x,e,u) =g(b(x,u),x,e)$ so that  $Y=h(X,\varepsilon,\eta)$. Then, by \cref{smoothness_unobservables} it follows that we may write the conditional cumulative distribution function as
\begin{align*}
    F_{Y\mid X}(y \! \mid \! x) = \iint \mathbb{I}(h(x,e,u)\le y) f_{\varepsilon,\eta\mid X}(e,u \! \mid \!  x)de\,du.
\end{align*}
Next, let $y$ be fixed and consider the decomposition
\begin{align*}
    \frac{F_{Y\mid X}(y \! \mid \!  x_0+t)-F_{Y\mid X}(y \! \mid \!  x_0)}{t} = A_{1,t}(y)+A_{2,t}(y),
\end{align*}
where
\begin{align*}
    A_{1,t}(y) &= \frac{1}{t}\iint \Big(\mathbb{I}(h(x_0+t,e,u)\le y)-\mathbb{I}(h(x_0,e,u)\le y)\Big) f_{\varepsilon,\eta\mid X}(e,u \! \mid \!  x_0)de\,du, \\
A_{2,t}(y) &= \frac{1}{t}\iint \mathbb{I}(h(x_0+t,e,u)\le y)\Big(f_{\varepsilon,\eta\mid X}(e,u \! \mid \!  x_0+t)-f_{\varepsilon,\eta\mid X}(e,u \! \mid \!  x_0)\Big)de\,du.
\end{align*}
Note that $A_{1, t}(y)$ is a structural term that holds $f_{\varepsilon,\eta\mid X}(\cdot \! \mid \!  x)$ fixed at $x_0$, and $A_{2, t}(y)$ is a selection  term that captures changes in $f_{\varepsilon,\eta\mid X}$ with $x$. We proceed with the latter term. Note that by  \cref{smoothness_unobservables} it follows that
\begin{align*}
   \underset{t \to 0}{\text{lim}} \  A_{2,t}(y) &= \iint \underset{t \to 0}{\text{lim}} \left\{ \mathbb{I}(h(x_0+t,e,u)\leq y)\left(\frac{f_{\varepsilon,\eta\mid X}(e,u \! \mid \!  x_0+t)-f_{\varepsilon,\eta\mid X}(e,u \! \mid \!  x_0)}{t} \right) \right\} de\,du  \\
   &= \underbrace{\iint \mathbb{I}(h(x_0,e,u)\leq y)\left( \frac{\partial}{\partial x} f_{\varepsilon,\eta\mid X}(e, u \!  \mid \!  x_0) \right)de\,du}_{:= S(y)}.
\end{align*}
Note that this limit is the same when considering both $t \uparrow 0$ and $t \downarrow 0$. Next, we consider $A_{1, t}(y)$. Define the one-sided derivatives
\begin{align*}
    H^+ &= \frac{\partial}{\partial x} h(x_0^+,\varepsilon,\eta) = b^\prime(x_0^+,\eta)g_1(T_0,x_0,\varepsilon)+g_2(T_0,x_0,\varepsilon) \quad \text{and} \\
    H^- &= \frac{\partial}{\partial x} h(x_0^-,\varepsilon,\eta) = b^\prime(x_0^-,\eta)g_1(T_0,x_0,\varepsilon)+g_2(T_0,x_0,\varepsilon) .
\end{align*}
Then we can write $h(x_0+t,e,u)$ in terms of the one-sided limits $h(x_0+t,e,u) = Y_0 + tH^+ + R_t^+$ where $R_t^+ = Y_t - Y_0 - tH^+$, $Y_t = h(x_0 + t, e, u)$, and analogous definitions are given for $h(x_0+t,e,u) = Y_0 + tH^- + R_t^-$. Importantly, note that the limit as $t \to 0$ of $A_{1, t}(y)$ is identical in form to the limit computations done in the proof of \cref{fuzzy_WZ_lemma1}, where $tH^+$ plays the role of $\delta Z$ (similarly, \cref{theorem_regularity_conditions}$(i)$ plays an analogous role to \cref{regularity_conditions}$(i)$ and \cref{theorem_regularity_conditions}$(ii)$ to \cref{regularity_conditions}$(ii)$). This can easily be seen by plugging in our decompositions of $h(x_0+t,e,u)$; observe that
\begin{align*}
    A_{1,t}(y) &= \frac{1}{t} \Big[ \mathbb{P} \! \left(h(x_0+t,\varepsilon,\eta)\leq y\mid X=x_0\right) - \mathbb{P} \! \left(h(x_0,\varepsilon,\eta)\leq y\mid X=x_0\right) \Big] \\
    &= \frac{1}{t} \Big[ \mathbb{P} \! \left(Y_0\leq y-tH^+ - R_t^+\mid X=x_0\right)- \mathbb{P} \! \left(Y_0\leq y\mid X=x_0 \right) \Big] \\
    &= \frac{1}{t} \Big[ \mathbb{P} \! \left(y<Y_0\leq y-tH^+-R_t^+\mid X=x_0\right) - \mathbb{P} \! \left(y-tH^+-R_t^+<Y_0\leq y\mid X=x_0\right) \Big].
\end{align*}

Thus, repeating the same steps as in the proof of \cref{fuzzy_WZ_lemma1}, it follows that
\begin{align*}
    \lim_{t\downarrow 0} \{ A_{1,t}(y) \} &= -f_{Y\mid X}(y \! \mid \!  x_0) \mathbb{E} \! \left[H^+\mid Y=y,X=x_0\right] \quad \text{and} \\
    \lim_{t\uparrow 0} \{ A_{1,t}(y) \} &= -f_{Y\mid X}(y \! \mid \!  x_0) \mathbb{E} \! \left[H^-\mid Y=y,X=x_0\right].
\end{align*}
Combining the limits for $A_{1,t}$ and $A_{2,t}$, we have the one-sided derivative formulas
\begin{align*}
\frac{\partial}{\partial x}F_{Y\mid X}(y \! \mid \!  x_0^+) &= -f_{Y\mid X}(y \! \mid \!  x_0) \mathbb{E}\! \left[H^+\mid Y=y,X=x_0\right] +S(y),\\
\frac{\partial}{\partial x}F_{Y\mid X}(y \! \mid \!  x_0^-) &= -f_{Y\mid X}(y \! \mid \!  x_0) \mathbb{E} \! \left[H^-\mid Y=y,X=x_0\right] +S(y).
\end{align*}
Thus, taking the difference, it follows that
\begin{align*}
\frac{\partial}{\partial x}F_{Y\mid X}(y \! \mid \!  x_0^+) -\frac{\partial}{\partial x}F_{Y\mid X}(y \! \mid \!  x_0^-) &= -f_{Y\mid X}(y \! \mid \!  x_0) \mathbb{E} \! \left[ H^+-H^-\mid Y=y,X=x_0 \right],
\end{align*}
and furthermore, after plugging in the definitions of $H^+$ and $H^-$,
\begin{align*}
    H^+-H^- = \big(b^\prime(x_0^+,\eta)- b^\prime(x_0^-,\eta)\big)g_1(T_0,x_0,\varepsilon) = \omega(\eta) g_1(T_0,x_0,\varepsilon),
\end{align*}
because the $g_2(T_0,x_0,\varepsilon)$ term cancels. Hence, for every $y$,
\begin{align*}
\frac{\partial}{\partial x}F_{Y\mid X}(y \! \mid \!  x_0^+)
-\frac{\partial}{\partial x}F_{Y\mid X}(y \! \mid \!  x_0^-)
=
-f_{Y\mid X}(y \! \mid \!  x_0)\mathbb{E}\! \left[\omega(\eta)g_1(T_0,x_0,\varepsilon)\mid Y=y,X=x_0\right].
\end{align*}
Finally, we must consider the denominator. Here, we can again apply the same decomposition argument made before. Let \begin{align*}
    \mu_B(x) = \int b(x,u) f_{\eta\mid X}(u \! \mid \!  x)du.
\end{align*}
Then, fix $t>0$ such that $x_0+t\in I_{x_0} \setminus \{x_0\}$. Starting from 
\begin{align*}
    \frac{\mu_B(x_0+t)-\mu_B(x_0)}{t}
    &= \frac{1}{t}\int b(x_0+t,u) f_{\eta\mid X}(u \! \mid \!  x_0+t) du
      - \frac{1}{t}\int b(x_0,u) f_{\eta\mid X}(u \! \mid \! x_0)du
\end{align*}
we can add and subtract $\int b(x_0+t,u) f_{\eta\mid X}(u \! \mid \!  x_0)du$ to obtain the decomposition:
\begin{align*}
    \frac{\mu_B(x_0+t)-\mu_B(x_0)}{t}
    &= \Bigg\{ \frac{1}{t}\int \Big(b(x_0+t,u)-b(x_0,u)\Big) f_{\eta\mid X}(u \! \mid \!  x_0)du \ + \\
    &\phantom{{}={\Bigg\{ }}\frac{1}{t}\int b(x_0+t,u)\Big(f_{\eta\mid X}(u \! \mid \!  x_0+t)-f_{\eta\mid X}(u \! \mid \!  x_0)\Big) du \Bigg\},
\end{align*}
which we term $S_{1, t}$ and $S_{2, t}$, respectively. Thus, following the same arguments as before, under \cref{theorem_regularity_conditions}$(iii)$ and \cref{theorem_regularity_conditions}$(iv)$, it can be shown that the right and left derivatives are given by
\begin{align*}
    \mu^\prime_B(x_0^+) &= \mathbb{E}\left[b^\prime(x_0^+,\eta)\mid X=x_0\right] + \int b(x_0,u) \left(\frac{\partial}{\partial x} f_{\eta\mid X}(u \! \mid \!  x_0) \right) du
\end{align*}
and
\begin{align*}
    \mu^\prime_B(x_0^-) &= \mathbb{E}\left[b^\prime(x_0^-,\eta)\mid X=x_0\right] + \int b(x_0,u) \left(\frac{\partial}{\partial x} f_{\eta\mid X}(u \! \mid \!  x_0) \right) du.
\end{align*}
Therefore, taking the difference yields
\begin{align*}
    \mu_B^\prime(x_0^+)-\mu_B^\prime(x_0^-) &=
\mathbb{E}\big[b^\prime(x_0^+,\eta)-b^\prime(x_0^-,\eta)\mid X=x_0\big] \\
&=\mathbb{E}[\omega(\eta)\mid X=x_0] \\
&=\Delta_B.
\end{align*}
Putting everything together, for all $y$,
\begin{align*}
\mathrm{FDRKD}(y) &= \frac{\frac{\partial}{\partial x}F_{Y\mid X}(y \! \mid \!  x_0^+)-\frac{\partial}{\partial x}F_{Y\mid X}(y \! \mid \!  x_0^-)}
{\mu_B^\prime(x_0^+)-\mu_B^\prime(x_0^-)}\\
&=
\frac{-f_{Y\mid X}(y \! \mid \!  x_0)
\mathbb{E}\big[\omega(\eta)g_1(T_0,x_0,\varepsilon)\mid Y=y,X=x_0\big]}{\Delta_B}\\
&=
-f_{Y\mid X}(y\mid x_0) \mathbb{E} \! \left[\frac{\omega(\eta)}{\Delta_B} g_1(T_0,x_0,\varepsilon) \mid Y=y,X=x_0\right]
=\Delta^F_{Id}(y).
\end{align*}
Therefore,
\begin{align*}
   \Delta_\phi^{F}
=\phi^\prime_{F_{Y\mid X=x_0}}\big(\Delta^F_{Id}(\cdot)\big)
=\phi^\prime_{F_{Y\mid X=x_0}}\big(\mathrm{FDRKD}(\cdot)\big), 
\end{align*}
which proves the theorem.
\end{proof}

\subsection{Proof of \cref{wass_deriv_decomp}} \label{wass_deriv_decomp_proof}

\begin{proof}[\textbf{Proof:}]
    To begin, let $\langle f, g \rangle = \int^1_0 f(u) g(u) du$ denote the inner product on $L^2(0, 1)$. Next, recall that a complete orthogonal basis $\{\phi_k\}^{\infty}_{k=1}$ in $L^2(0, 1)$ admits the generalized Fourier expansion
    \begin{align*}
        f(u) = \sum^\infty_{k=1} a_k \phi_k(u)
    \end{align*}
    for every $f \in L^2(0, 1)$, where the coefficients are given by
    \begin{align*}
        a_k = \frac{\langle f, \phi_k \rangle}{|| \phi_k ||^2_2}.
    \end{align*}
    Next, let $\{P^*_k\}^{\infty}_{k=0}$ be the shifted Legendre polynomials, such that $P^*_k(u) = P_k(2u - 1)$ where $P_k$ is the $k$th Legendre polynomial and $P^*_0 = 1$. Note that
    \begin{align*}
        \int_0^1 P_j^*(u)P_k^*(u)du = \begin{cases}
0, & j\neq k,\\
\frac{1}{2k+1}, & j=k.
\end{cases}.
    \end{align*}
    Now, recall that we defined
    \begin{align*}
         \Delta Q^\prime(u) =  \frac{\frac{\partial}{\partial x}Q_{Y\mid X}(u \! \mid \! x_0^+)-\frac{\partial}{\partial x}Q_{Y\mid X}(u \! \mid \! x_0^-)}
{\mu^\prime_B(x_0^+)-\mu^\prime_B(x_0^-)}.
    \end{align*}
    Then, under the assumption that $\int^1_0 [\Delta Q^\prime(u)]^2 du < \infty$ it follows that $\Delta Q^\prime \in L^2(0, 1)$. Therefore, we may apply the generalized Fourier expansion with $f = \Delta Q^\prime$ to find
    \begin{align*}
        \langle \Delta Q^\prime, P^*_{k-1} \rangle &= \int^1_0 \Delta Q^\prime(u) P^*_{k-1}(u) du \\
        &= \frac{1}{\Delta_B }\int^1_0 \left( \frac{\partial}{\partial x}Q_{Y\mid X}(u \! \mid \! x_0^+)-\frac{\partial}{\partial x}Q_{Y\mid X}(u \! \mid \! x_0^-) \right)  P^*_{k-1}(u) du \\
        &= \frac{1}{\Delta_B}\left(\lambda^\prime_k(x_0^+) - \lambda^\prime_k(x_0^-)\right).
    \end{align*}
    Thus, since $||P^*_{k-1}||^2_2 = (2(k-1) + 1)^{-1} = (2k-1)^{-1}$, it follows that
    \begin{align*}
        \Delta Q^\prime(u) = \sum^\infty_{k=1} (2k - 1) \left( \frac{\lambda^\prime_k(x_0^+) - \lambda^\prime_k(x_0^-)}{\Delta_B}\right) P^*_{k-1}(u).
    \end{align*}
    Finally, applying Parseval's identity for complete orthogonal expansions \citep{Conway1990CourseFunctionalAnalysis} yields
    \begin{align*}
        ||\Delta Q^\prime||^2_2 &= \sum^\infty_{k=1} a^2_k ||P^*_{k-1}||^2_2 \\
        &= \sum^\infty_{k=1} \left[(2k-1)^2 \left(\frac{\lambda^\prime_k(x_0^+) - \lambda^\prime_k(x_0^-)}{\Delta_B} \right)^2 \right] \cdot \frac{1}{2k-1} \\
        &= \sum_{k=1}^\infty (2k-1) \left( \frac{\lambda^\prime_k(x_0^+)-\lambda^\prime_k(x_0^-)}{\mu^\prime_B(x_0^+)-\mu^\prime_B(x_0^-)} \right)^2,
    \end{align*}
    which completes the proof, since by definition $||\Delta Q^\prime||^2_2 = \int^1_0 (\Delta Q^\prime(u))^2 du = (\Psi^\prime_{\mathcal{C}})^2$.
\end{proof}

\end{document}